\documentclass[a4paper,fleqn]{cas-sc}

\usepackage{amssymb}
\usepackage{amsmath}
\usepackage{amsthm}

\usepackage{wrapfig}
\usepackage{overpic}
\usepackage{tabularx} 
\usepackage[all]{xy}
\usepackage{contour} 
\contourlength{0.5pt}

\usepackage{booktabs} % Not included in the journal style

\usepackage{xcolor}

\definecolor{lavender}{RGB}{180, 160, 210}

\usepackage{etoolbox}
\newbool{anonymous}

\usepackage[dvipsnames]{xcolor} % For editing convenient
\usepackage{float}

\newcommand{\notforarxiv}[1]{}
\newif\ifshowedit
\showeditfalse

\ifshowedit
  \newcommand{\edit}[1]{\begingroup\color{blue}#1\endgroup}
  
\else
  \newcommand{\edit}[1]{#1}
  
\fi

\newcommand{\FP}{\mathcal{F}_\mathcal{P}}   % faces    -> oriented planes
\newcommand{\CE}{\mathcal{E}_\mathcal{C}}   % edges    -> oriented cones
\newcommand{\SV}{\mathcal{V}_\mathcal{S}}   % vertices -> oriented spheres

\newcommand{\Lf}{\mathcal{L}}

\newcommand{\OS}{\mathrm{OrSpheres}(\R^3)}
\newcommand{\OP}{\mathrm{OrPlanes}(\R^3)}
\newcommand{\OC}{\mathrm{OrCones}(\R^3)}

\newtheorem{defn}{Definition}[section]
\newtheorem{thm}[defn]{Theorem}

\newtheorem{cor}[defn]{Corollary}
\newtheorem{prop}[defn]{Proposition}

\newtheorem{remark}[defn]{Remark}
\def\be{\begin{equation}}
\def\ee{\end{equation}}

\newcommand{\R}{\mathbb{R}}

\def\la{\langle}
\def\ra{\rangle}

\makeatletter
\newsavebox{\@brx}
\newcommand{\lla}[1][]{\savebox{\@brx}{$\m@th{#1\langle}$}%
  \mathopen{\copy\@brx\kern-0.5\wd\@brx\usebox{\@brx}}}
\newcommand{\rra}[1][]{\savebox{\@brx}{$\m@th{#1\rangle}$}%
  \mathclose{\copy\@brx\kern-0.5\wd\@brx\usebox{\@brx}}}
\makeatother

\usepackage[authoryear,longnamesfirst]{natbib}

\def\tsc#1{\csdef{#1}{\textsc{\lowercase{#1}}\xspace}}
\tsc{WGM}
\tsc{QE}
\tsc{EP}
\tsc{PMS}
\tsc{BEC}
\tsc{DE}
\begin{document}
\let\WriteBookmarks\relax
\def\floatpagepagefraction{1}
\def\textpagefraction{.001}
\shorttitle{Approximation by Quad Meshes in Laguerre Geometry}

%\begin{frontmatter}

\title [mode = title]{Approximation by Quad Meshes in Laguerre Geometry}                      
% \tnotemark[1,2]

% \tnotetext[1]{This document is the results of the research
%    project funded by the National Science Foundation.}

% \tnotetext[2]{The second title footnote which is a longer text matter
%    to fill through the whole text width and overflow into
%    another line in the footnotes area of the first page.}

\setbool{anonymous}{false}  % set to false to reveal authors

%--- Function to revela authors or keep them anonymouse
\ifbool{anonymous}{
    \author[1]{Anonymous Authors}
    \affiliation[1]{Placeholder}
    \affiliation[2]{Placeholder}
}{
    \shortauthors{A. Ramos-Cisneros, M. Skopenkov, H. Pottmann}
    \author[1]{Anthony Ramos-Cisneros}[orcid=0000-0001-5430-6193]
    \cormark[1]
    \ead{anthony.cisneros@kaust.edu.sa}
    \credit{Writing - Original Draft, Investigation, Conceptualization, Methodology, Software, Visualization, Writing - Review \& Editing}
    \affiliation[1]{organization={King Abdullah University of Science and Technology},
                    city={Thuwal},
                    postcode={23955-6900}, 
                    state={Makkah},
                    country={Saudi Arabia}}
    \author[1]{Mikhail Skopenkov}[orcid=0000-0003-2453-0009]
    \ead{mikhail.skopenkov@gmail.com}
    \credit{Writing - Review \& Editing, Formal Analysis}
    \author[2]{Helmut Pottmann}[orcid=0000-0002-3195-9316]
    \ead{pottmann@geometrie.tuwien.ac.at}
    \credit{Writing - Original draft preparation, Investigation, Formal Analysis, Conceptualization, Writing - Review \& Editing, Supervision}
    \affiliation[2]{organization={Institute of Discrete Mathematics and Geometry, Vienna University of Technology}, 
                    city={Vienna},
                    postcode={1040},
                    country={Austria}}
}

% \author[1]{Anonymous Authors}
% \affiliation[1]{Placeholder}
% \affiliation[2]{Placeholder}

% \shortauthors{A. Ramos-Cisneros, M. Skopenkov, H. Pottmann}
% \author[1]{Anthony Ramos-Cisneros}[
%                         orcid=0000-0001-5430-6193]
% \cormark[1]

% \ead{anthony.cisneros@kaust.edu.sa}

% \credit{Writing - Original Draft, Investigation, Conceptualization, Methodology, Software, Visualization, Writing - Review \& Editing }

% \affiliation[1]{organization={King Abdullah University of Science and Technology},
%                 city={Thuwal},
% %               citysep={}, % Uncomment if no comma needed between city and postcode
%                 postcode={23955-6900}, 
%                 state={Makkah},
%                 country={Saudi Arabia}}

% \author[1]{Mikhail Skopenkov}[orcid=0000-0003-2453-0009]

% \ead{mikhail.skopenkov@gmail.com}

% \credit{Writing - Review \& Editing, Formal Analysis}

% \author[2]{Helmut Pottmann}[orcid=0000-0002-3195-9316]
% % \fnmark[2]
% \ead{pottmann@geometrie.tuwien.ac.at}
% % \ead[URL]{https://www.university.org}

% \credit{Writing - Original draft preparation, Investigation, Formal Analysis, Conceptualization, Writing - Review \& Editing, Supervision}

% \affiliation[2]{organization={Institute of Discrete Mathematics and Geometry, Vienna University of Technology}, 
%                 city={Vienna},
%                 postcode={1040},
%                 country={Austria}}

\cortext[cor1]{Corresponding author}
%\cortext[cor2]{Principal corresponding author}
% \fntext[fn1]{This is the first author footnote, but is common to third
%   author as well.}
% \fntext[fn2]{Another author footnote, this is a very long footnote and
%   it should be a really long footnote. But this footnote is not yet
%   sufficiently long enough to make two lines of footnote text.}

% \nonumnote{This note has no numbers. In this work we demonstrate $a_b$
%   the formation Y\_1 of a new type of polariton on the interface
%   between a cuprous oxide slab and a polystyrene micro-sphere placed
%   on the slab.
% %   }

\begin{abstract}
We study analogs of planar-quadrilateral meshes in Laguerre sphere geometry
and the approximation of smooth surfaces by them. These new \emph{Laguerre meshes} can be viewed as watertight surfaces formed by planar quadrilaterals 
%(corresponding to the vertices of a mesh)
, strips of right circular cones (representing the edges), and spherical panels. In the smooth limit, we get an analog of conjugate nets in Laguerre geometry, which we call \emph{Laguerre conjugate nets} with respect to an attached sphere congruence. We introduce the notion of \emph{Laguerre conjugate directions}, provide a method for computing them, and apply them to approximate surfaces by L-meshes with prescribed radii of spherical panels. 
\end{abstract}

% \begin{graphicalabstract}
% \includegraphics[width=1\linewidth]{res4.png}
% \end{graphicalabstract}

\notforarxiv{
\begin{highlights}
\item Analogs of planar-quadrilateral meshes in Laguerre sphere geometry are introduced and studied.
\item Laguerre conjugate nets are defined and studied.
\item A computational framework for approximating smooth surfaces with L-nets is presented.
\end{highlights}
}

% \begin{highlights}
% \item Analogs of planar-quadrilateral meshes in Laguerre sphere geometry are introduced and studied.
% \item Laguerre conjugate nets are defined and studied.
% \item A computational framework for approximating smooth surfaces with L-nets is presented.
% \end{highlights}

\maketitle
\section{Introduction} \label{sec:intro}
%%%%%%%%%%%%%%%%%%%%%%%%%%%%%%%%%%%%%%%%%%%%%%%%%%%%%%%%%%%%%%%%%%%%%%%%%%%%

The availability of spherical glass panels for cladding architectural skins has motivated recent research on meshes with spherical panels. %\edit{Similarly, the availability of rotational cylindrical and conical panels, preferred in architecture due to their reduced manufacturing cost, has motivated research on paneling freeform surfaces with developable faces~\citep{gavriil2019optimizing}.} 
This work aims at an extension of basic properties and results from discrete differential geometry on 
meshes with planar faces \citep{bobenko-2009-ddg} towards
their counterparts in sphere geometries. Initially,
the focus has been on meshes in M\"obius geometry \citep{spheremesh-2023,sphereapprox-2024}, which are formed
by spherical faces and circular edges, and on their application
to computational design and approximation problems. 

\begin{wrapfigure}[8]{l}{0.25\linewidth}
\vspace{-0.3cm}
\hspace{0.3cm}
\includegraphics[width=0.9\linewidth]{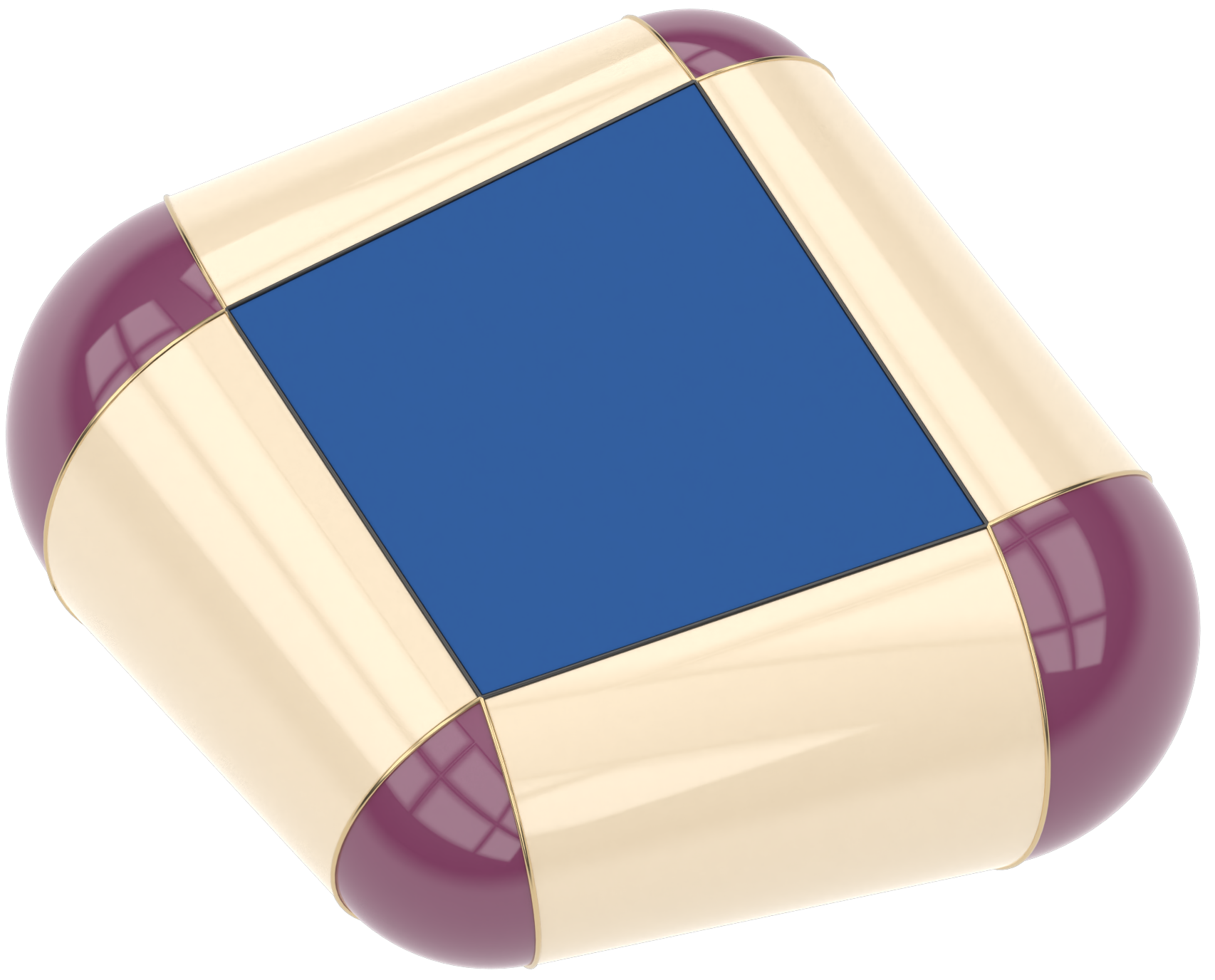}
\end{wrapfigure}

A study of their
counterparts in Laguerre sphere geometry  is not 
straightforward and a topic of our ongoing research. %\citep{Anthony-Laguerre1}.
One can see these Laguerre meshes, shortly called \emph{L-meshes}, as watertight surfaces
that are formed by planar polygons, 
%(corresponding to the faces of a mesh)
strips of right circular cones (representing
the edges), and spherical %quadrilateral 
panels (see \edit{see inset figure,} Figs.~\ref{fig:QL} and~\ref{fig:roof}). 
\edit{%Rotational cylindrical and 
Here, conical panels are also useful in architecture due to their reduced manufacturing cost, a particular case of paneling freeform surfaces with developable pieces~\citep{gavriil2019optimizing}.}

If an L-mesh closely approximates an underlying $C^2$ reference
surface $f$, its appearance depends on the
sign of the Gaussian curvature of $f$. In areas of positive curvature of $f$, 
the L-mesh is \edit{typically} a  $C^1$ surface, in contrast to their
M\"obius geometric counterparts, which are just $C^0$ surfaces. 
In negatively curved areas, the surface
represented by an L-mesh still has continuous tangent
planes, but is not regular anymore. %the contact points do not form a smooth surface. 
Rather, the surface exhibits regression curves, the appearance of which is known from
offset surfaces. This is not surprising, as offsetting 
is a special Laguerre transformation. More generally,
applying a Laguerre transformation to a smooth
L-mesh may result in the appearance of
singularities in the point set. 

\begin{figure}
\centering
  \includegraphics[width=0.9\linewidth]{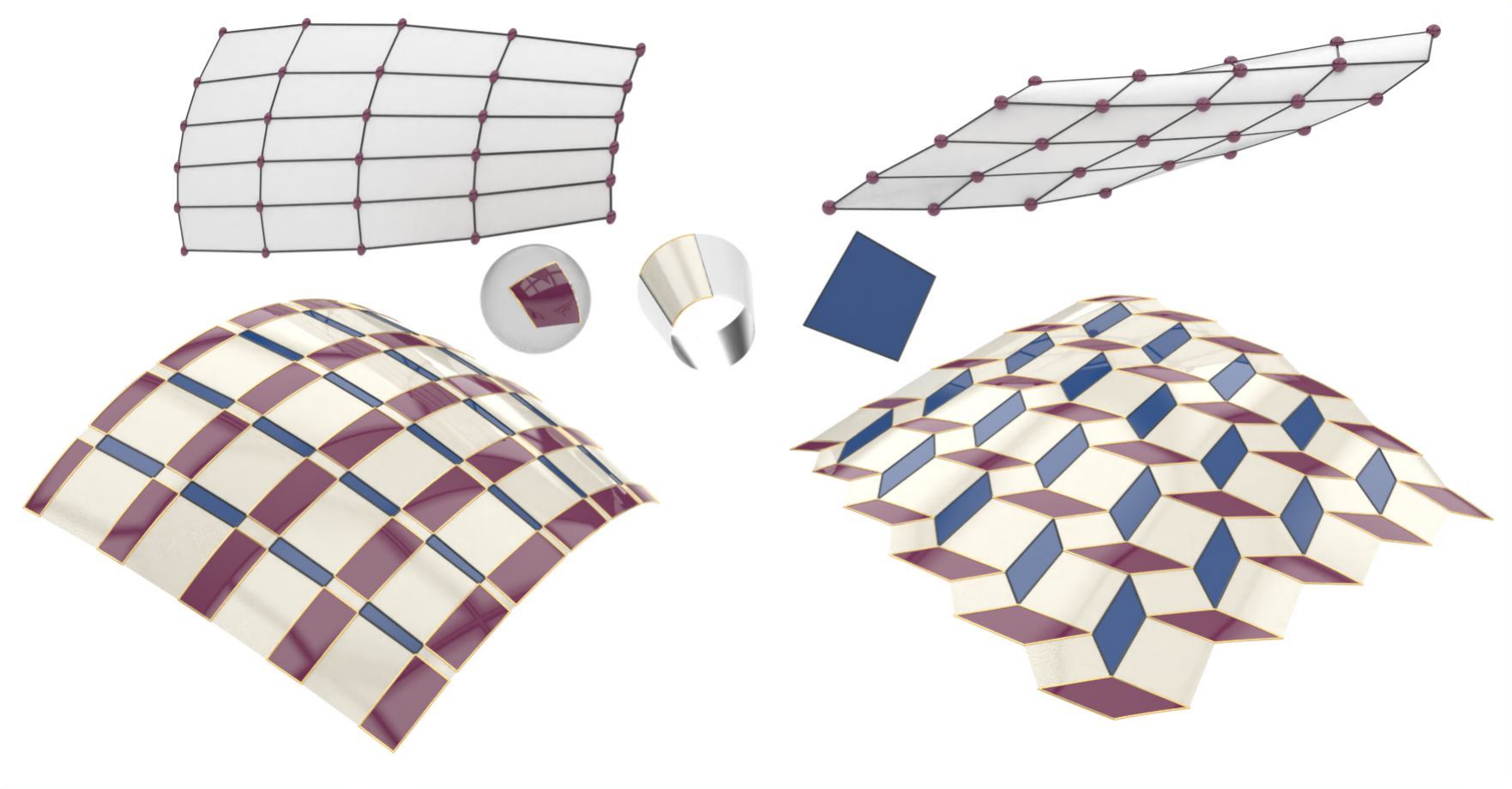}
\caption{An L-mesh represents a smooth surface in positively curved areas,
and consists of planar quads (blue), spherical panels bounded by circular
arcs (red), and pieces of rotational cones (beige), %which we 
\edit{viewed} as its edges. %The presence of \edit{conical patches} implies that the overall pattern is smooth 
For an L-mesh that is a discrete principal parameterization, \edit{the overall pattern is smooth} (left), while in a more general case, one has a staircase pattern (right). \edit{The sphere centers (red) form a mesh (gray) with %smooth, 
gradually curving discrete parameter lines (top row).}
%, as illustrated by the top view of the sphere center mesh (top row)}, shown as a \bluevarnew{gray} mesh with black edges and red vertices.
}
\label{fig:QL}
\end{figure}
The present paper deals with the \emph{approximation
of a given surface $f$
by a quad-L-mesh}, 
\edit{in which every planar and every spherical panel is surrounded by four conical patches
in a “gingham” pattern}
%\edit{, a surface composed of planar panels, conical sections, and spherical panels} 
(see, for example, Figs.~\ref{fig:spiral_defM} and~\ref{fig:seashellex5}). While we allow isolated combinatorial singularities,
we focus on \edit{regions with square grid ($\mathbb{Z}^2$) combinatorics} and call them \emph{L-nets}. However,
we want to \emph{prescribe the sphere radii of the
mesh spheres}. This may depend on the application.
A natural choice can be spheres that are good local
approximations of $f$.

Let us first look at a limit case where
all spheres in an L-net have radius zero.
Now, it is the same as a quad mesh with
planar faces, called Q-net in discrete differential geometry. Such Q-nets are discrete
counterparts of conjugate surface parameterizations \citep{bobenko-2009-ddg} and
this property is fundamental for solving the
approximation problem (see e.g.~\citep{liu+2006,pottmann-2015-ag}). One cannot
just apply numerical optimization to any initial quad mesh and push it towards an approximating
Q-net. Rather, the edges of the initial mesh should follow the directions in a conjugate frame field of $f$.

\edit{Our problem has a similar structure, which led us to a generalization of the concept of conjugacy to a Laguerre geometric counterpart, which we call \emph{L-conjugacy}. 
This allows us to follow a similar approximation approach for the present problem, where the sphere radii are prescribed. L-conjugacy is defined for a surface $f$ equipped with an attached sphere congruence, i.e., a 2-parameter family of spheres tangent to $f$. As we did not find such a concept in the literature, we first provide a thorough study  and then apply it to the approximation.} 

\begin{figure}
    \centering
    \includegraphics[width=1\linewidth]{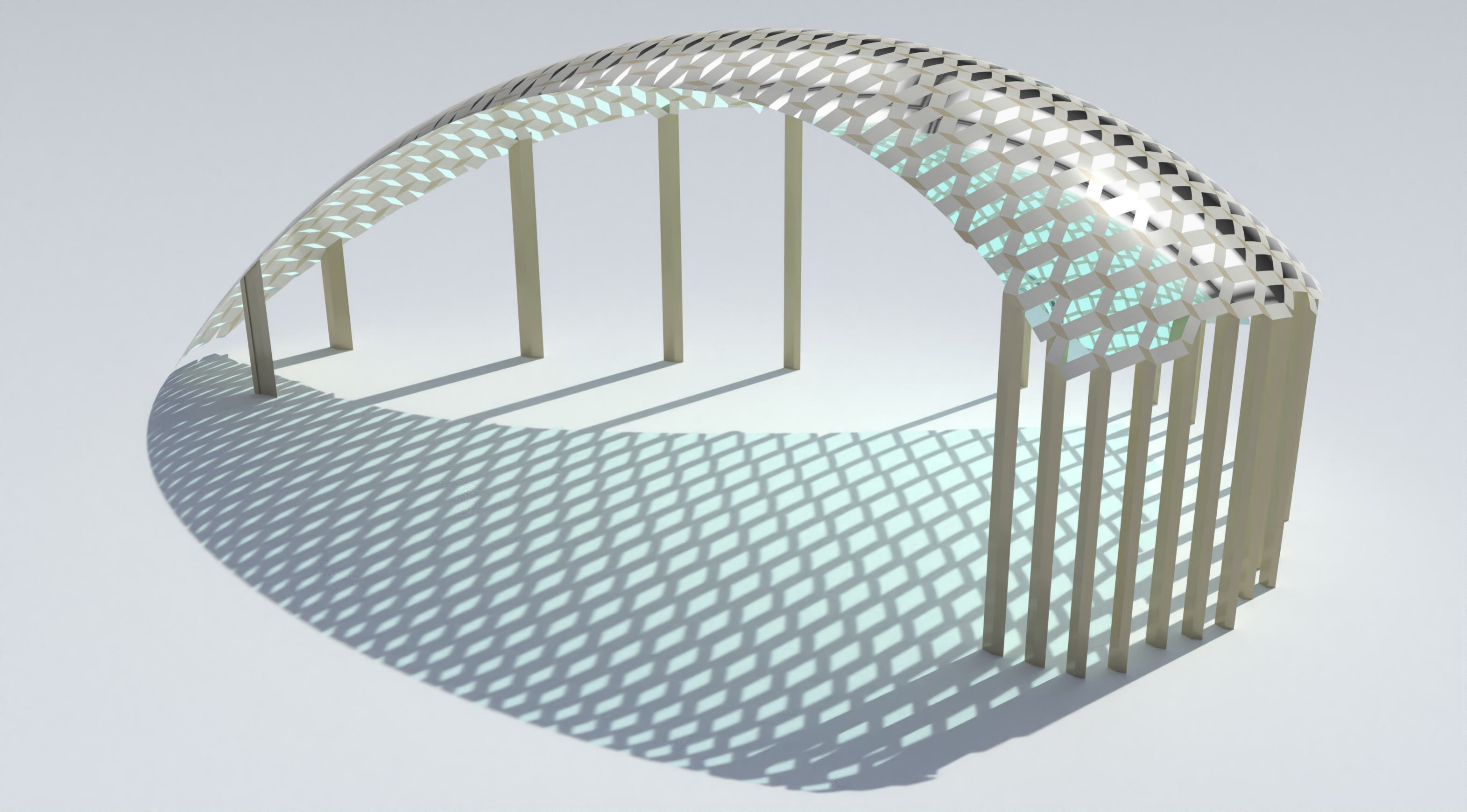}
    \caption{Roof concept design based on a portion of an L-net (shown in Fig.~\ref{fig:spiral_defM}), illustrating the watertight surface generated from the corresponding L-net structure.}
    \label{fig:roof}
\end{figure}

\subsection{Contributions and overview}
%--------------------------------------

The goal of the present research is the approximation of a given
surface with positive Gaussian curvature by L-nets. 

L-nets consist of smoothly joined planar, conical and spherical
panels.
For good approximation quality and other applications, we want to \edit{control the spheres through their radii}. Hence, our L-nets approximate a surface $f$ that appears as an envelope of a two-parameter family of spheres. 

After an introduction to basic concepts
of Laguerre geometry and meshes in this sphere geometry (Sec.~\ref{sec:basics}),
we introduce L-nets, in the discrete and smooth version in Sec.~\ref{sec:L-net} and~\ref{sec:limit}, respectively.
The control over the \edit{spheres through their radii}
amounts to the new concept of L-conjugate
parameterizations of a surface with respect to an attached sphere congruence.
We first derive the characterizing geometric properties and analytical representations of
these parameterizations in the 4-dimensional Minkowski
space (Sec.~\ref{ssec:L-conj-cyclo}), and then turn to the \edit{%Euclidean model of 
Laguerre geometry} 
in Euclidean 3-space (Sec.~\ref{ssec:L-conj}). Sec.~\ref{sec:curvature} relates
L-conjugacy to a dual viewpoint of surface theory that is based on dual
curvature radii. The number of self-L-conjugate directions provides 
a classification of contact elements of surfaces (Sec.~\ref{ssec:dualcurv}). 
In Sec.~\ref{ssec:L-asymp}, we discuss L-asymptotic parameterizations and their relation to a Laguerre geometric formulation of principal symmetric parameterizations \citep{s-nets-2020}. L-conjugacy with respect to important particular attached sphere congruences is studied in Sec.~\ref{ssec:special}. 
In Sec.~\ref{sec:approx}, L-conjugacy serves for initialization of a numerical optimization algorithm
which computes an L-net that approximates a given surface with control
over its \edit{spheres}. Various examples demonstrate the effectiveness
of our approach and illustrate the design flexibility and the quality of the resulting L-nets.
We conclude with \edit{computed examples demonstrating our algorithm} and pointers to future research in Sec.~\ref{sec:results}. 
 
\subsection{Related work}
%-------------------------

In view of the recent study of meshes with spherical
faces  within
M\"obius geometry \citep{spheremesh-2023,sphereapprox-2024},
%and its extension to Laguerre geometry \citep{Anthony-Laguerre1}, 
we keep the overview of prior research  short and refer for more details to the extensive literature reviews in these papers.

Our motivation for the research on spherical panel meshes is rooted in 
the paneling of architectural freeform skins \citep{eigensatz+2010, pottmann-2015-ag},
especially in solutions based on quad meshes with planar
faces (Q-nets). Since they are discrete counterparts to conjugate parameterizations of
surfaces, 
the major approach for approximating with Q-nets at first \edit{determines} a conjugate direction field on the given reference surface. Quad remeshing along these directions and post-optimization yield the final approximating Q-net (see e.g. \citep{liu+2006,liu:conjugate,zadravec-2010-vf}). There is a \edit{large space} of possible conjugate direction fields. A very recent contribution employs \edit{a data-driven approach based on neural networks to generate conjugate direction fields guided by user-specified strokes on the surface, avoiding expensive nonlinear optimization \citep{tao2025learningconjugatedirectionfields}.} 

In geometry processing, sphere meshes or approximations
of surfaces with spherical faces occur, for example, in \citep{sphere-panels,TGB:2013:SM,thiery16,tagliasacchi16}. 

Sphere geometries, especially from a differential geometric perspective, are
treated in the monographs by %Blaschke 
\cite{blaschke-1929} and %Hertrich-Jeromin 
\cite{hertrich-jeromin_2003}\edit{; a more recent exposition focusing on non-Euclidean Laguerre geometry is presented in \citep{alex2020noneuclidean}}. 

%Contributions to discrete differential sphere geometry include
%work on the Willmore energy for triangle meshes \citep{discrete-willmore},
% contributions to discrete principal nets \citep{bobenko-suris-isothermic,bobenko-2009-ddg,bobenko-suris-2007,liu+2006}, their extension to smooth cyclidic nets \citep{bobenko+2012} and principal symmetric nets \citep{s-nets-2020}. 
%Both principal and principal symmetric nets belong to Lie
% sphere geometry and \edit{have been studied} mainly from the perspectives of M\"obius geometry and/or Laguerre geometry.
\edit{Contributions to discrete differential sphere geometry include work on the Willmore energy for triangle meshes \citep{discrete-willmore}, discrete principal nets \citep{bobenko-suris-isothermic,bobenko-2009-ddg,bobenko-suris-2007,liu+2006}, their extension to smooth cyclidic nets \citep{bobenko+2012} and principal symmetric nets \citep{s-nets-2020}, as well as more recent discretizations of principal parametrizations such as checkerboard patterns \citep{peng2019checkerboard,dellinger2024discrete} and principal binets \citep{affolter2026principal}, formulated within Lie, M\"obius, and Laguerre geometry.}

Applications of Laguerre geometry in CAGD have so far mainly dealt with
% algebraic geometric 
problems \edit{in algebraic geometry}, such as rational offset families of curves or surfaces
 \citep{peternell-1998-lgaro,pottmann1998}. 
 We will also use the cyclographic model of Laguerre geometry %, which plays
in 4-dimensional Minkowski space. In this context, we refer to remarkable recent contributions on special quad meshes in Minkowski 3-space and their close relation
to certain
 models in statistical mechanics \citep{affolter+2024+ii,affolter+2025}. 
The isotropic model of Laguerre geometry is yet another classical model with nice applications, e.g., in CNC machining  \citep{pottmann-2009-lms,CNC-skopenkov-2020,skopenkov-2012-ruledlag}.

\begin{table}[htb]
    \centering
    \edit{
    \begin{tabular}{llll}
        \toprule
        \multicolumn{2}{c}{\textbf{Basic objects of Laguerre geometry}} 
        &
        \multicolumn{2}{c}{\textbf{Discrete and smooth L-conjugate nets}}
        \\[2pt] 
        \multicolumn{1}{c}{\textbf{Euclidean model (\ref{ssec:Lag-standard})}} & \multicolumn{1}{c}{\textbf{Cyclographic model (\ref{ssec:Lag-cyclo})}} &
        \multicolumn{1}{c}{\textbf{Euclidean model (\ref{sec:L-net},\ref{ssec:L-conj})}} & \multicolumn{1}{c}{\textbf{Cyclographic model (\ref{sec:L-net},\ref{ssec:L-conj-cyclo})}}
        %\\
        %\multicolumn{1}{c}{\textbf{Euclidean model (Sec.~\ref{ssec:Lag-standard})}} & \multicolumn{1}{c}{\textbf{Cyclographic model (Sec.~\ref{ssec:Lag-cyclo})}} &
        %\multicolumn{1}{c}{\textbf{Euclidean model (Secs.~\ref{sec:L-net},\ref{ssec:L-conj})}} & \multicolumn{1}{c}{\textbf{Cyclographic model (Secs.~\ref{sec:L-net},\ref{ssec:L-conj-cyclo})}}
        \\ 
        \midrule
        Oriented sphere & Point in $\mathbb{R}^{3,1}$ & Vertex sphere & Point-vertex in $\mathbb{R}^{3,1}$ \\
        Oriented cone   & Space-like line & Edge cone & Space-like edge \\
        Oriented plane  & %Isotropic 
        $\gamma$-hyperplane & Face plane & %Isotropic 
        $\gamma$-hyperplane attached to 
        face\\  
        Oriented contact %with 
        w.\,plane & %$\gamma$-hyperplane % THIS DOES NOT SOUND GOOD TO ME --- WE CAN DISCUSS IT IN PERSON 
        Incidence %with 
        w.\,$\gamma$-hyperplane
        %Inclusion into $\gamma$-hyperplane
        & Oriented contact cone-plane %along ruling
        %circle & Vertex-edge incidence \\ 
        %& edge in $\gamma$-hyperplane\\ 
        & Incidence edge--$\gamma$-hyperplane\\ 
        %& Inclusion edge$\,\subset\gamma$-hyperplane\\ 
        Laguerre transformation & Minkowski similarity & %Attached 
        Sphere congruence & Net in $\mathbb{R}^{3,1}$\\
        Offsetting operation & $x_4$-parallel translation & Net 
        %Tangents % THIS DOES NOT SOUND GOOD TO ME --- WE CAN DISCUSS IT IN PERSON
        in $\mathbb{R}^{3}$ & %Attached 
        $\gamma$-hyperplane congruence \\
        %Point (zero-radius sphere) & Point %in the hyperplane 
        %                                   with $x_4=0$ & & \\
        \bottomrule
    \end{tabular}    
    \caption{The informal correspondence between the Euclidean and the cyclographic model of Laguerre geometry. %In the case of discrete and smooth 
    %For L-conjugate nets, %the roles of the vertices and faces, as well as 
    %the roles of the net and the attached sphere/hyperplane congruence are naturally %inevitably 
    %swapped: 
    %\HP{what does this mean? I would remove this sentence} 
    Note that a net in $\mathbb{R}^{3,1}$ corresponds to a sphere congruence in $\mathbb{R}^3$ and a net in $\mathbb{R}^3$ corresponds to a $\gamma$-hyperplane congruence in $\mathbb{R}^{3,1}$. 
    For details, we refer to section numbers in parentheses.}
    \label{tab:correspondence}
    }
\end{table}

%%%%%%%%%%%%%%%%%%%%%%%%%%%%%%%%%%%%%%%%%%%%%%%%%%%%%%%%%%%%%%%%%%%%%%%%
\section{A few basics of Laguerre geometry and L-nets} \label{sec:basics}
%%%%%%%%%%%%%%%%%%%%%%%%%%%%%%%%%%%%%%%%%%%%%%%%%%%%%%%%%%%%%%%%%%%%%%%%

In this section, we provide a very short outline of Laguerre sphere
geometry and quad meshes with \edit{square grid} combinatorics, called L-nets, in this geometry. We do that in \edit{3-dimensional} Euclidean space (\edit{informally referred to as \emph{Euclidean model} of Laguerre geometry}; Sec.~\ref{ssec:Lag-standard}) and the so-called cyclographic model in 4-dimensional
Minkowski space (Sec.~\ref{ssec:Lag-cyclo}). \edit{We connect the models in Tab.~\ref{tab:correspondence} serving as a road-map of Secs.~\ref{sec:basics}--\ref{sec:limit}.} In the latter model, 
L-nets are probably more easily understood than in the \edit{Euclidean}
model. Most importantly, the cyclographic model is
particularly useful for the discussion of smooth analogs of L-nets
and the associated concept of L-conjugacy (see Sec.~\ref{sec:limit}). 

For detailed discussions of Laguerre geometry and its
various models, we refer to \citep{blaschke-1929,cecil2008lie,pottmann1998}. 
%More general types of meshes in Laguerre geometry and their
%representation in all kinds of models are
%studied in \citep{Anthony-Laguerre1}. 

\begin{figure}
    \centering
    \includegraphics[width=0.85\linewidth]{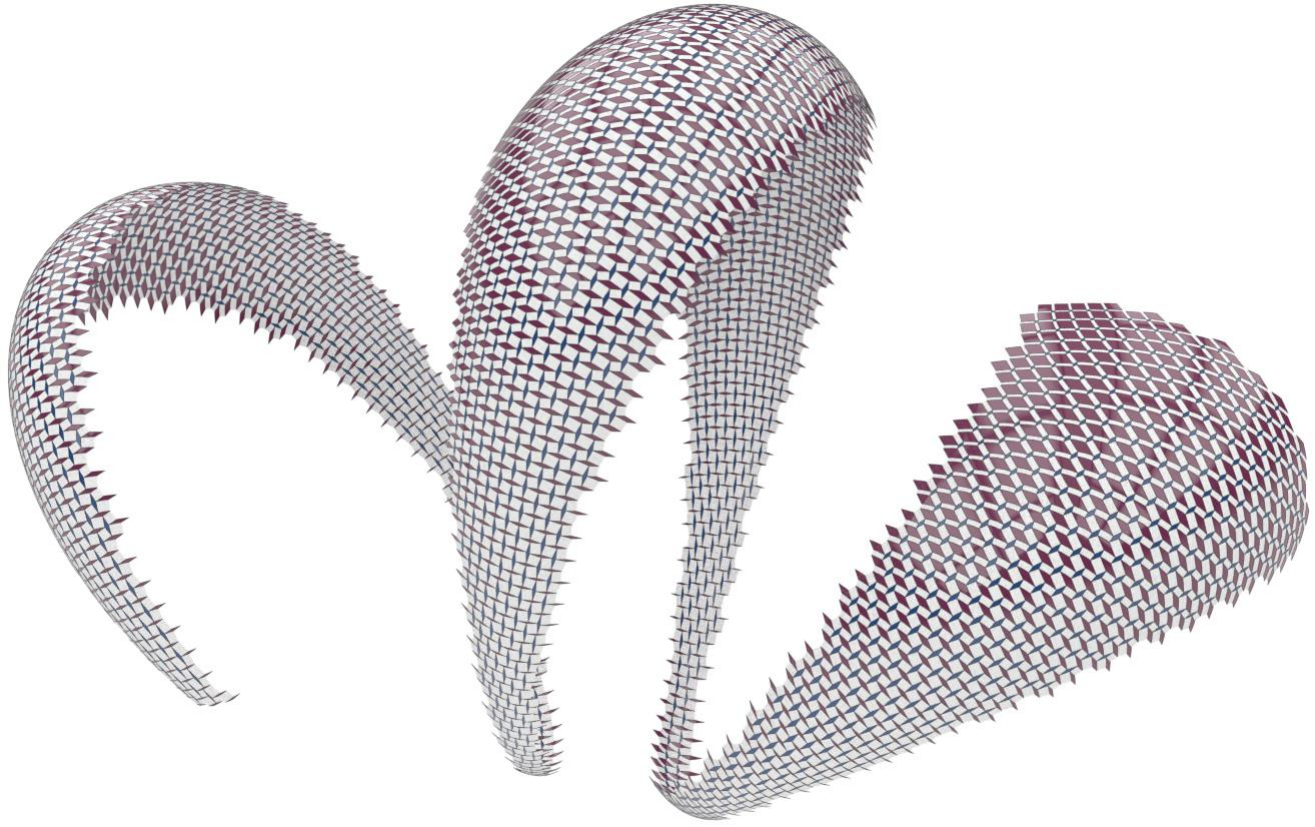}
    \caption{L-net approximation of a spiral-like surface. Smaller spherical panels appear in the narrow regions of the surface, whereas larger spherical panels  arise in the wider regions.}
    \label{fig:spiral_defM}
\end{figure}

\subsection{Laguerre Geometry in \edit{3-dimensional} Euclidean space} \label{ssec:Lag-standard}
%----------------------------------------------------------

First, we outline the main concepts informally, and then give \edit{rigorous} definitions.
The basic elements of Laguerre sphere geometry in Euclidean space $\R^3$ are
\emph{oriented spheres} %(or-spheres)  
and \emph{oriented planes}. % (or-planes). 
\emph{Orientation} is understood as the choice of 
a continuous family of unit normal vectors at all points of a surface. %normal direction for planes and spheres. 
\edit{\emph{Oriented contact} between elements means tangency and agreement in orientation.} 
Oriented spheres are allowed to have zero radius, i.e. degenerate \edit{to} points;
\edit{in this case,} %For elements degenerating to points, 
oriented contact is defined separately. 
Laguerre geometry studies properties of its elements that are invariant under \emph{Laguerre transformations} (see Fig.~\ref{fig:lagtrafo}). These
are pairs of bijective transformations of the sets of oriented spheres and oriented planes, 
respectively,
preserving oriented contact. %i.e., elements that are tangent and agree in orientation. 
%For elements degenerating to points, \edit{oriented contact} is defined separately. 
\edit{For example, an \emph{offsetting transformation}} takes an oriented plane or an oriented sphere to its offset at a fixed distance in the direction %opposite to
of the chosen normal. Laguerre transformations do not necessarily preserve points, as those are viewed as spheres of zero radius and can be mapped to other spheres. Likewise, lines can be mapped to \emph{oriented cones}. We may view an oriented cone as the envelope of oriented tangent planes that are in oriented contact with two distinct oriented spheres. % (if such planes exist). 
Oriented cones %are allowed to 
may degenerate \edit{to} %otational 
cylinders or lines and are preserved under Laguerre transformations.

\begin{figure}
    \centering
    \includegraphics[width=0.8\linewidth]{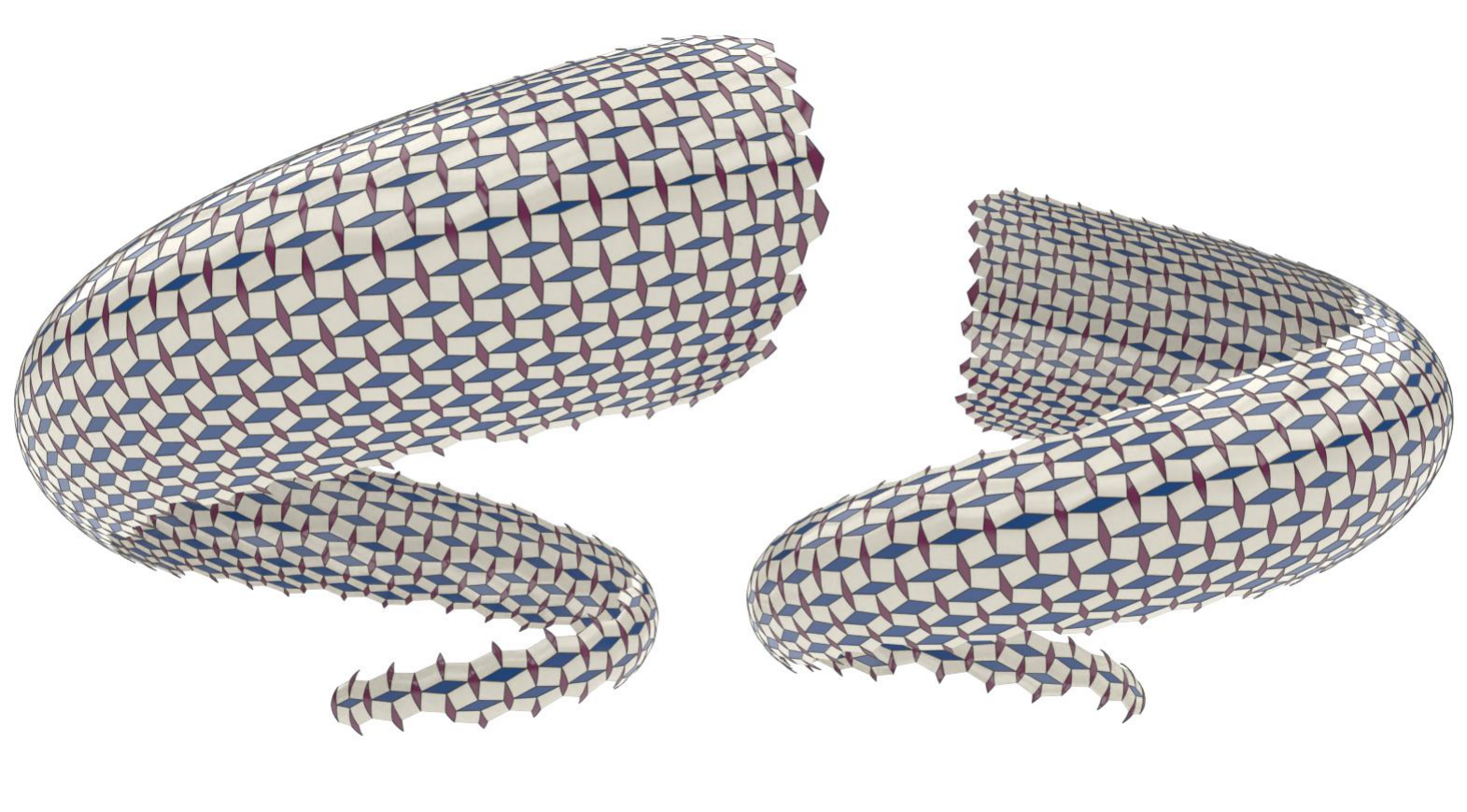}
    \caption{%Smooth quad-
    L-net approximation of a seashell-like surface.}
    \label{fig:seashellex5}
\end{figure}

Let us discuss these notions systematically and provide precise definitions.

\paragraph{Oriented spheres}  
An \emph{oriented sphere}, shortly \emph{or-sphere}, is %an oriented 
either a Euclidean sphere in $\R^3$ equipped with a continuous family of unit normal vectors at all points, or a point in $\R^3$. 
%An or-sphere is \emph{positively}
%the normals point inwards of the sphere, i.e. towards its center
%The \emph{signed radius} of the or-sphere is the radius taken with a sign depending on the orie
An or-sphere is given
by its center $c \in \R^3$ and \emph{signed radius} $r \in \R$, which is
defined as follows. %For the orientation of the Euclidean sphere, where 
If the normals point \edit{inward}, i.e., toward its center, the \emph{signed radius} is \edit{the} radius (with a positive sign).
For outward-pointing normals, the \emph{signed radius} is \edit{minus the} radius (i.e., the radius with a negative sign). For a point, the \emph{signed radius} is zero.
%which by definition is the radius of the sphere with a positive [negative] sign depending on if the normals 
%where $r >0$ expresses an orientation where the normals point inwards of the sphere, i.e. towards its center, $r=0$ characterizes a point, and $r<0$ belongs to outwards pointing normals. \MS{A more accurate definition is needed. MS is thinking of it.}

\paragraph{Oriented planes}
An \emph{oriented plane}  (\emph{or-plane}) $p$ is a plane in $\R^3$ equipped with a continuous family of unit normal vectors at all points (which \edit{is therefore} the same vector $n$ at all points). \edit{An or-plane $p$ is given by its unit normal vector $n$ and the \emph{signed distance} $h$ to the origin, which is defined as follows. 
If the plane contains the origin, the \emph{signed distance} is zero.
Otherwise, if the normal $n$ points toward the origin, the \emph{signed distance} is the distance (with a positive sign). If the normal points in the opposite direction, the \emph{signed distance} is minus the distance. % (i.e., the distance with a negative sign). For a point, 

The plane $p$ is thus the set of all $x \in \R^3$ satisfying the equation}
%% whose orientation is given by a unit normal vector $n$. 
%We describe $p$ by its Hesse normal form 
%
\be \la n,x \ra + h=0. \label{eq:or-plane}  \ee
%
%%where $n$
%%is the unit normal vector representing its orientation. 
%Here, the left side is the \emph{signed distance} from a point $x$ to $p$. In particular, the intercept $h$ is the signed distance from the or-plane $p$ to the origin.
%\edit{We describe $p$ by its unit normal vector $n$ and the signed distance $h$ to the origin.
%The \emph{signed distance} from a point $x \in \R^3$ to $p$ is $\langle n, x \rangle + h$; in particular, $h$ is the signed distance from the origin to $p$. The plane $p$ is thus the set $\{ x \in \R^3 \mid \langle n, x \rangle + h = 0 \}$.}

\paragraph{Oriented contact} \emph{Oriented contact} (\emph{or-contact}) of oriented spheres/planes
means tangency in the usual sense and agreement of unit normals
at a contact point. A point is in \emph{oriented contact} with an oriented sphere/plane if it lies on the sphere/plane.
%Thus, an oriented sphere $s=(c,r)$ with  center $c \in \R^3$ and signed radius $r$ is in oriented contact with the or-plane $p$ if and only if $c$ has \emph{signed distance} $r$ from $p$, i.e.
Thus, \edit{the or-plane $p$  is in oriented contact with} an oriented sphere $s=(c,r)$ with 
center $c \in \R^3$ and signed radius $r$ if
and only if \edit{$p$} has \emph{signed distance} $r$ from \edit{$c$}, i.e.
\be \la n,c \ra + h=r. \label{eq:contact} \ee
\paragraph{Laguerre transformations}
\emph{Laguerre transformations} are pairs of bijective transformations of the sets of \edit{or-}spheres and \edit{or-}planes, respectively, preserving oriented 
contact. \edit{A simple example is} the \emph{$d$-offsetting operation} which adds the same constant $d$ to the signed radius $r$ of each or-sphere and to the signed distance~$h$ from each or-plane to the origin, while preserving the centers $c$ of or-spheres and the normal vectors $n$ of or-planes. By~\eqref{eq:contact}, \edit{this} transformation preserves the \edit{or-}contact between or-spheres and or-planes. \edit{Euclidean isometries of $\mathbb{R}^3$ are also Laguerre transformations.}
%\edit{We refer to Section~\ref{ssec:Lag-cyclo} for a description of more non-trivial Laguerre transformations in the cyclographic model.} 

\paragraph{Oriented cones} A \emph{rotational cone} is the union of lines (\emph{rulings}) passing through a fixed point (\emph{vertex}) and forming a fixed angle (nonzero and non-right) with a fixed line in~$\mathbb{R}^3$ passing through the fixed point. An \emph{oriented rotational cone} is a rotational cone equipped with a continuous family of unit normal vectors at all points (excluding the vertex) such that the normal vectors along each ruling are the same. (The latter requirement is needed because the cone without the vertex splits into two components, which a priori can be oriented independently.) An \emph{oriented rotational cylinder} is the union of lines at a fixed nonzero distance from a fixed line, equipped with a continuous family of unit normal vectors at all points. A \emph{generalized oriented cone} (%or simply
\emph{or-cone}) is an oriented rotational cone or oriented rotational cylinder or (non-oriented) straight line in~$\mathbb{R}^3$.

\emph{Oriented contact} %(\emph{or-contact}) 
between an or-plane and a two-dimensional or-cone means tangency in the usual sense and agreement of unit normals at a contact point (distinct from the cone vertex). \emph{Oriented contact} between an or-plane and a line means that the line is contained in the or-plane. 
%\MS{An or-cone and an or-sphere are in \emph{oriented contact along a circle} if one of the following conditions holds:
%\begin{itemize}
%    \item both the or-cone and or-sphere are two-dimensional and are in oriented contact at each point of some circle of positive radius;
%    \item the or-cone is two-dimensional, whereas the or-sphere has radius $0$ and coincides with the vertex of the cone;
%    \item the or-cone is a line, and the or-sphere has radius $0$ and lies on it.
%\end{itemize}}

The set of or-planes in oriented contact with an or-cone coincides with the set of or-planes in oriented contact with a fixed pair of spheres $s_i=(c_i,r_i)$, where $i=0,1$. Namely, for an oriented rotational cone, these two spheres have radii $r_0 \ne r_1$, for an oriented rotational cylinder, we have $r_0=r_1 \ne 0$, and for a straight line, we have $r_0=r_1=0$. The or-spheres $s_0$ and $s_1$ must have common
oriented tangent planes, which requires $\|c_0-c_1\|^2 > (r_0-r_1)^2$.

All the or-spheres in the \emph{linear family} $s(t):=(1-t)s_0+ts_1 :=\left((1-t)c_0+tc_1,(1-t)r_0+tr_1\right)$, where $t\in \R$, are also in oriented contact with all those or-planes. Any of the or-spheres $s(t)$ is referred to as an or-sphere \emph{in oriented contact with the or-cone along a circle}. In particular, a point $s(t)$ (the or-sphere with $r(t):=(1-t)r_0+tr_1=0$) is the vertex of the cone or
a point of the straight line to which the or-cone degenerates in case of $r_0=r_1=0$; it is still viewed as having \emph{oriented contact with the or-cone along a circle}.

We use $\OS$ to denote the set of or-spheres, $\OP$ for the set of or-planes, and $\OC$ for the set of or-cones.

\subsection{Cyclographic Model}  \label{ssec:Lag-cyclo}
%-----------------------------------------------------------------
\begin{figure}
    \centering
    \includegraphics[width=0.7\linewidth]{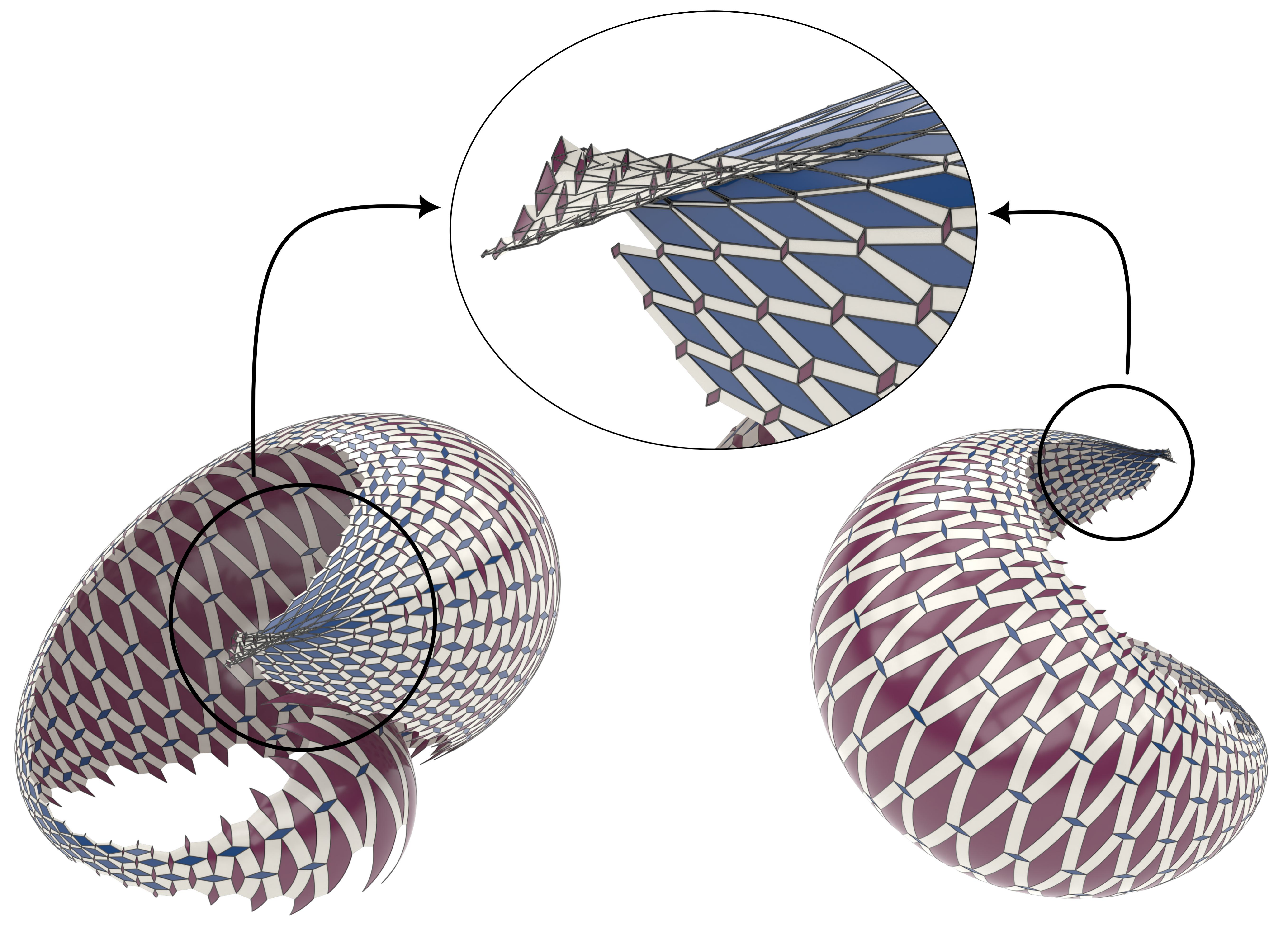}
    \caption{A Laguerre transformation of the L-net in Fig.~\ref{fig:seashellex5}. \edit{While such transformations offer a tool for exploring different designs, singularities may emerge in the process.} The inset shows a close-up of a region where singularities appear.}
    %Laguerre transformations allow us to explore different designs, but singularities can appear in the process.
    %\AR{A Laguerre transformation of the L-net in Fig.~\ref{fig:seashellex5}, illustrating shape exploration through Laguerre transformations. The inset shows a close-up of a region where singularities appear. While Laguerre transformations offer a tool for exploring different designs, singularities may emerge in the process.}}
    \label{fig:lagtrafo}
\end{figure}

The cyclographic model of Laguerre geometry considers an oriented sphere $s = (c, r)$ with 
center $c =(c_1, c_2, c_3)\in \R^3$ and signed radius $r$ as a point $S = (c_1, c_2, c_3, r)$ in Minkowski space $\mathbb{R}^{3,1}$ with the inner product $\lla \cdot , \cdot \rra$ of signature $(+,+,+,-)$. The map $\zeta\colon s \mapsto S$ is called the \emph{Minkowski lift}, and its inverse $\zeta^{-1}$ is %known as 
the \emph{cyclographic map}; \edit{see Tab.~\ref{tab:correspondence}}.

The Minkowski lift of the or-plane $p$ 
is defined as the set of points $X=(x, x_4) \in \mathbb{R}^{3,1}$, taken to or-spheres in oriented contact with $p$ by the cyclographic map. By~\eqref{eq:or-plane} and~\eqref{eq:contact}, this is expressed as
%W
\be x_4= \la n,x \ra + h, \label{eq:g-hyperplane} \ee
and represents a hyperplane $P=:\zeta(p) \subset \R^{3,1}$ forming the Euclidean angle $\gamma=\pi/4$ with the hyperplane $\R^3$ given by $ x_4=0$. A hyperplane of form~\eqref{eq:g-hyperplane} is called a \emph{$\gamma$-hyperplane} or
\emph{isotropic hyperplane}.
We denote by $\Gamma(\R^{3,1})$ the set of all $\gamma$-hyperplanes. Note
that \eqref{eq:g-hyperplane} can be written with the isotropic normal
vector $N=(n,1)$ that is both normal and parallel to the hyperplane,  as 
$$    \lla N, X \rra +h=0. $$

\paragraph{Lines in Minkowski space}
A line in Minkowski space is called \emph{space-like}, \emph{light-like} (\emph{isotropic}), or \emph{time-like}, respectively, if its direction vector $g$ satisfies  $\lla g, g \rra >0, \ \lla  g,  g \rra =0 , \ \lla  g,  g \rra < 0$, respectively. The cyclographic image of a space-like \edit{line is a family of or-spheres in oriented contact with an or-cone along a circle, and the line is viewed as the \emph{Minkowski lift} of that or-cone. The cyclographic image of an isotropic line is a family of or-spheres tangent to each other at a single point and sharing the same tangent plane there (thus forming a \emph{contact element}, but we do not need this terminology). A time-like line corresponds to a family of or-spheres whose centers lie on a fixed line orthogonal to a plane, while all spheres intersect that plane at a constant angle.}

\edit{
\paragraph{Planes in Minkowski space}
A 2-dimensional plane in $\R^ {3,1}$ is called \emph{space-like}, \emph{light-like (isotropic)}, or \emph{time-like} if it contains zero, one, or two isotropic lines
%, one %bundle of parallel real 
%isotropic line 
through each point, %of $\Pi$, %or if there is a real quadratic cone of isotropic lines %of $H$ 
%through each point of $\Pi$, 
respectively.}

\paragraph{Laguerre transformations} In $\R^{3,1}$, Laguerre transformations appear
as those bijective affine maps which map isotropic lines to isotropic lines. This implies that isotropic planes, characterized by containing a single real family of parallel
isotropic lines, are mapped to isotropic planes. Hence, in $\R^3$, we have two bijective
maps of the sets of or-spheres and or-planes, respectively, which preserve oriented contact (Fig.~\ref{fig:lagtrafo}). \edit{These transformations correspond to Minkowski similarities of $\mathbb{R}^{3,1}$, generated by the orthogonal group $O(3,1)$, translations, and scalings.} For instance, the $d$-offsetting operation appears as the translation $x_4\mapsto x_4+d$ in $\R^{3,1}$.

\subsection{L-nets}
\label{sec:L-net}
%---------------------------

We are now ready to define the main objects of interest, namely
quad meshes with \edit{square grid} combinatorics in Laguerre geometry. We call
them L-nets and define them as follows (see Fig.~\ref{fig:defL-net} to the left).

\begin{defn}[L-net] %\MS{MS will simplify it for a square grid}
An \textbf{$m \times n$ square grid} is the plane graph with the vertices at integer points $(i,j)$, where $0 \le i \le n, \
0  \le j \le m$, and edges joining the vertices at unit distance. Denote by  $V$ the set of vertices, by $E$ the set of edges, and by $F$ the set of bounded faces of the square grid.
An \textbf{L-net} \edit{in $\mathbb{R}^3$} is a triple of maps
\[
    \edit{\SV : V \to \OS}, \ \
    \CE : E \to \OC, \ \
    \edit{\FP : F \to \OP} ,
\]
satisfying the following conditions:
\begin{itemize}
    \item \edit{if an edge $e \in E$ contains a vertex $v \in V$, then the or-cone  $\CE(e)$ is in oriented contact with the or-sphere $\SV(v)$ along a circle};
    \item \edit{if an edge $e \in E$ meets a face $f \in F$, then the or-cone $\CE(e)$ is in oriented contact with the or-plane $\FP(f)$ along a ruling}. %\MS{We need to define what is ``oriented contact along a circle/a ruling''. This is non-trivial in the case when this circle degenerates to a point or the or-cone degenerates to a line.}
\end{itemize}
%A map $\VP: \Z^2 \to \OP$ is called an L-net if at each $(i,j) \in \Z^2$ the four oriented planes $p_{ij}:=\VP(i,j), p_{i+1,j},p_{i+1,j+1}$ and $p_{i,j+1}$ are in oriented contact with an or-sphere $s_{i,j}$.
%\MS{The latter condition is generically automatic, so that the whole definition looks weird. We'd better revert to the one from the previous paper.}
\end{defn}

\begin{wrapfigure}[9]{l}{0.30\linewidth}
\vspace{-0.3cm}
\hspace{0.3cm}

\begin{overpic}[width=0.9\linewidth]{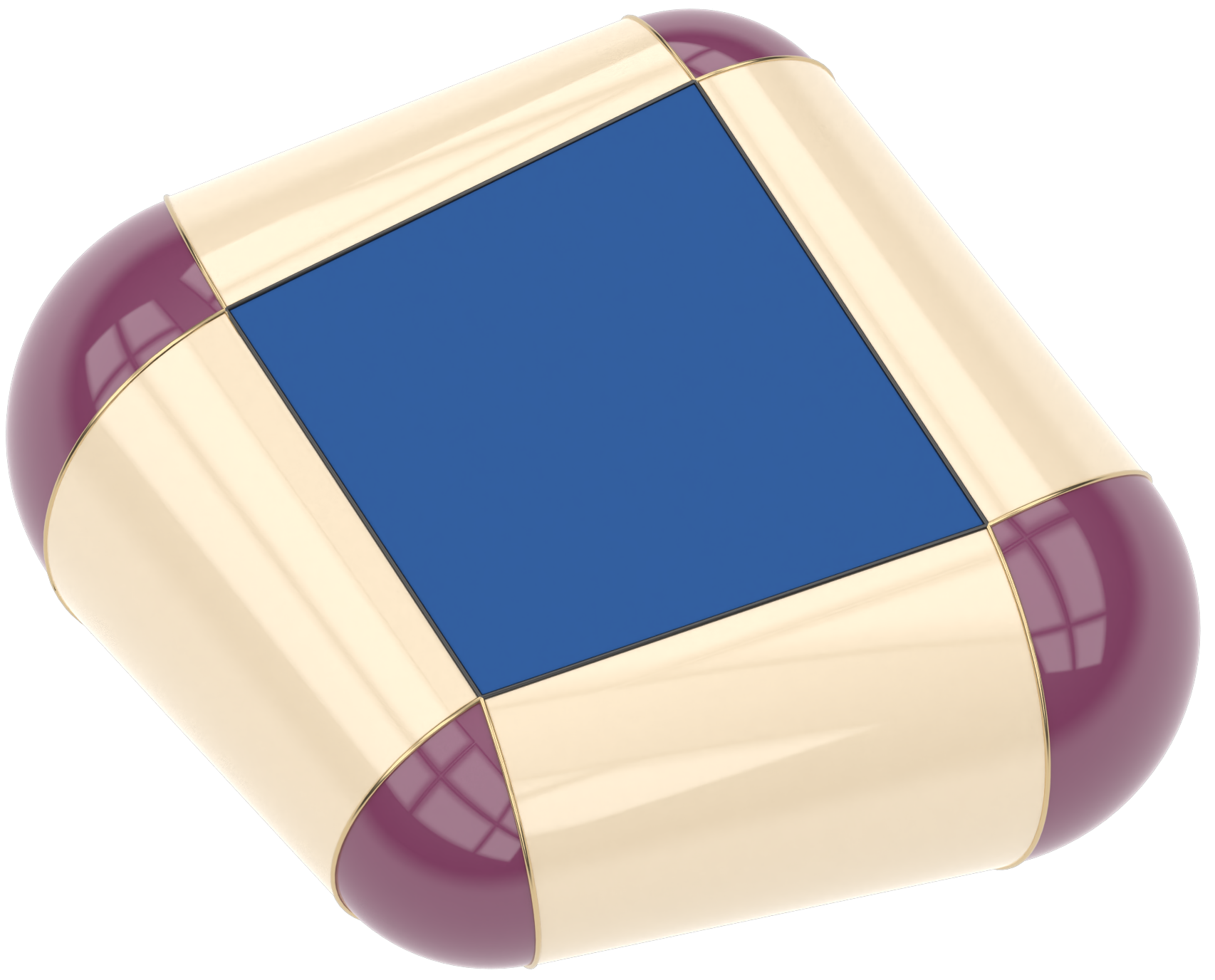}
\put(45,50){\contour{white}{$p_{ij}$}}
\put(90,30){\contour{white}{$s_{i+1,j}$}}
\put(1,55){\contour{white}{$s_{i,j+1}$}}
\put(32,10){\contour{white}{$s_{ij}$}}
\put(58,75){\contour{white}{$s_{i+1,j+1}$}}
\end{overpic}
\end{wrapfigure}

Our notation indicates that we view the \edit{or-spheres as vertices}%. %Likewise
, the \edit{or-planes as faces} %$s_{ij}$ are seen
, and %form the set $\SF$. As edges we consider
the or-cones %$c_{i+,j}$
as edges.
%that are in or-contact with planes $p_{i,j},p_{i+1,j}$ and or-spheres $s_{ij},s_{i,j-1}$; analogously we define the edge cones $c_{i,j+}$ in the $j$-parameter direction. All edge cones form the set $\CE$. To point out the presence of vertex planes, edge cones and face spheres, we will write an L-net also as a triple $(\VP,\CE,\SF)$. 
Quadrilateral patches of or-planes, or-cones, and or-spheres of an L-net define \edit{a} surface $\mathcal{L}$ as shown in the inset and Fig.~\ref{fig:QL}. We assume that the figure is self-explanatory
and do not further detail the boundaries of these patches. We see
that $\mathcal{L}$ can be a $C^1$ surface. % (unless some of or-cones or or-spheres degenerate). 

\edit{
%\begin{remark}
    The definition extends %naturally 
    to more general combinatorics, because the oriented contact conditions are imposed locally. For simplicity, % of exposition, 
    we restrict our study to square grid combinatorics; the general case is the subject of ongoing work.
%\end{remark}

}

%\HP{Maybe we will not need the following symbol.}
%We use $\sim$ to denote the or-contact relation between elements of an L-net. The notation $p\sim s$ means or-plane $p$ is in oriented contact at a single point with or-sphere $s$. The notation $p \sim c$ means or-plane $p$ is in oriented contact with or-cone $c$ along a ruling, and $s \sim c$ means or-sphere $s$ is in oriented contact with or-cone $c$ along a circle.

%\MS{MS will edit the following text to avoid the notation $p_{ij}$ and $s_{ij}$.}
For each vertex $(i,j)\in V$, denote \edit{$s_{ij}:=\SV(i,j)$.} For a face $f$ with the vertices $(i,j)$, $(i+1,j)$, $(i+1,j+1)$, $(i,j+1)$ \edit{(see inset figure)}, denote \edit{$p_{ij}:=\FP(f)$.} %By a slight abuse of notation, we
Sometimes, we also denote by \edit{$\SV$ and $\FP$} the images of the maps \edit{$\SV$ and~$\FP$.} % as well.

\begin{figure}
\centering
    \begin{overpic}[width=0.8\linewidth]{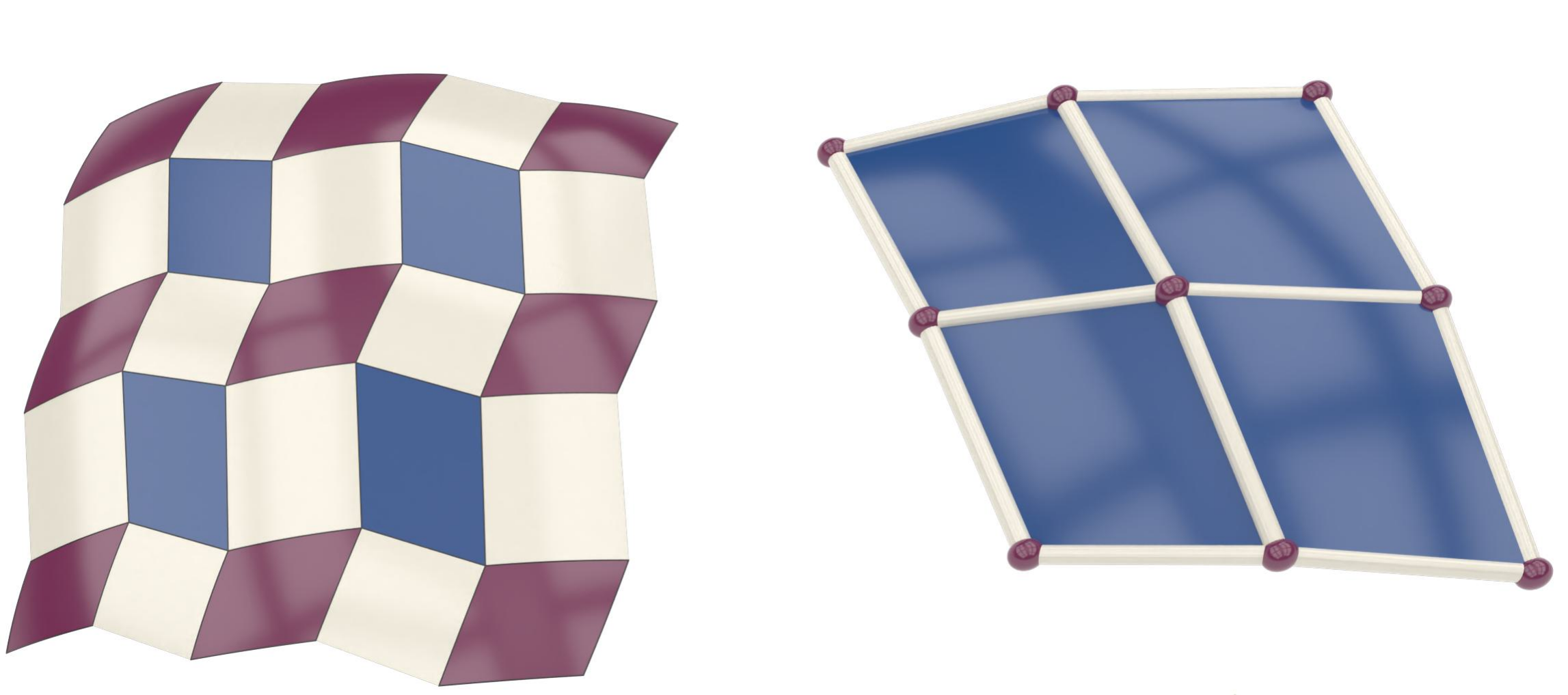}
    \put(2.5,6){\color{white} $s_{ij}$}
    \put(14.6, 6.5){\color{white} $s_{i+1,\  j}$}
    \put(31.5,5){\color{white} $s_{i+2,j}$}
    \put(3.7, 22.5){\color{white} $s_{i,j+1}$}
    \put(16.5,23){\color{white} $s_{i+1,j+1}$}
    \put(32.7,22.5){\color{white} $s_{i+2,j+1}$}
    \put(6.5,35){\color{white} $s_{i,j+2}$}
    \put(19,36.5){\color{white} $s_{i+1,j+2}$}
    \put(35.5,35.5){\color{white} $s_{i+2,j+2}$}
    \put(10,15){\color{white}$p_{ij}$}
    \put(25,15){\color{white}$p_{_{i+1,j}}$}
    \put(12,31){\color{white}$p_{_{i,j+1}}$}
    \put(27,31){\color{white}$p_{_{i+1,j+1}}$}
    \put(64,5){\contour{white}{$S_{ij}$}}
    \put(72,28){\contour{white}{$S_{i+1,j+1}$}}
    \put(80,6){\contour{white}{$S_{i+1,j}$}}
    \put(95,5){\contour{white}{$S_{i+2,j}$}}
    \put(55,27){\contour{white}{$S_{i,j+1}$}}
    \put(90,27){\contour{white}{$S_{i+2,j+1}$}}
    \put(52.5,37){\contour{white}{$S_{i,j+2}$}}
    \put(67.5,40){\contour{white}{$S_{i+1,j+2}$}}
    \put(84,40){\contour{white}{$S_{i+2,j+2}$}}
    \put(70,17){\color{white}$P_{ij}$}
    \put(83,17){\color{white}$P_{i+1,j}$}
    \put(61,32){\color{white}$P_{i,j+1}$}
    \put(75,33){\color{white}$P_{i+1,j+1}$}
    \end{overpic}
\caption{L-net elements. (Left) \edit{Euclidean model of Laguerre geometry}; (right)  cyclographic model.
}
\label{fig:defL-net}
\end{figure}
 %In the cyclographic model, $\sim$ is understood as the incidence relation between mesh elements.

\paragraph{Examples of L-nets}  \edit{If all spheres $s_{ij}$ of an L-net are points,
we get a Q-net (quad mesh with planar faces).} Q-nets are discrete
conjugate surface parameterizations and \edit{have played a fundamental role} in the development of
discrete differential geometry \citep{bobenko-2009-ddg}. To obtain a less degenerate
appearance of the L-net, we may apply a Laguerre transformation, for example, a $d$-offsetting
operation. %It adds a constant to the signed radius of each sphere and thus appears as translation $x_4 \mapsto x_4+d$ in $\R^{3,1}$. 
Now all \edit{spheres} have radius $d$ and edge cones
are cylinders of radius $d$. That L-net represents a smooth surface
in ``positively curved'' areas, but has singularities in areas of ``negative curvature''; see  Fig.~\ref{fig:L-examples}.

\begin{figure}
    \includegraphics[width=0.95\linewidth]{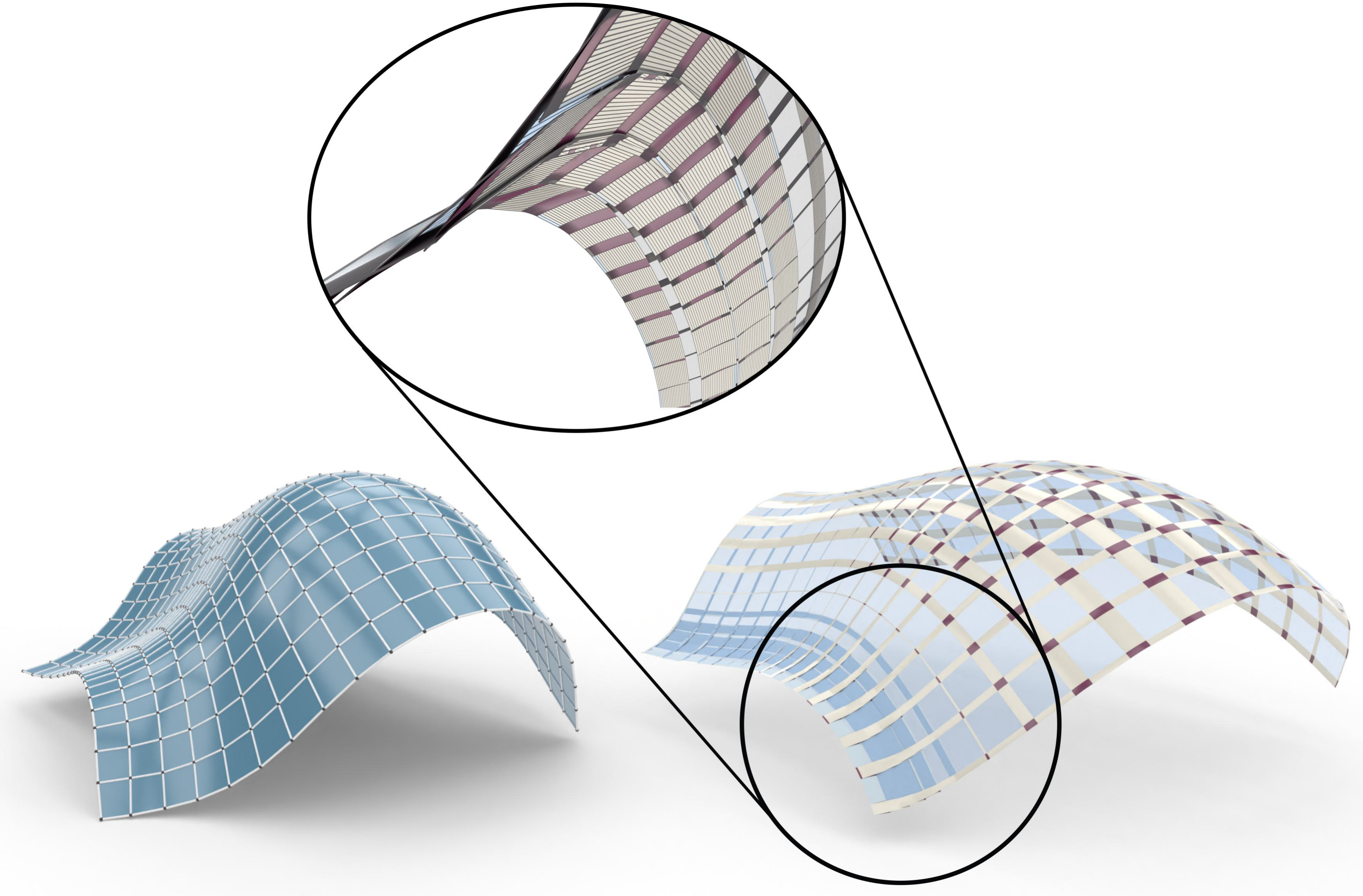}
\caption{A %conical 
Q-net (left) extended to an L-net (right) by offsetting. Spheres are shown in red and planar faces as translucent blue panels to reveal the singular behavior. The inset highlights \edit{the} bottom view of the swallowtail singularities arising in regions of ``negative curvature''.
}
\label{fig:L-examples}
\end{figure}

A special case of the previous example is an L-net in which the four or-planes in or-contact with an or-sphere $s_{ij}$ are in or-contact with another or-sphere $\bar{s}_{ij}$, and thus are in or-contact with an or-cone $c_{ij}$ that is in or-contact with $s_{ij}$ and $\bar{s}_{ij}$ along circles. If spheres $s_{ij}$ are points, such an L-net is a \emph{conical net}, known as a Laguerre geometric
object and a discrete principal curvature parameterization \citep{bobenko-2009-ddg,liu+2006}.
It gives rise to infinitely many less degenerate L-nets, since we can replace the point $s_{ij}$
with any sufficiently small or-sphere $s'_{ij}$ in oriented contact with $c_{ij}$ along a circle, while keeping all the or-planes $p_{ij}$ and adjusting the edges of the L-net accordingly.
The freedom in choosing the sphere radius provides
a simple solution to our approximation problem. We may approximate the
given surface $f$ by a conical net, \edit{obtained by computing principal curvature directions on $f$, extracting a quad mesh aligned with these directions, and optimizing it to satisfy the conical net constraints \citep{liu+2006},} and then modify it to an L-net by replacing the vertices with spheres that are close to the prescribed spheres $s_f$ tangent to $f$. %(Fig.~\ref{fig:L-examples}).

This solution is very restricted. We will see that there is a much larger variety of solutions that may be preferred in a given application. 

\paragraph{\edit{L-nets in the cyclographic model}} 
In the cyclographic model, an L-net is a quad mesh with
\edit{square grid} combinatorics in $\mathbb{R}^{3,1}$ %, \edit{matching the combinatorics of} to the one in the \edit{Euclidean model of Laguerre geometry} 
(see Fig.~\ref{fig:defL-net} to the right). Its vertices $S_{ij}$ are the Minkowski lifts of the or-spheres $s_{ij}$. 
Its edges are segments of space-like lines and represent the edge cones.
Its faces $P_{ij}$
are not necessarily planar quadrilaterals $S_{i-1,j-1}S_{i,j-1}S_{ij}S_{i-1,j}$, lying 
in the isotropic hyperplanes corresponding to the planes $p_{ij}$. % (with $0<i<m$ and $0<j<n$). 
\edit{This leads to the following definition; the Minkowski lift then induces a one-to-one correspondence between L-nets in $\mathbb{R}^3$ and $\mathbb{R}^{3,1}$.}
%\AR{This description gives a one-to-one correspondence between L-nets in $\mathbb{R}^3$ and quad meshes with square grid combinatorics in $\mathbb{R}^{3,1}$ whose faces lie in isotropic hyperplanes, via the Minkowski lift.}

\edit{
\begin{defn}[L-net] Let $\mathrm{Graff}_1^+(\R^ {3,1})$ be the set of space-like lines in $\R^ {3,1}$. %\edit{Let $V$, $E$, and $F$ be the sets of vertices, edges, and bounded faces of the $m \times n$ square grid, respectively.}
An \textbf{L-net} in $\mathbb{R}^{3,1}$ is a triple of maps
\[
    \edit{S : V \to \mathbb{R}^{3,1}}, \ \
    \edit{C : E \to \mathrm{Graff}_1^+(\R^ {3,1})}, \ \
    \edit{P : F \to \Gamma(\R^ {3,1})},
\]
preserving incidences, i.e., \edit{for any $v \in V$, $e \in E$, and $f \in F$, the inclusion $v\subset e$ implies $S(v)\subset C (e)$ and $e\subset f$ implies $C(e)\subset P (f)$}.
\end{defn}

}

%Clearly, the Minkowski lift induces a one-to-one correspondence between L-nets in $\mathbb{R}^3$ and $\mathbb{R}^{3,1}$.}

\paragraph{L-net as discrete surface and attached sphere congruence} 
We conclude this section by sketching \edit{the} informal meaning of the notion of an L-net.
Namely, an L-net can be viewed %shows that it is 
as a discrete version of
a 2-dimensional set of or-planes that may be seen as oriented tangent
planes of a surface $f(u,v)$ in $\R^3$. Moreover, it is also a discrete version of a 2-dimensional set of or-spheres $s(u,v)$ that are tangent to the surface $f$.
Such a set of spheres is often called a %\emph{
sphere congruence. %}. 
Hence,
we have a discrete dual (i.e., plane based) representation of a surface $f$ and
a sphere congruence $s$ that is tangent to $f$, in other words, %. We will speak of a 
%\emph{
sphere congruence $s$ attached to the surface $f$. \edit{It is important to note that an L-net is therefore not a discretization of a general sphere congruence or of a general tangent plane congruence, but rather the special case in which the two are coupled through this oriented contact, as developed in the following sections.}

%In other words,}
Moving to the cyclographic model, we have a discrete 2-dimensional surface $S$ and
an attached set of isotropic hyperplanes. \edit{Note that the role of the surface and the attached sphere/hyperplane congruence has now been swapped.} These hyperplanes define one ``discrete
envelope'' of the discrete sphere congruence. The other ``discrete
envelope'' is in general not represented as an L-net. The smooth counterpart
is a surface $S(u,v)$ and an attached set of tangent isotropic hyperplanes 
$P(u,v)$ that define one envelope of the sphere congruence $s(u,v)$ in $\R^3$.
This smooth analog will be discussed in the next two sections in
detail, as it provides the key for the approximation of a given surface
$f$ by an L-net whose or-spheres are close to a given sphere congruence $s(u,v)$ attached
to~$f$. 
 
%%%%%%%%%%%%%%%%%%%%%%%%%%%%%%%%%%%%%%%%%%%%%%%%%%%%%%%%%%%%%%%%%
\section{%Smooth limit of discrete L-nets and L-conjugacy \MS{
Smooth L-nets and L-conjugacy} \label{sec:limit}
%%%%%%%%%%%%%%%%%%%%%%%%%%%%%%%%%%%%%%%%%%%%%%%%%%%%%%%%%%%%%%%%%%%

We are now ready to present the main concept of our paper, namely, \edit{smooth}
%limits of
L-nets. They are L-conjugate
nets on a surface with respect to an attached sphere congruence. Since the
appearance of L-nets in the cyclographic model is simpler and quite close
to the familiar Q-nets, we prefer to start our study there (Section~\ref{ssec:L-conj-cyclo})
and then turn to the \edit{Euclidean model of Laguerre geometry} in Section~\ref{ssec:L-conj}.

\subsection{L-conjugacy in the cyclographic model} \label{ssec:L-conj-cyclo}
%---------------------------------------------------------------------------

Recall that an L-net in $ \R^{3,1}$ is a quad mesh $S$ %with vertices $S_{ij}$ 
whose faces lie in isotropic hyperplanes $P_{ij}$. The smooth analog is a surface $S(u,v)$ with an attached family
of tangent isotropic hyperplanes $P(u,v)$. As analogs of faces in $S$, hyperplanes
$P(u,v)$ contain the partial derivatives $S_u, S_v, S_{uv}$. The situation is similar to Q-nets and their smooth analogs, \edit{where conjugacy requires $S_{uv}$ to lie in the tangent plane spanned by $S_u, S_v$; analogously, L-conjugacy requires $S_{uv}$ to lie in the attached isotropic hyperplane $P(u,v)$.} This motivates the following definition. %\MS{Maybe we switch to a rectangle instead of $\R^2$}
\edit{Recall that $\Gamma(\R^{3,1})$ is the set of all isotropic hyperplanes in $\R^{3,1}$.}

\begin{defn}[L-conjugate net] 
    Throughout, $S\colon \R^2 \to \R^{3,1}$ is a \textbf{regular smooth 
%\edit{spacelike} 
\edit{space-like or light-like}
net}, that is, a smooth map such that %$S_u(u,v)\not\,\parallel S_v(u,v)$
\edit{$S_u(u,v) $ and $S_v(u,v)$ 
span a 2-dimensional space-like or light-like plane} at each point $(u,v) \in \R^2$. We equip it with an \textbf{attached tangent isotropic hyperplane congruence} $P\colon\R^2 \to \Gamma(\R^{3,1})$,
%\edit{, where $\Gamma(\R^{3,1})$ denotes the set of all isotropic hyperplanes in $\R^{3,1}$,}
that is, a smooth map such that 
    the %isotropic 
    hyperplane $P(u,v)\ni S(u,v)$ and $P(u,v)\parallel S_u(u,v), S_v(u,v)$ %$ is parallel to both $S_u(u,v)$ and $S_v(u,v)$ 
    at each point $(u,v) \in \R^2$.
    
    The net $S$ is called an \textbf{L-conjugate net} in $\R^{3,1}$ with respect to the attached congruence $P$ if $S_u, S_v, S_{uv}$ are parallel to %the isotropic hyperplane 
    $P(u,v)$ at each point $(u,v) \in \R^2$.
    % 
    %A map $S: \R^2 \to \R^{3,1}$ is called an \emph{L-conjugate net} in $\R^{3,1}$ with respect to 
    %an attached set of tangent isotropic hyperplanes $P$
    %if at every $(u,v) \in \R^2$, vectors $S_u, S_v, S_{uv}$ are parallel to the isotropic hyperplane $P(u,v)$. 
\end{defn}

\begin{remark}

\edit{Generically, the isotropic hyperplane $P(u,v)$ is uniquely determined by 
$S_u, S_v, S_{uv}$, except for the special case where $S_u, S_v, S_{uv}$ are 
linearly dependent and their span reduces to the tangent plane at $S(u,v)$,
when it is precisely a conjugate net in the usual sense. In the case where the tangent 
plane is space-like, there are exactly two isotropic hyperplanes 
%$P(u,v), P'(u,v)$ 
through it. As we shall see 
in a subsequent publication, %this describes the case where
then both envelopes of the sphere congruence in the Euclidean model of Laguerre 
geometry are represented by L-conjugate nets.}
\end{remark}

%In our general set up, the net $S$ is allowed to have isotropic tangent planes, but for the approximation by discrete L-nets, they should be excluded to avoid singularities.

The net $S$ takes a sufficiently small neighborhood of each point $(u,v)\in\mathbb{R}^2$ to a smooth surface in $\R^{3,1}$, still denoted by $S$ (by a slight abuse of notation).
For our application, it is important to define L-conjugate tangent directions of the
surface $S$ and not just an entire L-conjugate parameterization. Moreover, we need
to be able to compute L-conjugate directions from any parameterization of $S$.
We therefore \edit{arrive at the following ``self-evident''} %formulate the following 
definition.

\begin{defn}[L-conjugate tangents]
   A \textbf{reparametrization} of the pair $(S,P)$ is a pair of maps $(S\circ D,P\circ D)$ for some  diffeomorphism $D\colon \R^2 \to \R^2$. Two tangents $T_1 \ne T_2$  (or tangent vectors $T_1 \not\,\parallel T_2$) at a surface point $X=S(u,v)\in \R^{3,1}$
   % $X$ of the surface $S(\R^2) \subset \R^{3,1}$ 
   are called \textbf{L-conjugate} with respect to the attached congruence $P$ if there is a reparametrization $(\bar S,\bar P)$ of $(S,P)$
   such that $\bar  S_u\parallel T_1$, $\bar S_v\parallel T_2$, and $\bar S_u, \bar S_v, \bar S_{uv}\parallel \bar P$ at the point $X$.

\end{defn}

The remainder of this subsection aims to find a simple characterization of L-conjugate tangents in terms of an arbitrary parameterization $S$. 

First, we obtain a characterization in terms of an analog of the second fundamental form.

\begin{thm} \label{th-II} %Let $S: \R^2 \to \R^{3,1}$ be a regular smooth net.  
Two non-parallel tangent vectors $A=a_1S_u+a_2S_v$ and $B=b_1S_u+b_2S_v$ at a surface point 
$S(u,v)\in \R^{3,1}$ %of the surface $S(\R^2) \subset \R^{3,1}$ 
are L-conjugate with respect to the attached congruence $P$ if and only if
\be
\mathrm{II}_{S,P,\edit{N}}(A,B):=L_{\edit{S,P,N}}a_1b_1 + M_{\edit{S,P,N}}(a_1b_2+a_2b_1) + N_{\edit{S,P,N}} a_2b_2=0, \label{eq:L-conj4}
\ee
%with $L_P,M_P,N_P$ as in \eqref{eq:2ndformP}.
where
\begin{eqnarray} \label{eq:2ndformP}
    &L_{\edit{S,P,N}}:=&\lla S_{uu},N\rra =- \lla S_u,N_u\rra, \nonumber \\
    &M_{\edit{S,P,N}}:=& \lla S_{uv},N\rra =- \lla S_u,N_v\rra=- \lla S_v,N_u\rra, \\
    &N_{\edit{S,P,N}}:=& \lla S_{vv},N\rra =- \lla S_v,N_v\rra, \nonumber
\end{eqnarray}
and $N(u,v)$ is a normal vector to $P(u,v)$ smoothly depending on $(u,v)$.
\end{thm}

\begin{proof}
%With their isotropic normal vectors $N$, hyperplanes $P$ attached to $S$ can be written as
%
%$$ P(u,v):\ \lla X-S(u,v), N(u,v) \rra =0,$$
%
%with
\edit{By the definition of an attached tangent isotropic hyperplane congruence $P$ and a normal vector $N$, we have}
\be \lla S_u, N \rra =\lla S_v, N \rra =0. \label{eq:normalP} \ee
Differentiating the latter equation, we get $\lla S_{uu},N\rra =- \lla S_u,N_u\rra$ and the other \edit{equalities} in~\eqref{eq:2ndformP}. Thus
\begin{multline}
        \mathrm{II}_{S,P,\edit{N}}(A,B)
\edit{=-\lla S_u,N_u\rra a_1b_1-\lla S_u,N_v\rra a_1b_2 - \lla S_v,N_u\rra a_2b_1-\lla S_v,N_v\rra a_2b_2}\\
=-\lla
    a_1 S_u + a_2 S_v,
    b_1 N_u + b_2 N_v
\rra
=-\lla A,\frac{\partial N}{\partial B} \rra .
\end{multline}

(In the latter expression, by a slight abuse of notation, $N$ is viewed as a function in a vicinity of the point $S(u,v)$ on %the surface 
$S(\R^2)$ and $\frac{\partial N}{\partial B}$ is its derivative in the $B$-direction.) %of the tangent vector $B$.)
\edit{Since $\frac{\partial N}{\partial B}$ is invariant under reparametrizations,
%This shows that 
$\mathrm{II}_{S,P,\edit{N}}(A,B)$ is, that is, $\mathrm{II}_{S,P,\edit{N}}(A,B)=\mathrm{II}_{\bar S,\bar P, N}(A,B)$ for any reparametrization $(\bar S,\bar P)$.}
% independent of the parametrization of the surface $S(\R^2)$.
In particular, %for any reparametrization $(\bar S,\bar P)$ %\colon \R^2 \to S(\R^2)$ 
%such that 
\edit{if} $\bar S_u=A$ and $\bar S_v=B$, we get $$\mathrm{II}_{\bar S,\bar P, \edit{N}}(A,B)=M_{\bar S,\bar P,N}=
\lla \bar S_{uv},N\rra.$$ 
Now, if $A$ and $B$ are L-conjugate with respect to $P$, then $\bar S_{uv}\parallel \bar P$ for suitable reparametrization, hence $\mathrm{II}_{S,P,\edit{N}}(A,B)=\mathrm{II}_{\bar S,\bar P, \edit{N}}(A,B)=0$. Conversely, if $\mathrm{II}_{S,P,\edit{N}}(A,B)=0$, then for any reparametrization $(\bar S,\bar P)$ such that $\bar S_u\parallel A$ and $\bar S_v\parallel B$, we have $\bar S_{uv}\parallel \bar P$ (in particular, such a reparametrization exists), hence $A$ and $B$ are L-conjugate with respect to $P$.
%Thus, $\mathrm{II}_P(A,B)=0$ if and only if $\bar S_{uv}\parallel P$, i.e., if and only if $A$ and $B$ are L-conjugate with respect to $P$. In this case, any parametrization $\bar S\colon \R^2 \to S(\R^2)$ such that $\bar S_u\parallel A$ and $\bar S_v\parallel B$ will satisfy the property $\bar S_{uv}\parallel P$.
\end{proof}

\edit{Note that $\mathrm{II}_{S,P,N}(A,B)$ depends on the choice of the normal vector $N$, but the condition $\mathrm{II}_{S,P,N}(A,B)=0$ is invariant.}
This theorem suggests the following extension of our definition.

\begin{defn}[Self-L-conjugate tangents] 
A tangent $T$ with a direction vector $A$ is \textbf{L-conjugate} to itself %(or \emph{self-L-conjugate}) 
with respect to~$P$ if $\mathrm{II}_{S,P,\edit{N}}(A,A)=0$ \edit{for some (hence any) choice of normall vector $N(u,v)$ to $P(u,v)$}.
\end{defn}

Our next result is guided by the following \edit{well-known} construction of
conjugate surface tangents in $\R^3$: Given
a curve $c$ on a \edit{generic} surface $s$, the envelope of tangent planes of $s$ along
$c$ is a developable surface, whose rulings are conjugate
to the tangents of $c$. This \edit{can be seen, for instance,} by considering the 
limit of a strip in a Q-net.

We proceed in a very similar way and first look at the discrete L-net $S$ informally.  
Polylines $i=\mathrm{const}$ and $j=\mathrm{const}$ are called discrete \emph{parameter
lines}. Moreover, we consider \emph{strips}, namely parts of $S$ bounded by adjacent parameter lines, say $i=i_0$ and $i=i_0+1$; likewise, we have strips in the other parameter direction. The edges in which adjacent faces
of a strip join are called its \emph{transversal edges}. 
They are discrete tangents and lie in the intersection
planes of the isotropic hyperplanes through adjacent
faces. These intersection planes are discrete characteristic planes of the
discrete envelope of hyperplanes represented by the strip. 
The discrete view suggests the following concepts. %result, which we prove in a moment.

\begin{defn} 
family of hyperplanes $P\colon [0;1]\to \Gamma(\R^{3,1})$ be given by
$P(t): \lla X-S(t), N(t) \rra =0$ for some smooth $S,N\colon [0;1]\to \R^{3,1}$. The family $P$ is \textbf{regular} if $\dot {N}\not\,\parallel N$ for each $t\in  [0;1]$.  
In this case, the union of planes $\Pi(t):={P}(t)\cap \dot{P}(t)$ over $t\in  [0;1]$, where
%introduce the planes
%$$\dot{P}=\dot{P}(u(t),v(t)):\  \lla X-S, \dot{u}N_u +\dot{v}N_v \rra =0,$$
$$\dot{P}(t):\  \lla X-S, \dot {N} \rra -\lla \dot {S}, {N} \rra=0,$$
is called the \textbf{envelope of the family $P$}, and each individual plane $\Pi(t)$ is called \edit{a} \textbf{characteristic plane} of the envelope. 
\end{defn}

\begin{prop} \label{prop:Lconj4} %\MS{Some nondegeneracy assumptions are required.}
   Assume that for a regular smooth curve $C\colon [0;1] \to S(\R^{2})$, $t\mapsto C(t)={S}(u(t),v(t))$, the 
   family of hyperplanes $P(u(t),v(t))$ is regular and thus has an envelope $D$. Then, at each point $C(t)$ of $C$, any tangent line to $S(\R^{2})$ contained in the characteristic plane \edit{$\Pi(t)$} of $D$ is L-conjugate to the tangent of $C$ with respect to~$P$. 
   %%%%%%%%%%%%
   %Given a curve $C$ on a 2-dimensional surface $S \subset \R^{3,1}$ with an attached family $P$ of tangent isotropic hyperplanes, we consider the envelope $D$ of hyperplanes $P$ along $C$. Then, at each point $C(t)$ of $C$, the intersection line of the characteristic plane of $D$ with the tangent plane of $S$ at $C(t)$  is L-conjugate to the tangent of $C$ w.r.t. the hyperplanes $P$. 
\end{prop}

\begin{proof}

The characteristic planes of the envelope $D$ are the planes $P \cap \dot{P}$.
If a tangent line is contained in such a plane, the direction vector 
 $\lambda S_u + \mu S_v$ is orthogonal to $\dot{N}=\dot{u}N_u +\dot{v}N_v $,
leading to 
$$
 \lambda \dot{u}\lla S_u,N_u\rra +\lambda \dot{v}\lla S_u,N_v\rra+
 \mu\dot{u}\lla S_v,N_u\rra + \mu \dot{v}\lla S_v,N_v\rra=0.
$$
By Theorem~\ref{th-II}, the vector is L-conjugate to %the tangent of 
$\dot{C}=\dot{u}S_u +\dot{v}S_v $ with respect to~$P$.

\end{proof}

%This proves also the following characterization of L-conjugacy.

\begin{remark}
    %Our derivations did not use the 
    \edit{The above results and their proofs remain true without the assumptions} that hyperplanes $P$ and therefore
 normals $N$ are isotropic \edit{and the net is space-like or light-like.}
 %; our results remain true without those assumptions.} 
 Hence, here we have a more general concept
    of conjugacy with respect to an attached set of tangent hyperplanes $P$.
    Moreover, L-conjugacy is clearly a concept of projective geometry,
    as is true for ordinary conjugacy in surface theory. \edit{We plan to return to this generalization in future work.}
\end{remark}

\subsection{L-conjugacy in Euclidean 3-space} \label{ssec:L-conj}
%----------------------------------------------------------------

Although it would be possible to base our approximation algorithm on the cyclographic model, it 
is preferred to work directly in the design space $\R^3$. Thus, we now transfer
the results to Euclidean 3-space. This is particularly important for a better understanding
of L-conjugacy and its relation to an extended version of curvature theory of surfaces
with respect to an attached sphere congruence (see Section~\ref{sec:curvature}). We start with an informal motivation, and then give precise definitions.

L-conjugacy in $\R^{3,1}$ has been motivated by geometrically
interpreting the edges of the quad mesh $S$ in two ways: as discrete tangents of
parameter lines and as tangential components of discrete characteristic planes of 
envelopes of isotropic hyperplanes. 

\begin{figure}
\centering
 \begin{overpic}[width=1\linewidth]{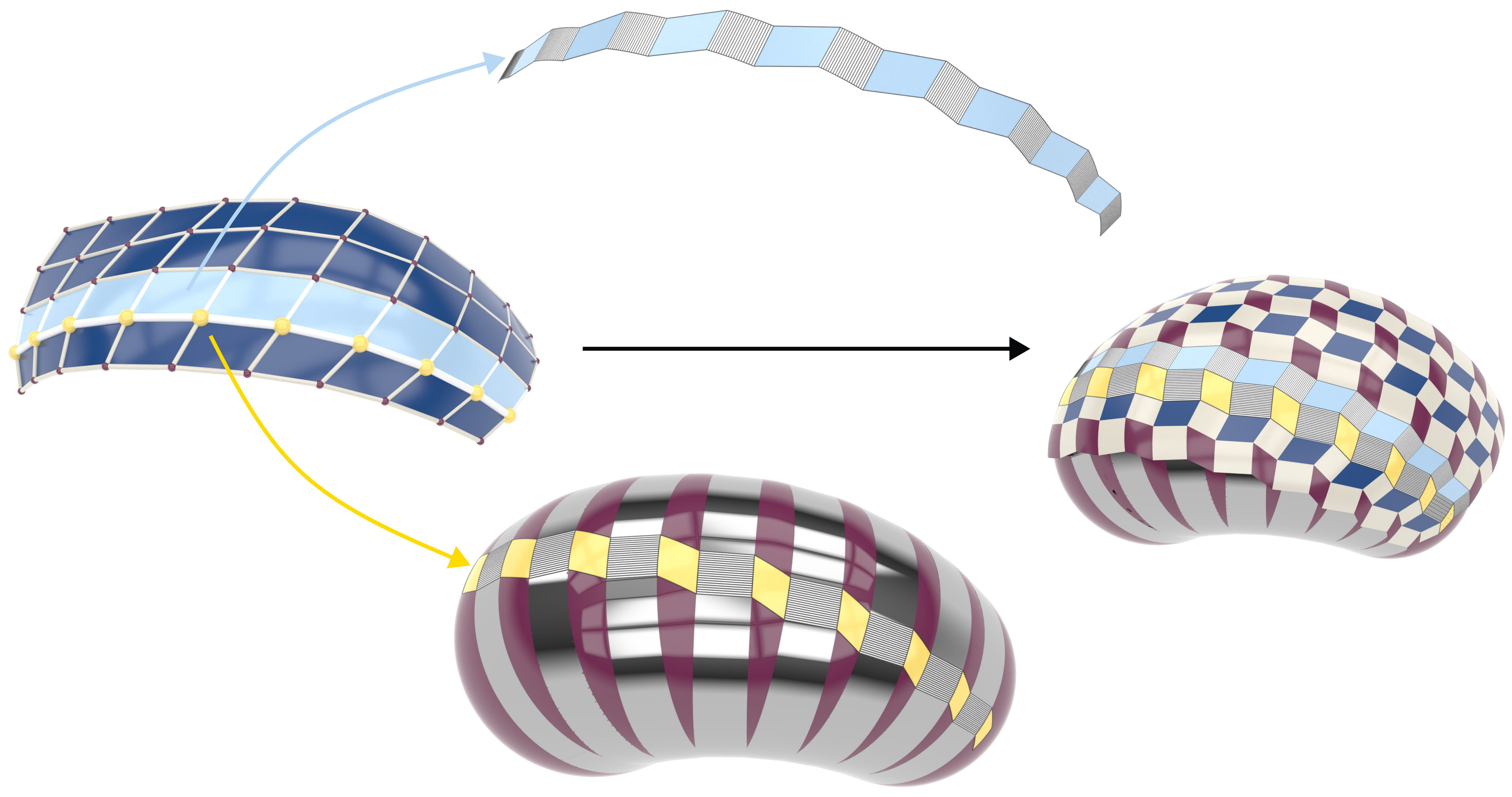}
 \put(5,42){$S \subset \R^{3,1}$}
 \put(45,22){$\Sigma \subset \R^3$}
 \put(45,45){$\Delta \subset \R^3$}
 \put(90,35){$\mathcal L \subset \R^3$}
     
 \end{overpic}
\caption{A guad mesh $S$ in $\R^{3,1}$ and its corresponding
L-net $\mathcal{L}$ in $\R^3$. A parameter line
of $S$ (yellow vertices) defines a tangent discrete channel surface $\Sigma$ formed by spheres of the attached 
congruence joined by or-cones, and a strip of $S$ (light-blue quads) defines a tangent discrete developable surface $\Delta$. }
\label{fig:conj3}
\end{figure}
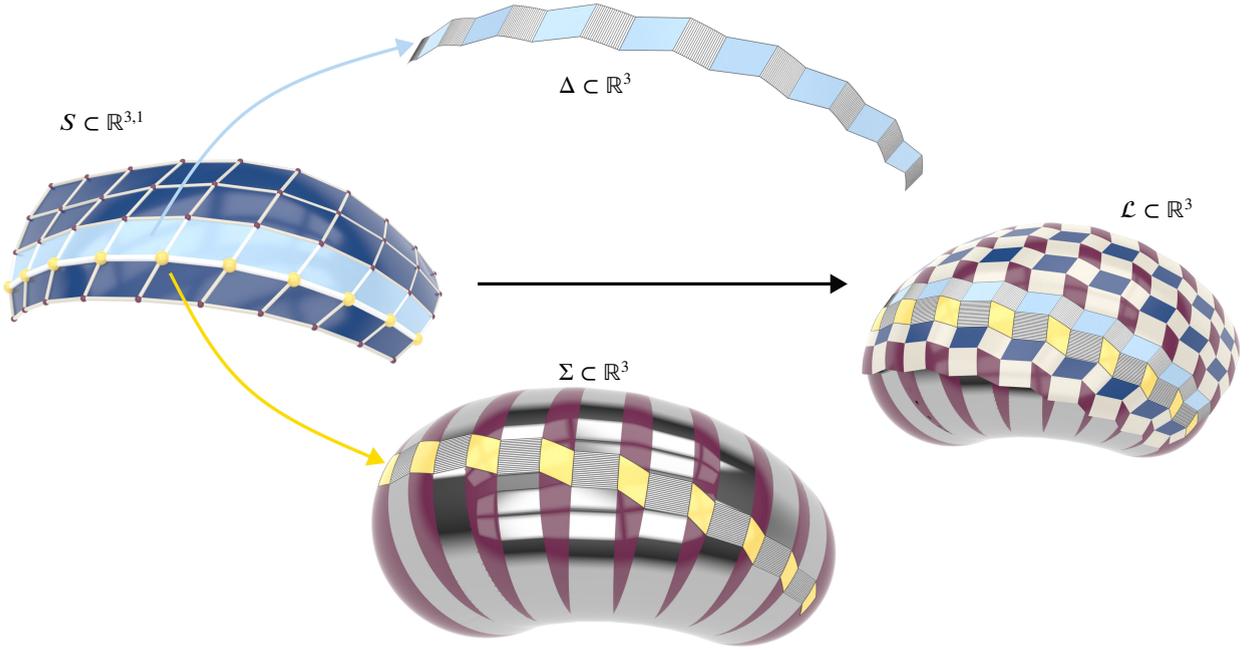

%\HP{
%We start with a discrete L-net $S \subset \R^{3,1}$. 
In $\R^3$, edges of $S$ correspond to or-cones. A discrete parameter line of $S$ corresponds
to a discrete set of or-spheres joined by or-cones (Fig.~\ref{fig:conj3}). This is a discrete version of a
channel surface that is in oriented contact with an oriented surface $f$ and enveloped by spheres taken from the
selected attached sphere congruence. In the smooth limit, the rulings of the or-cones correspond to rulings of
tangential cones of the channel surface. They are principal curvature directions orthogonal to the characteristic circles. 

Strips of $S$ represent strips of an L-net $\Lf$ which are formed by segments of or-planes and 
or-cones (Fig.~\ref{fig:conj3}). They are discrete models of tangent developable surfaces of a surface $f$, and the rulings of the cones are discrete rulings. The boundaries of the discrete developable strips exhibit, in general, a staircase behavior. However, the
discrete rulings change \edit{continuously} along the strip.
We look at the discrete rulings of the strips, which are
also cone rulings, to define L-conjugate directions
of a surface $f$ with respect to an attached sphere congruence. 
We shall see that the opposite edges in each vertex quadrilateral should be viewed as two discrete versions of the same direction, and adjacent edges should be viewed as L-conjugate directions. See Fig.~\ref{fig:alignment}.
%\MS{This does not sound very convincing at this stage.}
%In each planar face we may view opposite edges as discrete versions of the same direction, and adjacent edges as L-conjugate directions. 

Turning to the smooth setting, we define L-conjugate directions
in $\R^3$ therefore as follows.
%}

%\MS{

\begin{defn}[L-conjugate net] \label{def:L-conj3}
Throughout, $f\colon \R^2 \to \R^{3}$ is a regular smooth net (a smooth map such that $f_u(u,v)\not\,\parallel f_v(u,v)$ for each $(u,v) \in \R^2$) without parabolic and umbilic points
(so that the principal curvatures $\kappa_1$ and $\kappa_2$ are distinct and non-vanishing at each point $(u,v) \in \R^2$).
% non-vanishing Gaussian curvature
%, that is, a smooth map such that $f_u(u,v)\not\,\parallel f_v(u,v)$ and the Gaussian curvature is non-zero at each point $(u,v) \in \R^2$. 
We equip the net with an \textbf{orientation}, that is, a choice of unit normal $n\colon \R^2 \to \R^{3}$, smoothly depending on %a point 
$(u,v) \in \R^2$.
% that is, a smooth map such that $n(u,v) \perp f_u(u,v),f_v(u,v)$ and $|n(u,v)|=1$ at each point $(u,v) \in \R^2$. 
Denote by $p(u,v)$ the oriented tangent plane at a point $f(u,v)$, with the orientation given by the normal $n(u,v)$. %Hereafter, the orientation of $p(u,v)$ is given by the normal $n(u,v)$.
We also equip the net %$f$ 
with an \textbf{attached tangent sphere congruence} $s\colon\R^2 \to \mathrm{OrSpheres}(\R^{3})$, that is, a %smooth 
map such that 
%$s(u,v)\ni f(u,v)$ and 
$s(u,v)$ is in oriented contact with the oriented tangent plane $p(u,v)$ at each point $f(u,v)$
and the Minkowski lift $S(u,v)$ of $s(u,v)$ is a regular smooth \edit{space-like or light-like} net in~$\R^{3,1}$. %\MS{It would be more aesthetic to state this condition in 3D terms.}

The net $f$ is called an \textbf{L-conjugate net} in $\R^{3}$ with respect to the attached congruence $s$
if the Minkowski lift of $s$ is an L-conjugate net in $\R^{3,1}$ with respect to the Minkowski lift of $p$. In this case, the directions conjugate to $f_u(u,v)$ and $f_v(u,v)$ in the tangent plane $p(u,v)$ of $f$ are called \textbf{L-conjugate}.
\end{defn}

Beware that $f_u(u,v)$ and $f_v(u,v)$ themselves are not considered L-conjugate with respect to %the attached congruence 
$s$; only their \emph{conjugate directions}, in the ordinary sense, are L-conjugate. In other words, the rulings of the two tangent developables along the isoparameter lines $u=\mathrm{const}$
and $v=\mathrm{const}$ are L-conjugate.  Only such a definition of $L$-conjugacy leads to a \edit{Laguerre invariant} concept. Likewise, the visualization of L-conjugate rulings in the discrete model appears nicely in the edges of each planar face, while
discrete parameter lines $u=\mathrm{const}$ and $v=\mathrm{const}$ are actually not directly present in a discrete L-net.

Now, we generalize \edit{the notion of L-conjugate directions} to an arbitrary parametrization $f(u,v)$. %, not necessarily an L-conjugate net.

\begin{defn}[L-conjugate tangents] \label{def:L-conj3-tangents}
   Two tangent vectors $a \not\,\parallel b$ at a surface point $x=f(u,v)\in \R^{3}$ %of the surface $f(\R^2) \subset \R^{3}$ 
   are %called 
   \textbf{L-conjugate} with respect to the attached congruence $s$ if \edit{there is a} reparametrization $(\bar f,\bar s)=
   (f\circ D,s\circ D)$ of $(f,s)$
   such that $\bar  f_u$ is conjugate to $a$, $\bar f_v$ is conjugate to $b$, and $\bar S_u$ is L-conjugate to $\bar S_v$ with respect to $\bar P$, where $\bar S$ and $\bar P$ are the Minkowski lifts of $s\circ D$ and $p\circ D$, respectively.
\end{defn}

%\AR{\begin{defn}[L-conjugate tangents — alternative]
%Two tangent vectors $a \not\,\parallel b$ at a surface point $f(u,v)$ are \textbf{L-conjugate} with respect to the attached sphere congruence $s$ if the corresponding tangents of the Minkowski lift $S$ are L-conjugate.
%\end{defn}}

\begin{figure}
    \centering
    \begin{tabular}{l r}
    \includegraphics[width=0.4\linewidth]{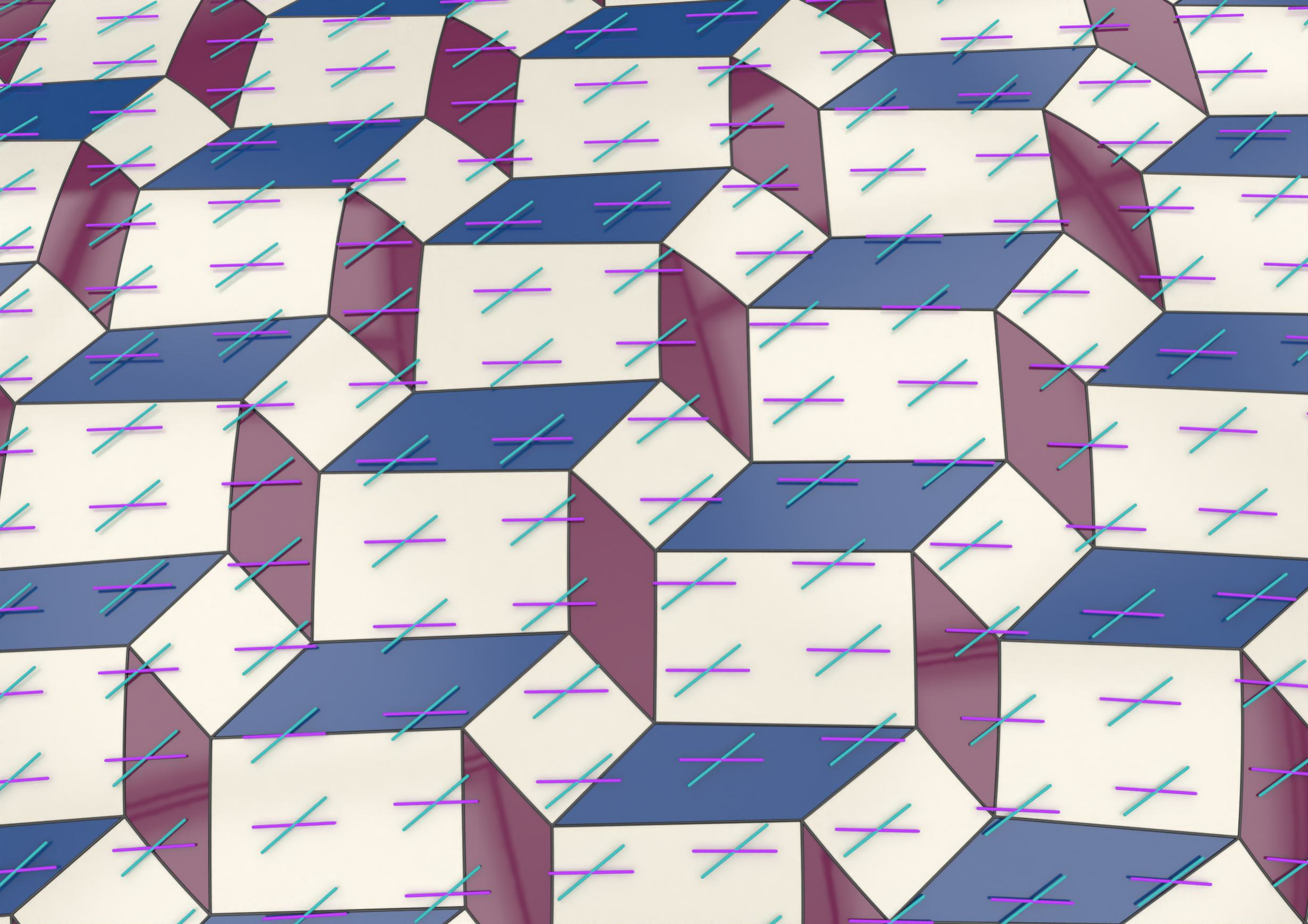}     &  
    \includegraphics[width=0.4\linewidth]{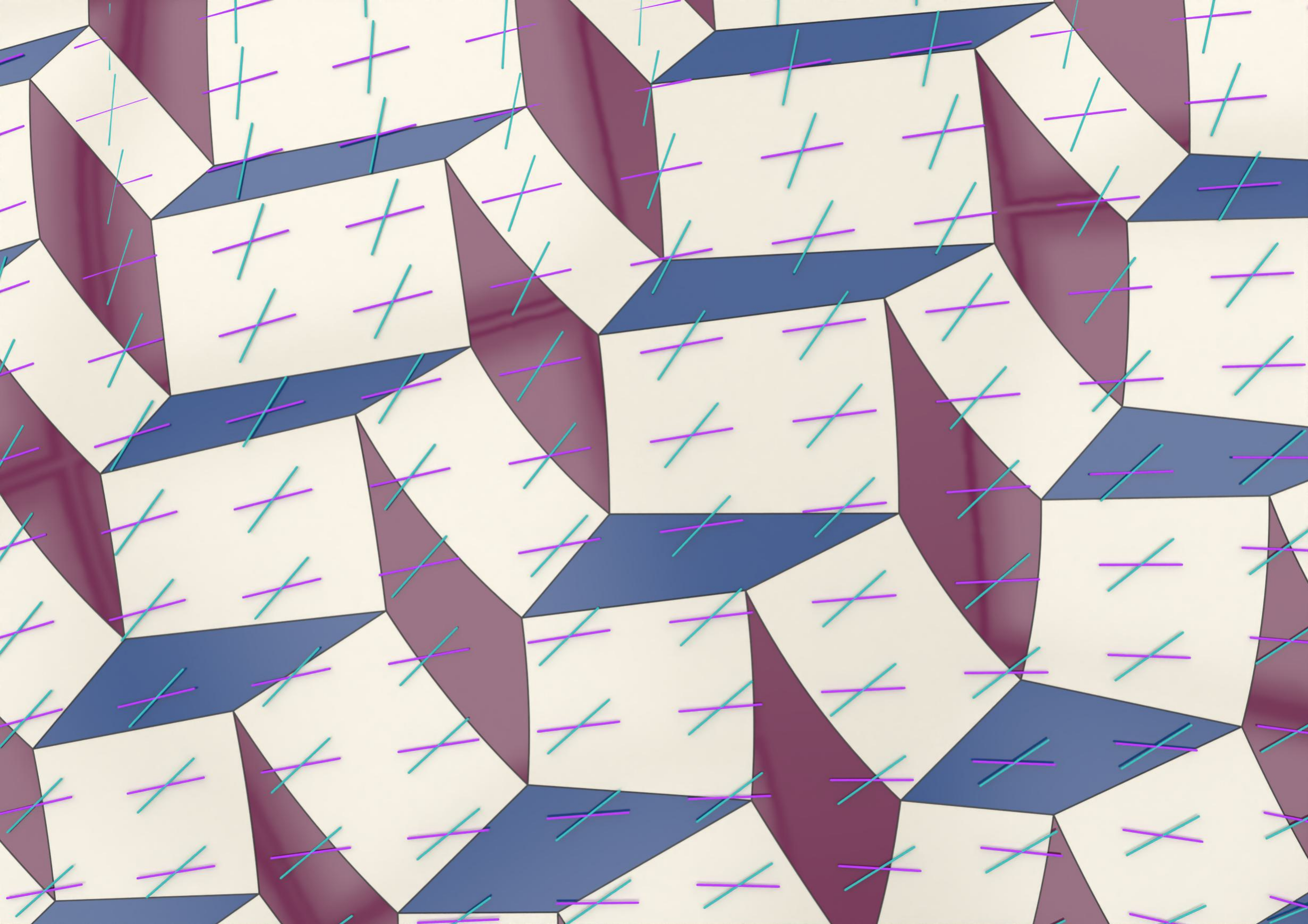}
    \end{tabular}
    
    \caption{The sides of planar quadrilaterals (blue) of an L-net approximating a surface with an attached tangent sphere congruence are aligned with L-conjugate directions (shown in cyan and magenta).}
    \label{fig:alignment}
\end{figure}
%}

We now turn to the analytical formulation of L-conjugacy and prefer to work in a local principal
parameterization $f(u,v)$ of the surface. % $f(\R^2)$. %\MS{This breaks down at an umbilic point.} 
This is sufficient for the geometric understanding
and for the application to approximation. %Let $n$ be the unit normals of $f$ defining its orientation. 
With $\kappa_1,\kappa_2$ as principal curvatures of $f$, our parameterization
satisfies
$$ n_u=-\kappa_1f_u,\ n_v= -\kappa_2 f_v,\ \langle f_u,f_v \rangle=0.  $$
We are going to consider a single surface point, and assume that at this point,
%Since we are only interested in a single surface point, we can even assume that there \MS{Do we actually need this assumption? If yes, then it should be in the theorem.}
%
$$ \langle f_u,f_u \rangle = \langle f_v,f_v \rangle=1.$$
Let the tangent or-sphere $s(u,v)$ have signed radius $r(u,v)$ and center $c_s(u,v)=f+r n$.
This defines $S$ and $N$ for the representation in $\R^{3,1}$,
$$    S=(c_s,r)=(f+r n,r), \ N=(n,1).$$
Inserting into \eqref{eq:2ndformP}, we obtain
\begin{eqnarray}  \label{eq:2ndLform3}
    &L_{\edit{S,P,N}}&=  -\langle f_u, n_u \rangle - r \langle n_u,n_u \rangle = \kappa_1 -r \kappa_1^2 ,  \nonumber \\
    &M_{\edit{S,P,N}}&=  -\langle f_u, n_v \rangle - r \langle n_u,n_v \rangle =0, \\
    &N_{\edit{S,P,N}}&=  -\langle f_v, n_v \rangle - r \langle n_v, n_v \rangle = \kappa_2 -r \kappa_2^2. \nonumber 
\end{eqnarray}
Inserting  \eqref{eq:2ndLform3} into \eqref{eq:L-conj4}, 
we get the condition
\be (\kappa_1-r \kappa_1^2)a_1b_1 +  (\kappa_2-r \kappa_2^2)a_2b_2=0 \label{eq:Lconjcontact} \ee
on two tangent vectors $a_1S_u+a_2S_v$ and $b_1S_u+b_2S_v$ in $\R^{3,1}$ to be L-conjugate.

We now want to express L-conjugacy in the
tangent space of $f$, spanned by $f_u,f_v$, following Definition~\ref{def:L-conj3-tangents}.
By this definition, \eqref{eq:Lconjcontact} is equivalent to L-conjugacy of two vectors $\bar{a}=\bar{a}_1f_u+\bar{a}_2f_v$ and $\bar{b}=\bar{b}_1f_u+\bar{b}_2f_v$ conjugate to $a={a}_1f_u+{a}_2f_v$ and $b={b}_1f_u+{b}_2f_v$, respectively (because $(S_u,S_v)$ and $(f_u,f_v)$ are transformed in the same way under any reparametrization of $(f,s)$).
The vector $\bar{a}$ is related to $a$ via ordinary conjugacy, \edit{which is equivalent to vanishing second fundamental form:}
 \be \kappa_1 a_1 \bar{a}_1 +  \kappa_2 a_2 \bar{a}_2 =0. \label{eq:conjf}
 \ee
%\AR{The vectors $a$ and $\bar{a}$ are related by ordinary conjugacy, meaning that the second fundamental form vanishes on them: $\mathrm{II}(a, \bar{a}) = \kappa_1 a_1\bar{a}_1 + \kappa_2 a_2\bar{a}_2 = 0$ in principal curvature coordinates.}
Rewriting~\eqref{eq:Lconjcontact} via \eqref{eq:conjf}, we get the following expression of L-conjugacy \edit{in terms of the \emph{signed principal curvature radii} $\rho_1:=\kappa_1^{-1}$ and $\rho_2:=\kappa_2^{-1}$ of the oriented smooth net $f$ at the point $f(u,v)$}.

\begin{thm} \label{thm:L-conj3-principal}
Consider %an oriented surface $f$ with an attached tangent or-sphere congruence $s$, and 
an orthonormal principal frame $(t_1,t_2,n)$ at a surface point $f(u,v)$. 
Then, two non-parallel tangent vectors $\bar{a}=\bar{a}_1t_1+\bar{a}_2t_2$ and $\bar{b}=\bar{b}_1t_1+\bar{b}_2t_2$ at $f(u,v)$ are L-conjugate with respect to the congruence $s$, if and only if
\be (\edit{\rho_2} -r) \bar{a}_1\bar{b}_1  + (\edit{\rho_1} -r)  \bar{a}_2\bar{b}_2 =0,
\label{eq:conjcurv}\ee
where $r$ is the signed radius of the or-sphere $s(u,v)$ and  \edit{$\rho_1$ and $\edit{\rho_2}$ are the signed principal curvature radii} of~$f$ at~$f(u,v)$. 
\end{thm}

\edit{We observe that for $r=0$, L-conjugacy becomes ordinary conjugacy, and for $r\to \infty$, it becomes orthogonality.}

This theorem suggests the following extension of our definition.

\begin{defn}[Self-L-conjugate tangents] \label{def:self-conj} \edit{With} the notation of Theorem~\ref{thm:L-conj3-principal},
the tangent vector $\bar a$ is called \textbf{L-conjugate} to itself %(or \emph{self-L-conjugate}) 
with respect to~$s$ if 
\edit{$(\edit{\rho_2} -r) \bar{a}_1^2  + (\edit{\rho_1} -r)  \bar{a}_2^2 =0$.}
\end{defn}

Equation \eqref{eq:conjcurv} implies %immediately proves 
the following properties of the principal
curvature directions in the context of L-conjugacy.

\begin{cor} %\MS{This breaks down at an umbilic point.}
The principal tangents at a surface point $f(u,v)$ are the only orthogonal L-conjugate tangents with respect to $s(u,v)$. They are also the only ones that are simultaneously L-conjugate with respect to two attached sphere congruences that have different sphere radii at the point. 
\end{cor}

\edit{\begin{remark} Recall that in the previous definitions, umbilic points were excluded. However, the notion of L-conjugacy can be easily extended to this case via~\eqref{eq:conjcurv}. In this case, \eqref{eq:conjcurv} reduces to $(\rho - r)(\bar{a}_1\bar{b}_1 + \bar{a}_2\bar{b}_2) = 0$, where $\rho:={\rho_1} = \edit{\rho_2}$. 
%When $\edit{\rho_1} = \rho_2 = \rho$, the principal curvature radii coincide $r_1 = r_2 = \rho$, and equation~\eqref{eq:conjcurv} reduces to $(\rho - r)(\bar{a}_1\bar{b}_1 + \bar{a}_2\bar{b}_2) = 0$. 
If $r \neq \rho$, this forces $\bar{a} \perp \bar{b}$, so every pair of orthogonal directions is L-conjugate. If $r = \rho$, the equation is satisfied for all directions, meaning every direction is L-conjugate to every other.
\end{remark}}

\begin{remark} \label{rem:channel}
    Equation \eqref{eq:Lconjcontact} relates the tangents of contact curves of iso-parameter channel
    surfaces %(to $u=const$ or $v=const$) 
    or equivalently iso-parameter tangent developable surfaces with the base surface $f$ in an L-conjugate parametrization $f(\lambda,\mu)$, i.e., the smooth limit
    of an L-net approximating $f$. This is useful for the initialization of
    an approximating L-net, if one wants to base it on an initial selection of the discrete contact curves; see Sec.~\ref{sec:approx}.
    \end{remark}

%\HP{We may add an analytical proof of Prop.~\ref{prop:Lconj3}, if co-authors find the
%current geometric derivation via a limit process of the discrete case insufficient.
%I avoided it, since the simplest derivation I know of uses conjugacy in the sense
%of Moebius geometry, that is planned for a paper with TU Wien}

We conclude this section with the following informal description of L-conjugate
tangents (for which we do not specify the required non-degeneracy assumptions):
%
%Since cone rulings in a discrete L-net have a double meaning as explained above \MS{We need to be more specific here}, we also have the following construction of L-conjugate tangents. 
%
%\begin{prop} \label{prop:Lconj3} \MS{Nondegeneracy assumptions. A proof of the proposition is also needed.}
   Given a curve $c$ on the surface $f$, %an \MS{oriented} surface 
   %$f(\R^2) \subset \R^3$, % with an attached family of tangent or-spheres,
  consider the developable surface $\Delta$ defined as envelope of tangent or-planes of $f$ along $c$,  and the channel surface $\Sigma$ defined as envelope of
  or-spheres of $s$ along $c$ \edit{(Fig.~\ref{fig:conj3})}. Then, at each point %$c(t)$ 
  of $c$, the ruling of $\Delta$ is
   L-conjugate to that principal tangent of $\Sigma$ which is orthogonal to its 
   characteristic circle. \edit{Cf.~Fig.~\ref{fig:dualcurv}}.
%\end{prop}

\edit{A side remark is that the surface $f$ together with its unit normal field defines a Legendre immersion \citep{musso2018}; thus, our setup corresponds to a Legendre immersion with a choice of sphere at each contact element, given by the attached sphere congruence.}

%%%%%%%%%%%%%%%%%%%%%%%%%%%%%%%%%%%%%%%%%%%%%%%%%%%%%%%%%%%%%%%%%%%%%%%%%%%%%%
\section{L-conjugacy and dual curvature radii of a surface with respect to a tangent sphere congruence} \label{sec:curvature}
%%%%%%%%%%%%%%%%%%%%%%%%%%%%%%%%%%%%%%%%%%%%%%%%%%%%%%%%%%%%%%%%%%%%%%%%%%%%%%

In this section, we discuss informally the close relation of the L-conjugacy relation \eqref{eq:conjcurv}
to a dual viewpoint of surface theory that is based on oriented tangent planes rather than
points. In particular, there is a relation to dual curvature radii which are a dual counterpart
to normal curvatures. We also discuss the various instances of self-L-conjugate (L-asymptotic) directions (Sec.~\ref{ssec:dualcurv}). This is
relevant for approximation, since we need to avoid L-asymptotic directions. 
Despite that, we briefly study L-asymptotic parameterizations in Sec.~\ref{ssec:L-asymp}, mainly
due to the appearance of a Laguerre-geometric formulation of principal symmetric nets introduced in \citep{s-nets-2020}. Finally, we address important
special cases of attached sphere congruences (Sec.~\ref{ssec:special}) that are used in Sec.~\ref{sec:approx}, where we \edit{then} discuss
approximation with L-nets.

%---------------------------------------------------------------------------
\subsection{Classification of contact elements of a surface with
respect to an attached sphere congruence} \label{ssec:dualcurv}
%-----------------------------------------------------------------------

\paragraph{Dual curvature radii} Let us briefly summarize some basic facts on the
mentioned dual viewpoint of surface theory in Euclidean 3-space. We view a surface $f(u,v)$ as
set of oriented tangent planes $p(u,v)$ and instead of curves on $f$ we consider developable
surfaces enveloped by one-parameter families of tangent planes $p$. At a given surface point $f(u,v)$ we consider a tangent $e$ and the cylinder $C(e)$ enveloped by those or-planes $p$ of $f$ that are parallel to $e$. In other words, $C(e)$ has rulings parallel
to $e$ and is tangent to $f$ along a curve (see Fig.~\ref{fig:dualcurv}). The dual curvature radius
$\rho^*(e)$ to the ruling direction $e$ is defined as curvature radius (inverse normal curvature)
of the cylinder $C(e)$ in direction orthogonal to $e$. It is the signed radius of a right 
circular cylinder that is in 2nd order contact with $C(e)$ along the ruling~$e$. 

\begin{figure}
     \begin{overpic}[width=0.45\linewidth]{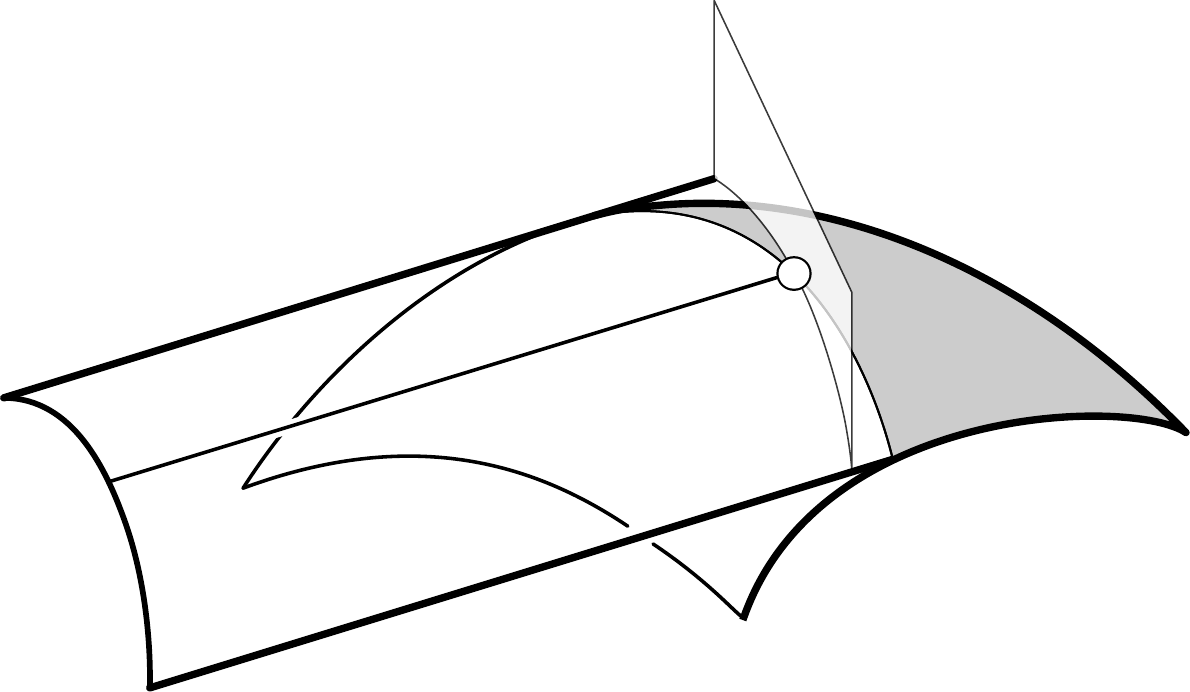}
    \put(48,25){\footnotesize\contour{white}{$e$}}
    \put(66,22){\footnotesize\contour{white}{$c$}}
    \put(61,45){\footnotesize\contour{white}{$\sigma$}}
    \put(13,8){\footnotesize\contour{white}{$C(e)$}}
  \end{overpic}
  \hfill
  \begin{overpic}[width=.5\linewidth]{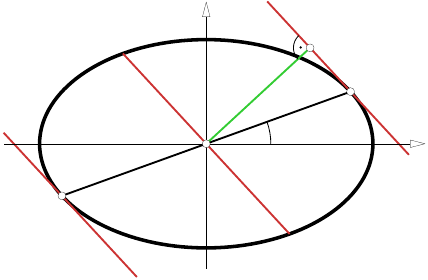}
    \put(65,32){\scriptsize\rotatebox{21}{$|\kappa_n|^{-1/2}$}}
    \put(51,40){\scriptsize\rotatebox{41}{$\sqrt{|\rho^*|}$}}
    \put(23,33){\scriptsize\rotatebox{0}{$|\kappa_1|^{-1/2}$}}
    \put(42,38){\scriptsize\rotatebox{90}{$|\kappa_2|^{-1/2}$}}
    \put(18,17){\scriptsize\rotatebox{21}{contact curve tgt.}}
    \put(52,29){\scriptsize\rotatebox{-47}{ruling direction}}
    \put(59,32){\scriptsize\rotatebox{0}{$\phi$}}
    \put(100,34){\small\rotatebox{0}{$t_1$}}
    \put(42,63){\small\rotatebox{0}{$t_2$}}
  \end{overpic}
\caption{Left: A cylinder $C(e)$ which is tangent to a given surface along a curve defines along its ruling $e$ a dual curvature radius $\rho^*$. It is the signed curvature
radius of the intersection curve $c$ with a plane  $\sigma$ orthogonal
to $e$. Right: The Dupin indicatrix at a surface point visualizes the
normal curvatures as distances $|\kappa_n|^{-1/2}$ of its points from the
center (contact point), and dual curvature radii via distances $\sqrt{|\rho^*|}$ of its tangents from the center. The parallel tangents are the
Dupin indicatrix of the  right circular cylinder that is in 2nd order contact
with the cylinder $C(e)$ along $e$. Images taken from \citep{s-nets-2020}. 
}
\label{fig:dualcurv}
\end{figure}
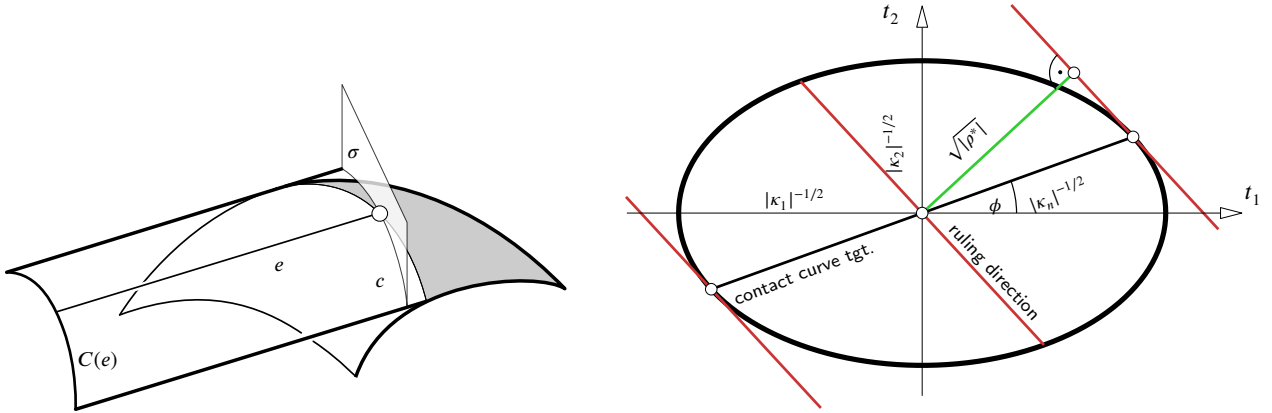

Fig.~\ref{fig:dualcurv} illustrates this at hand of the Dupin indicatrix at $f(u,v)$, which visualizes  normal curvatures $\kappa_n(\phi)$ for tangents that form the angle $\phi$ with the first principal tangent $t_1$. It is a radial diagram of
$|\kappa_n(\phi)|^{-1/2}$ and also visualizes dual curvature radii as explained in Fig.~\ref{fig:dualcurv}. 

Normal curvatures $\kappa_n(\phi)$ at a surface point satisfy Euler's formula
$$ \kappa_n(\phi)=\kappa_1\cos^2\phi +\kappa_2 \sin^2\phi.$$
Its dual
%The dual counterpart 
is Blaschke's formula for the dual curvature radius to the ruling \edit{direction} $t_1\cos \edit{\phi_e}  + t_2\sin \edit{\phi_e}$,
%$\phi$,
%
\be \rho^*(\phi_e)=\edit{\rho_2}\cos^2 \edit{\phi_e} + \edit{\rho_1}\sin^2 \edit{\phi_e}, \label{eq:blaschke} \ee
\edit{where $\rho_1:=\kappa_1^{-1}$ and $\rho_2:=\kappa_2^{-1}$ are the signed principal curvature radii of $f$ \edit{\cite[Proposition~2.8]{s-nets-2020}}.}

\paragraph{Self-L-conjugate tangents} According to Definition~\ref{def:self-conj}, %\eqref{eq:conjcurv}, 
a tangent vector
$\bar{a}=t_1\cos \edit{\phi_e}  + t_2\sin \edit{\phi_e}$ is
L-conjugate to itself with respect to an attached tangent sphere congruence $s$, if and only if
\be  (\edit{\rho_2}-r) \cos^2 \edit{\phi_e}  + (\edit{\rho_1}-r) \sin^2 \edit{\phi_e} =0, \label{eq:selfL}
\ee 
which, by \eqref{eq:blaschke}, expresses a vanishing difference $\rho^*(\edit{\phi_e})-r$ of dual curvature radius and sphere radius. These tangents are called \emph{L-asymptotic with respect to $s$}. They are real for $(\edit{\rho_1}-r)(\edit{\rho_2}-r)<0$ and symmetric
with respect to the principal directions, with which they form the angle $\phi_a$ given by
\be \tan \phi_a=\pm \sqrt{-\frac{\edit{\rho_2}-r}{\edit{\rho_1}-r}}. \label{eq:Lasymp} \ee
L-asymptotic tangents are fixed elements in the involutory projective map between L-conjugate
tangents at a surface point. 

\paragraph{L-characteristic tangents} L-conjugate tangents that are \edit{reflection-}symmetric %in the sense of reflection}
with respect to \edit{the} principal directions are called \emph{L-characteristic
tangents}. Their angles $ \pm \phi_c$ against the
first principal direction are computed via
\be \tan \phi_c=%\pm \sqrt{\frac{ \rho_{s,1}^*   }{\rho_{s,2}^*}}=
\pm \sqrt{\frac{\edit{\rho_2}-r}{\edit{\rho_1}-r}},\label{eq:L-char} \ee
where we applied~\eqref{eq:conjcurv} for $(a_1,a_2)=(b_1,-b_2)=(\cos\phi_c,\sin\phi_c)$. These directions are real for $(\edit{\rho_1}-r)(\edit{\rho_2}-r)>0$.

Just as we distinguish between
hyperbolic, parabolic, elliptic, and flat points in elementary surface theory, we have an analogous classification here. 
%
%It becomes even more obvious if 
To make it even more transparent, motivated by \eqref{eq:selfL}, we define \emph{dual curvature radii $\rho_s^*$ of a surface $f$ with respect to an attached sphere congruence $s$} as
\be    \rho_s^*:=\rho^*-r.      \label{eq:defduals}\ee
Clearly, these dual curvature radii with respect to  $s$ satisfy a generalized Blaschke formula
%in which we prefer to use the signed curvature radii $\rho_i=\kappa_i^{-1}$ of $f$, 
%
\be  \rho_s^*(\edit{\phi_e} )=(\rho_2-r)\cos^2 \edit{\phi_e} + (\rho_1-r)\sin^2 \edit{\phi_e}=
\rho_{s,1}^* \cos^2 \edit{\phi_e} + \rho_{s,2}^*\sin^2 \edit{\phi_e}. \label{eq:genblaschke} \ee
Note that the dual curvature radii $\rho_{s,1}^*, \rho_{s,2}^*$ with respect to  $s$ for rulings in 
first and second principal directions are 
$$ \rho_{s,1}^*=\rho_2-r,\  \rho_{s,2}^*=\rho_1-r.$$
Introduce the value
$$\Lambda:=\rho_{s,1}^*\rho_{s,2}^*.$$
%\AR{Since Laguerre transformations map contact elements to contact elements, 
It is more \edit{natural} to speak of \emph{contact elements} of the oriented surface (that is, pairs $(f(u,v),p(u,v))$ consisting of a surface point and the oriented tangent plane at the point), rather than surface points, \edit{because Laguerre transformations naturally act on contact elements rather than points.} 

We can now formulate the classification of surface contact elements with respect to an attached sphere congruence:
\begin{description}
\item[L-hyperbolic contact element ($\Lambda < 0$):] two real L-asymptotic tangents
are characterized by~\eqref{eq:Lasymp}.
\item[L-parabolic contact element ($\Lambda=0$ and exactly one
%principal curvature radius 
$\rho_{s,i}^* \ne 0$):] a single L-asymptotic tangent that is also principal; any tangent direction is L-conjugate to that direction. 
\item[L-elliptic contact element ($\Lambda>0$):] no real L-asymptotic tangent, but a pair of
principal symmetric and L-conjugate tangents given by (\ref{eq:L-char}).
\item[L-flat contact element ($\rho_{s,1}^*=\rho_{s,2}^*=0$):] the congruence sphere is in second order contact with $f$. This is only possible at umbilic points, which were excluded in our previous discussion; however, the notion of L-conjugacy can be easily extended to this case via~\eqref{eq:conjcurv}. 
\end{description}

%---------------------------------------------------------------------------
\subsection{L-asymptotic developables and L-asymptotic parameterizations} \label{ssec:L-asymp}
%--------------------------------------------------------------------------

Since points
are not preserved under Laguerre transformations, the same is true for curves
on surfaces. Instead of integrating special directions such as principal directions
or asymptotic directions to curves and further to isoparameter curves of special
parameterizations, we now have to view tangents as rulings of developable
surfaces that are tangent to $f$. We call them  \emph{tangent developables}. \edit{Recall that $f$ has no parabolic and umbilic points.}

For principal tangents, this is simple. A principal developable is just a developable surface that is tangent to $f$ along
a principal curvature line. 
If the curvature line follows the first principal
directions, the rulings of the principal developable are orthogonal to them and
thus parallel to the second principal direction. The arising principal parameterizations
constitute orthogonal smooth L-conjugate nets.

We now focus on L-asymptotic tangents (rulings) $e_a$  and integrate them to \emph{L-asymptotic 
developables}. The tangents of their contact curves with $f$ are conjugate to
the rulings $e_a$ in the usual sense, given by equation \eqref{eq:conjf}. At
each ruling of an L-asymptotic developable, the radius $r$ of the sphere of the congruence equals the dual curvature radius
$\rho^*$ of the ruling $e_a$. This sphere of radius $\rho^*$ is known as \emph{Mannheim
sphere} \citep{blaschke+1973} and is a \edit{dual} counterpart to the well-known Meusnier sphere \citep{blaschke+1973}. While the latter contains the osculating
circles of all surface curves through a given line element, the Mannheim sphere is
enveloped by all oriented rotational cones $\Gamma(e_a)$ which share the ruling $e_a$, have their vertex $v_c$ on 
$e_a$ and are in second order plane contact with $f$. Such an osculating cone $\Gamma(e_a)$ can
be constructed as follows: All tangent planes of $f$ through $v_c$ envelope a general cone
and $\Gamma(e_a)$ is in second order contact with that cone along the entire
ruling $e_a$ (see Fig.~7 in~\citep{s-nets-2020}). 

An \emph{L-asymptotic parameterization} is a parameterization $f(u,v)$ such that the
isoparameter tangent developables are L-asymptotic developables. The rulings
of the developables are symmetric with respect to the principal directions and at each
contact element, the common Mannheim sphere of the two isoparameter developables
equals the congruence sphere $s(u,v)$. This is \edit{related to} smooth
asymptotic parameterizations, where the common osculating plane of the isoparameter
curves equals the tangent plane of the surface. 
An L-asymptotic parameterization is the same as a smooth $S^*$-net introduced
in \citep{s-nets-2020}. %This is summarized and extended to $\R^{3,1}$ as follows.

In the cyclographic model, the sphere congruence $s(u,v)$ appears as a 
    surface $S(u,v) \subset \R^{3,1}$, \edit{such that all} isoparameter curves 
    have osculating planes that lie in the attached isotropic tangent hyperplanes $P(u,v)$.    
Indeed, since the tangents of isoparameter lines are self-L-conjugate, by~\eqref{eq:L-conj4}, we 
  get $L_P=N_P=0$. By \eqref{eq:normalP} and \eqref{eq:2ndformP}, this
  is equivalent to $\lla S_{u},N\rra= \lla S_{uu},N\rra =0$ and
  $\lla S_{v},N\rra= \lla S_{vv},N\rra =0$, which %together with $\lla N,N\rra =0$ expresses 
  means that the osculating planes of the isoparameter curves lie in 
  the attached isotropic hyperplanes.

Three consecutive or-planes in a discrete model of a 
developable surface determine an or-cone which is in oriented contact
with them. It is the discrete version of the osculating cone. 
This leads to the discrete versions of L-asymptotic parameterizations. They are not L-nets, but
counterparts to the well-studied A-nets \citep{bobenko-2009-ddg}.
One can view them as nets of or-planes, where we obtain exactly the discrete $S^*$-nets of  \citep{s-nets-2020}, 
or as nets of or-spheres, with the
following characterization:

\textit{
%\begin{prop} \MS{This is probably a definition, not a proposition.}
    A discrete L-asymptotic parameterization is a net of or-planes such
    that each or-plane $p_{ij}$ and its four %direct 
    neighbors $p_{i-1,j}$,
    $p_{i+1,j}$, $p_{i,j-1}$, $p_{i,j+1}$, are in or-contact to a
    common or-sphere $s_{ij}$. Viewing the structure as a net of or-spheres, each or-sphere $s_{ij}$ and its four %direct 
    neighbors $s_{i-1,j}$,
    $s_{i+1,j}$, $s_{i,j-1}$, $s_{i,j+1}$, are
    in or-contact to a common or-plane. 
%\end{prop}
}

Especially the second formulation is an obvious \edit{generalization} of discrete
A-nets; there, the spheres are points.

%--------------------------------------------------------
\subsection{Important special cases of attached sphere congruences} \label{ssec:special}
%---------------------------------------------------------------

\paragraph{Constant sphere radius} We start the discussion with a discrete
L-net \edit{with constant signed radius} $r$. At a \edit{plane} $p_{ij}$ of the L-net, we have four or-spheres of the same radius that are tangent to $p_{ij}$. Thus,  the centers of the four 
spheres lie in a plane that is parallel to $p_{ij}$ and
at signed distance $r$.
%Thus, \emph{an L-net \edit{with} a constant sphere radius is an offsetof a $Q^*$-net}. Switching the meaning of vertices and faces, 
\edit{This is the same as
an offset of a Q-net}. Clearly, all edge cones are cylinders of radius \edit{$|r|$}. 
\edit{Even in this simple case, our optimization approach can be preferable to working with the offset directly, as it can avoid singularities and improve approximation.}

\paragraph{Congruence of mid spheres} The mid sphere $s_m$ of a surface $f$ at a contact element
is the tangent or-sphere whose center is the midpoint of the principal curvature centers. Hence, its radius equals $r_m=(\rho_1+\rho_2)/2$. Being an average of the principal
spheres, the mid sphere provides a good local approximation of $f$. Moreover, like
the principal spheres, mid spheres are a concept of Laguerre geometry. A Laguerre transformation
maps mid spheres to mid spheres, in general changing their radii. 

If we attach the congruence of
mid spheres $s_m(u,v)$ to a surface $f$, we obtain as dual principal curvature
radii with respect to $s_m$,
$$  \rho_{s,1}^*=\rho_2-r_m=\frac{\rho_2-\rho_1}{2},\qquad  \rho_{s,2}^*=\rho_1-r_m=
\frac{\rho_1-\rho_2}{2}. $$
%
%Excluding the trivial case of a sphere $f$, 
At non-umbilic points, we have $\Lambda < 0$
and thus \emph{all contact elements are L-hyperbolic}. Due to \eqref{eq:Lasymp}, $\phi_a=\pm \pi/4$, hence the \emph{L-asymptotic tangents are orthogonal, bisecting the principal tangents}. L-conjugacy \eqref{eq:conjcurv} simplifies to
\be \bar{a}_1\bar{b}_1  -  \bar{a}_2\bar{b}_2 =0. 
\label{eq:conjcurvm}\ee
This expresses symmetry with respect to the L-asymptotic tangents. We see here 
relations to Euclidean minimal surfaces, which appear as that special case where
the mid spheres have radius zero and are the points of $f$. This is also the 
case in which L-conjugacy is the same as ordinary conjugacy. 

Using the attached congruence of mid spheres, we have frame fields of L-conjugate directions with a constant angle. They may provide interesting alternatives
to the principal directions.

\paragraph{Surfaces with only L-parabolic contact elements} At an L-parabolic
point, one principal curvature radius equals the radius $r$ of the congruence 
sphere. Hence, every surface $f$ has only L-parabolic contact elements if we attach
to it a congruence of principal spheres, say, the ones with radius $\rho_1$. Then,
any direction is L-conjugate to the second principal direction. This can provide
interesting types of flexibility for approximation \edit{and deserves further study}. 

%%%%%%%%%%%%%%%%%%%%%%%%%%%%%%%%%%%%%%%%%%%%%%%%%%%%%%%%%%%%%%%%%%%%%%%%%%%%%%
\section{Surface approximation with L-nets} \label{sec:approx}
%%%%%%%%%%%%%%%%%%%%%%%%%%%%%%%%%%%%%%%%%%%%%%%%%%%%%%%%%%%%%%%%%%%%%%%%%%%%%%

\begin{figure}
    \centering
    \begin{overpic}[width=1\linewidth]{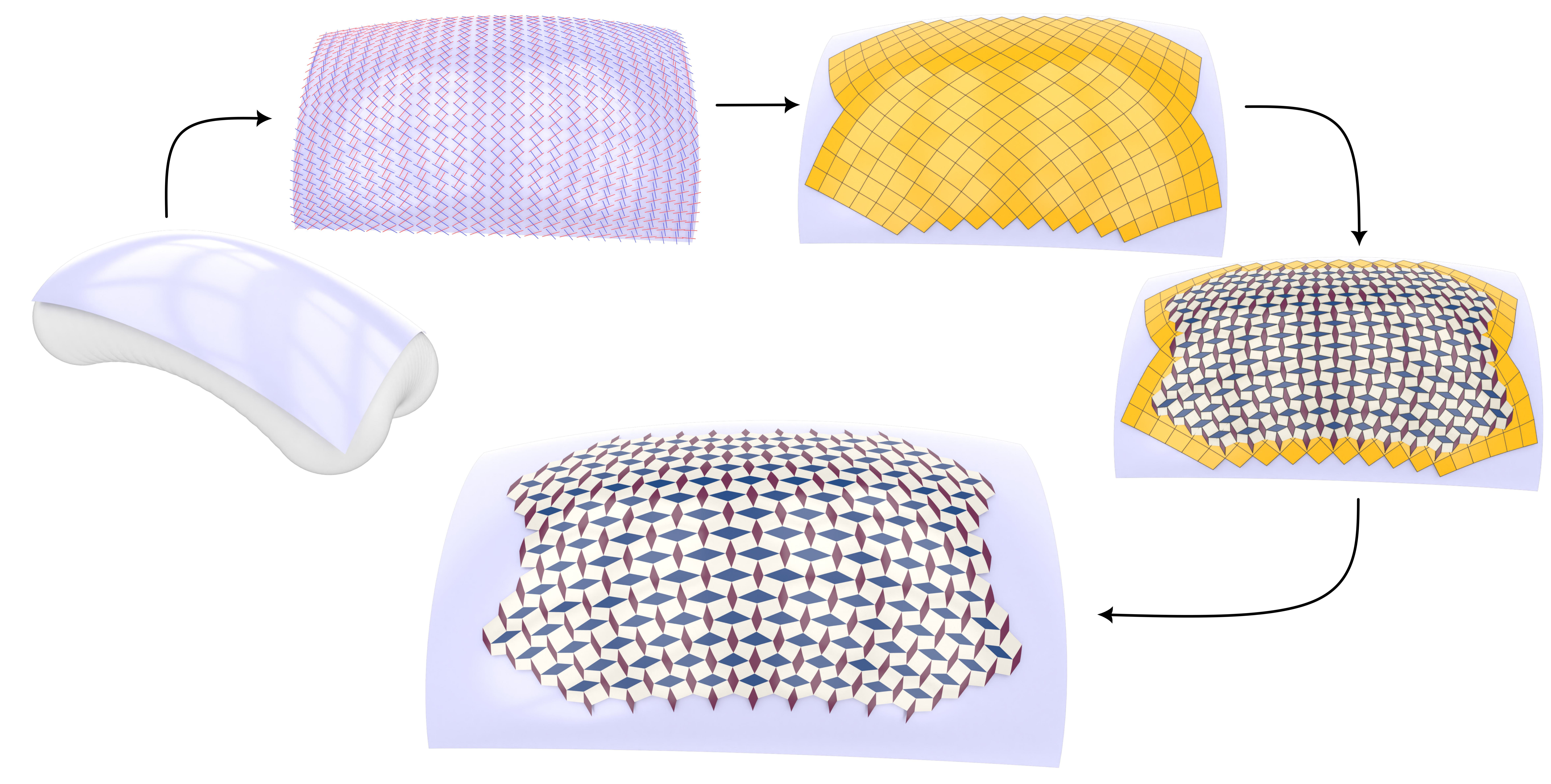}
    \put(5,20){\scriptsize Input surface}
    \put(10,17){\scriptsize $+$ }
    \put(4.5,15){\scriptsize sphere congruence }
    \put(0,47){\scriptsize  pseudo-L-conjugate}
    \put(4.7,44){\scriptsize  computation}
    \put(45,46){\scriptsize  Remeshing}
    \put(85,45){\scriptsize  Initialization}
    \put(80,8){\scriptsize Optimization}
    \end{overpic}
    \caption{Algorithm flow. Starting from an input surface with an attached sphere congruence, we compute pseudo-L-conjugate directions and remesh along these directions to obtain a quad mesh. We then initialize the L-net using planes at the vertices of the quad mesh, and refine it through a final optimization step. }
    \label{fig:flow}
\end{figure}
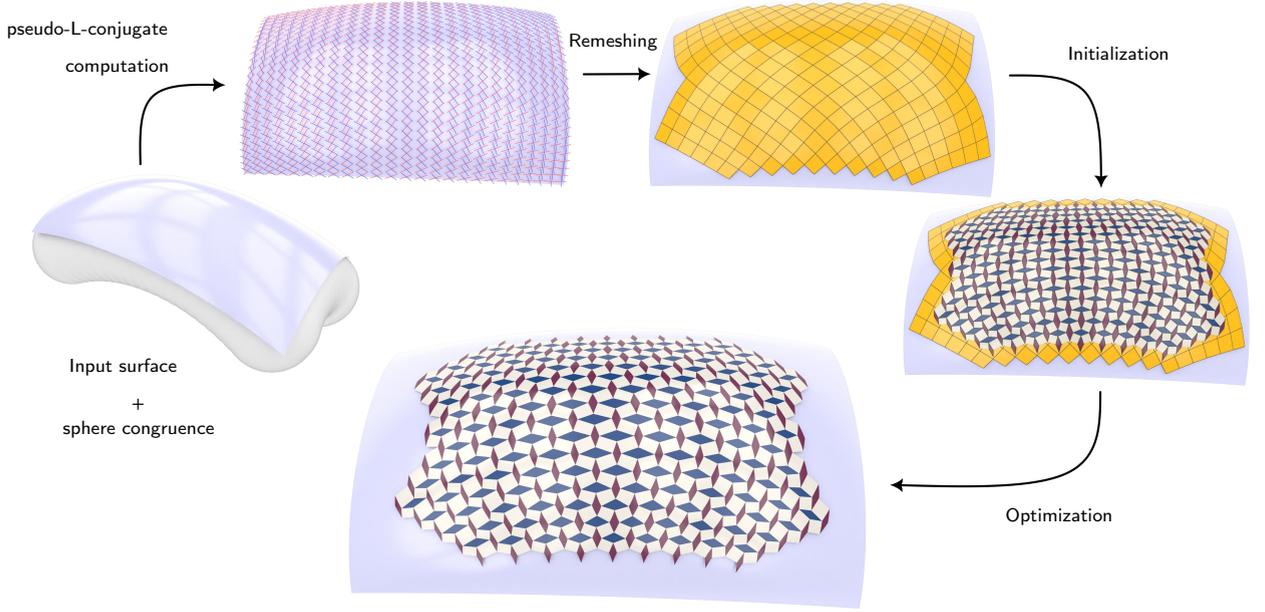

To approximate a surface with an L-net, we start with an input surface $f$ of positive Gaussian curvature and a signed radius function $r$, both defined on the same parameter domain and represented as tensor product B-splines.

An overview of our approach is as follows:
We use the B-spline representation to compute a smooth L-conjugate direction field over sampled surface points. This direction field guides a quad remeshing procedure, leading to a quadrilateral mesh aligned with the L-conjugate directions. This aligned quad mesh serves as the initialization for our L-net construction algorithm.
Finally, we refine the initial approximate L-net through numerical optimization using the Levenberg-Marquardt \citep{lm} algorithm to 
fulfill the constraints for an L-net with high accuracy while maintaining close
proximity to the reference surface and fairness (Fig.~\ref{fig:flow}). 

\subsection{Computation of L-conjugate directions}

\paragraph{Principal frames} 
The input surface $f(\lambda,\mu)$ is represented as a tensor 
product B-spline surface. At each sample point $f({\lambda}_i, 
{\mu}_j)$, we compute a principal frame $(t_1, t_2, 
n)$ and \edit{signed principal curvature radii $\rho_1$ and $\rho_2$}, where ${\lambda}_i$ and ${\mu}_j$ are obtained by 
uniform subdivision of the $\lambda$ and $\mu$ parameter 
domains. %To maintain consistency with the Laguerre geometry convention for oriented contact, we ensure that the surface has inward-pointing normals,
We choose $n$ inward-pointing, 
 so that 
\edit{$\rho_1,\rho_2\ge 0$}. 
%and all oriented spheres have positive signed radii. \MS{Unclear: the signed radii are not necessarily positive, even if the normals point inward the surface.}

\begin{figure}
    \centering
    \includegraphics[width=1\linewidth]{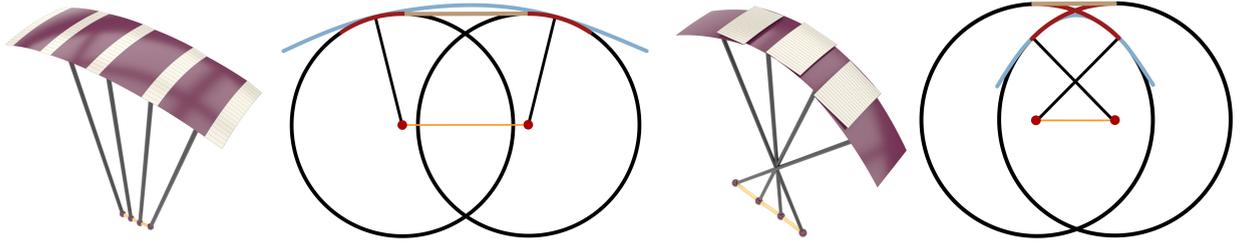}
    \caption{Spherical panels strips of an L-net aligned with the principal directions and their 2D counterparts. Black lines connect contact points on the 
    approximated surface/curve to the sphere/circle centers, 
    shown as red  dots, which are joined by orange lines. In the 2D diagrams, 
    brown lines represent common tangent 
    lines of both circles (the 2D counterparts of edge cones), red arcs correspond to the 2D counterparts of spherical panels, and 
    light-blue curves represent the approximated curve.  (Left) A radius function 
    within the safe range produces a smooth strip of spherical panels and cone 
    strips; (Right) A radius function %chosen with  
    $r > \min\{\rho_1, 
    \rho_2\}$ produces swallowtail singularities.}
    \label{fig:radiusImage}
\end{figure}

\paragraph{Signed-radius function}
The attached sphere congruence is represented by the signed radius function $r$. 
For reasonable approximation of the given surface, we require it to be positive. 
In the discrete setting, the choice of radius affects the overall appearance of the approximation. In particular, when approximating a surface $f$, singularities emerge when the sphere radii exceed the radii of the principal curvature spheres at the corresponding contact points.  Specifically, if the sphere radius satisfies $r > \min \{\rho_1, \rho_2\}$, swallow-tail singularities (containing regression curves) appear. For aesthetically pleasing approximations, we use $r < \min \{\rho_1, \rho_2\}$. See Fig.~\ref{fig:radiusImage}.

\paragraph{L-conjugate directions} 
For our initialization algorithm, we consider a discrete version of 
the net of contact curves of isoparameter tangent developable surfaces in 
a smooth L-conjugate net defined on $f$. Equivalently, this curve net is formed by the contact
curves of isoparameter channel surfaces in the attached sphere congruence.
This is not
fully Laguerre-geometric, but turned out to be easier to implement than
the plane-based definition of L-conjugate directions. Thus, we
are not initially computing L-conjugate directions, but directions that satisfy equation~\eqref{eq:Lconjcontact} (cf. Remark~\ref{rem:channel}). We call them \emph{pseudo-L-conjugate}. For their computation, we first prescribe one direction by specifying an angle $\theta$ against the first principal direction $t_1$. This angle is 
a design parameter (see Fig.~\ref{fig:bigpanels}). It may be set as a constant across all sampled points, or it can be defined as an angle field varying over $f$. However, a smoothly varying angle field is required
to arrive at a proper discrete version of a smooth L-net. The second direction is computed by solving equation~\eqref{eq:Lconjcontact}. In this way,
we obtain a pseudo-L-conjugate frame field. %via ordinary conjugacy, as explained in Sec.~\ref{ssec:L-conj}. 
We point out that the actual
initialization of the L-net, as explained below, indeed respects the associated L-conjugate frame
field, which is seen in the sides of the planar quadrilaterals (see Fig.~\ref{fig:alignment}).

\subsection{Quad remeshing and initialization}

\paragraph{Quad remeshing} Once the frame field is computed, we 
proceed with an anisotropic quad remeshing procedure 
implemented in the igl library \citep{jacobson:2018:libigl}, where LibQEx \citep{libQEX}
is used to extract the quad mesh when singularities 
are present. The extracted quad mesh $Q$ has edges aligned with the directions
of the frame field.

\begin{figure}
    \centering
    \includegraphics[width=1\linewidth]{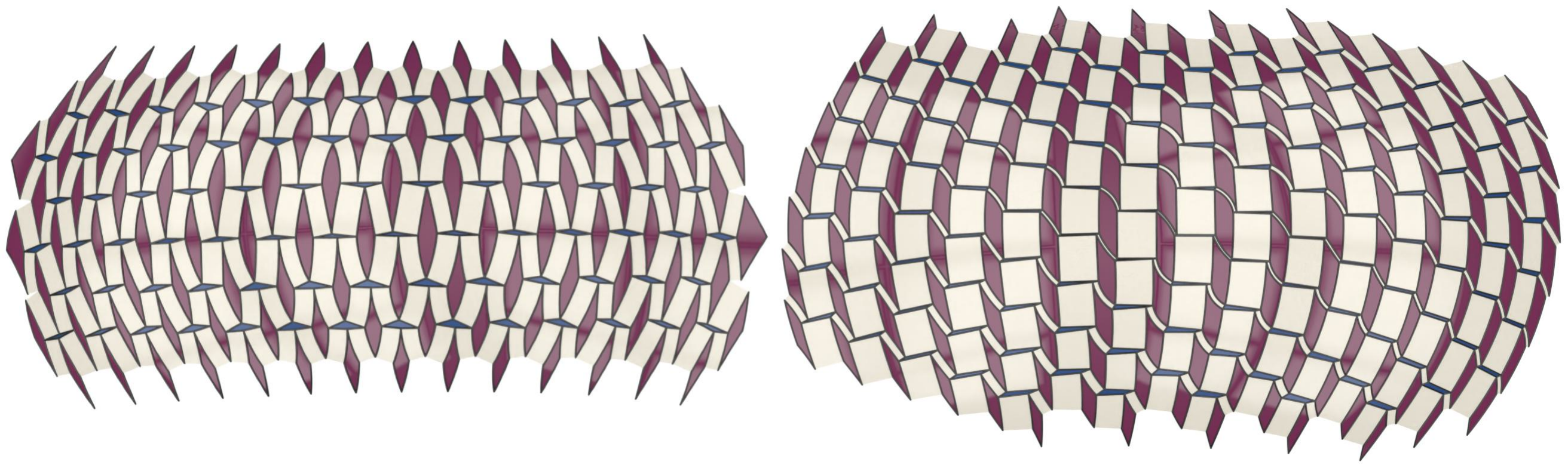}
    \caption{Two L-nets constructed from the same attached sphere congruence but with different initial L-conjugate directions. The initialization angle $\theta$ between the first L-conjugate direction and the first principal direction is set to constant value $\theta = \pi/3$ (left) and $\theta = \pi/6$ (right). }
    \label{fig:bigpanels}
\end{figure}

\paragraph{Initialization} To initialize an L-net, we 
can interpret the vertices of $Q$ as contact points 
of oriented tangent planes or as contact points 
of tangent oriented spheres. 

\edit{The first approach, placing planes at the vertices of $Q$, preserves the mesh extent near the boundary: when initializing with or-spheres instead, the boundary spheres are not used for panel generation, resulting in a shrunken L-net. However, the differences are visually negligible} (see Fig.~\ref{fig:Different Initialization}). Following 
the first interpretation, our L-net $L=(\SV, 
\CE, \FP)$ is computed as follows: 
(i) $Q$ and $L$ have \edit{dual combinatorics}, so for 
each vertex $v_i \in Q$, the corresponding 
plane $p_i \in \FP$ is the plane with normal 
$n_i$ (the normal of the reference surface $f$ at vertex $v_i$) and signed distance from the origin 
$h_i = - \langle v_i, n_i\rangle$; (ii) for each quad $q \in Q$, 
we define an initial sphere $s_f=
(c_f, r_f)$ by projecting the barycenter of $q$ onto 
the reference surface as $b_q = f(\lambda_q, \mu_q)$; the radius of the sphere is $r_f = r(\lambda_q, \mu_q)$ and its center is $c_f = b_q + r_f n(\lambda_q, \mu_q)$; (iii) we compute the cone strip for each pair of adjacent spheres. % $s_i \sim s_j$.
In general, this initial configuration is not an L-net since the spheres and planes are not exactly in oriented contact.
However, it provides a good initial guess for the final optimization algorithm to obtain an L-net.

\begin{figure}
    \centering
    \includegraphics[width=0.9\linewidth]{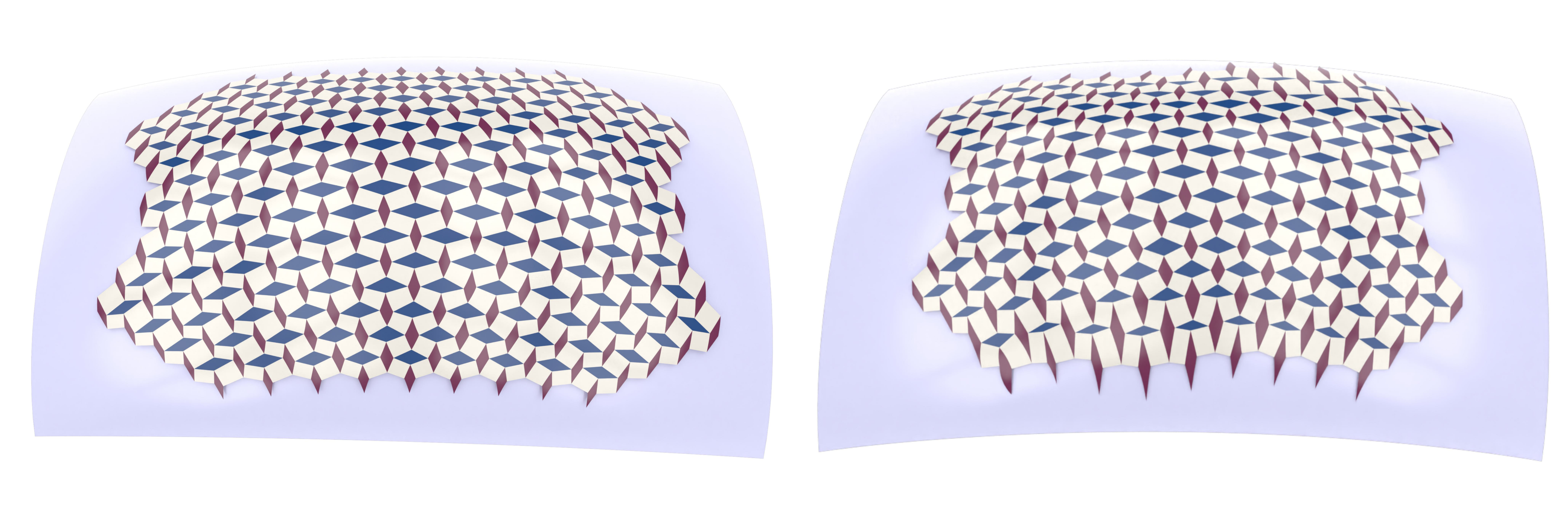}
    \caption{L-nets initialized interpreting vertices of the quad mesh $Q$ as contact points of oriented tangent planes (left) and tangent oriented spheres (right). }
    \label{fig:Different Initialization}
\end{figure}

%\paragraph{L-net data structure} [This part can be skipped if found irrelevant]
%We use a Half-edge data structure to handle the combinatorics of the L-net where oriented planes are considered as vertices, oriented spheres as faces and oriented cone strips as edges. [Details of the visualization and data structure but basically is just a Half-edge structure the details are more computational and probably not too relevant]

\subsection{Final optimization}

\paragraph{Variables} 
In the final optimization, our goal is to \edit{approximate $f$ with an L-net $L=(\edit{\SV}, \CE, \edit{\FP})$ with high accuracy}. The whole structure depends only on the or-spheres $s_{ij} = (c_{ij},r_{ij}) \in \edit{\SV}$ and or-planes $p_{ij} = (n_{ij},h_{ij}) \in \edit{\FP}$, so the variables in our Levenberg-Marquardt algorithm are the sphere centers $c_{ij}$, signed radii $r_{ij}$, plane normals $n_{ij}$, and signed distances $h_{ij}$ from the origin. Plane normals are constrained to be unit; we express this by vanishing energy 
\begin{equation}
    E_{\mathrm{unit}} = \sum_{p_{ij} \in \edit{\FP} } ( \|n_{ij}\|^2 -1)^2.
\end{equation}

\paragraph{Oriented contact} We need to ensure that the or-spheres and their neighboring or-planes are in oriented contact. This is expressed by vanishing energy %$E_{OC}$,
\begin{equation}
    E_{\mathrm{oc}} = \sum_{p_{ij} \in \edit{\FP}} \sum_{a, b = 0}^{1} \left( \langle c_{i+a,\ j+b}, n_{ij} \rangle + h_{ij} - r_{i+a,\ j+b} \right)^2.
\end{equation}
% \paragraph{L-net skeleton ?} [Not final version]
% An L-net is a Laguerre geometric object; however, to achieve aesthetically pleasing results, we must control the Euclidean elements of the L-net. The edge  network of an L-net consists of two  discrete families of polylines, each formed by alternating straight line segments and circular arcs that we refer to as the skeleton of the L-net. 

\paragraph{L-net Fairness}

\edit{As is well known in optimization, the energy needs to include fairness terms that penalize abrupt changes in the directions of the patch boundaries.
}

%\edit{By \emph{fairness} of an L-net we mean that the alternating straight segments and circular arcs forming its boundary curves change direction smoothly — without abrupt transitions — as detailed below.}

\begin{wrapfigure}[8]{r}{0.32\linewidth}
\vspace{-0.5cm}
\hspace{-0.2cm}
 \includegraphics[width=0.85\linewidth]{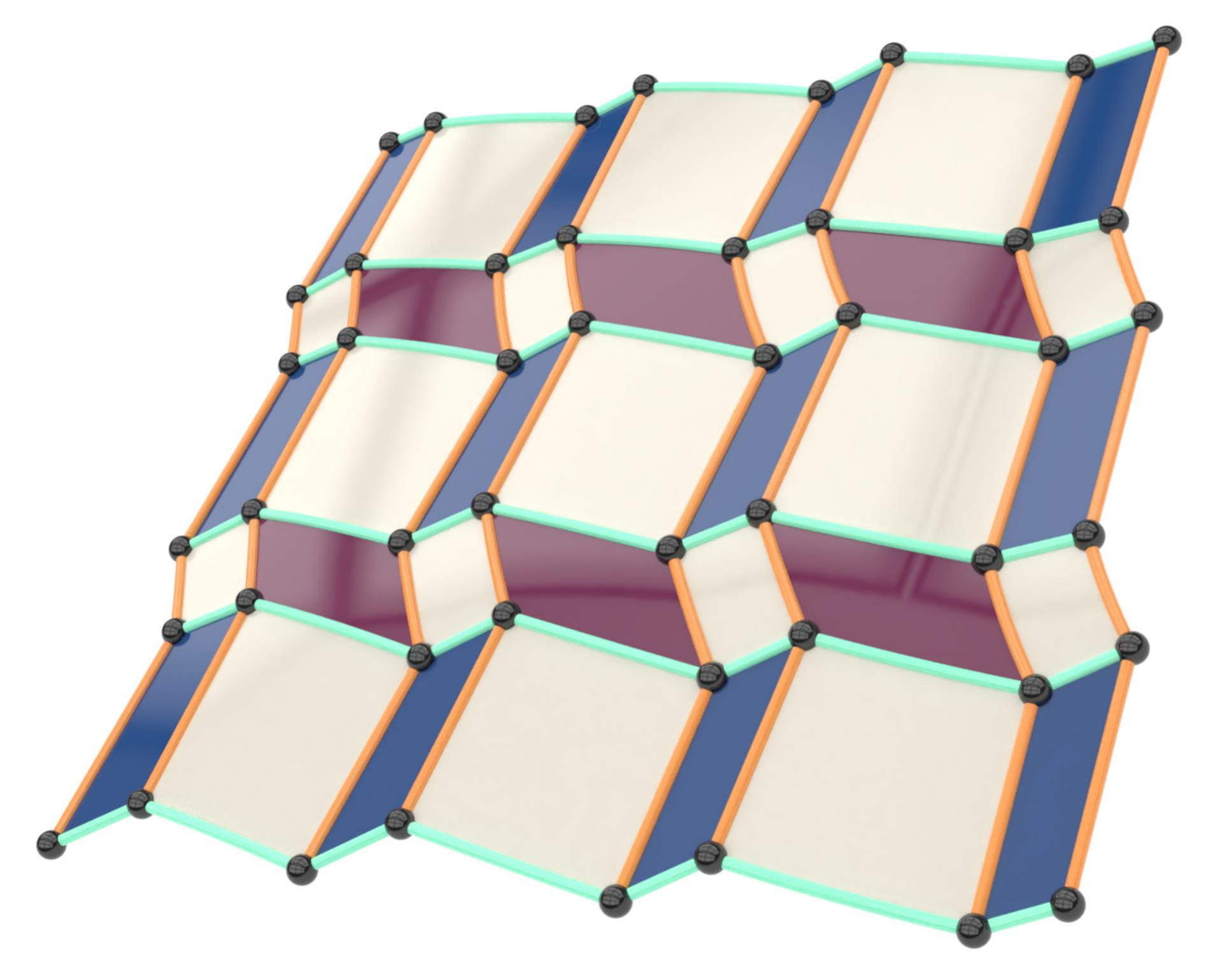}
\end{wrapfigure}
The boundary curves of the elements of an L-net form two families of
piecewise-smooth curves (see the orange and green curves in the inset), each consisting of alternating
straight line segments and circular arcs meeting at the contact points between 
or-spheres and or-planes. They follow the pattern $\ell_1, \gamma_1, \ell_2, 
\gamma_2, \ldots$, where $\ell_i$ denotes a straight segment and $\gamma_i$ a 
circular arc (Fig.~\ref{fig:curvesLnet}). To preserve the structure of L-nets, we impose 
constraints independently on each subsequence: consecutive straight segments 
$\ell_i, \ell_{i+1}$ are constrained regardless of the arc~$\gamma_i$ between 
them, and similarly for consecutive arcs~$\gamma_i, \gamma_{i+1}$.

\begin{figure}
\begin{center}
    \begin{overpic}[width=0.55\linewidth]{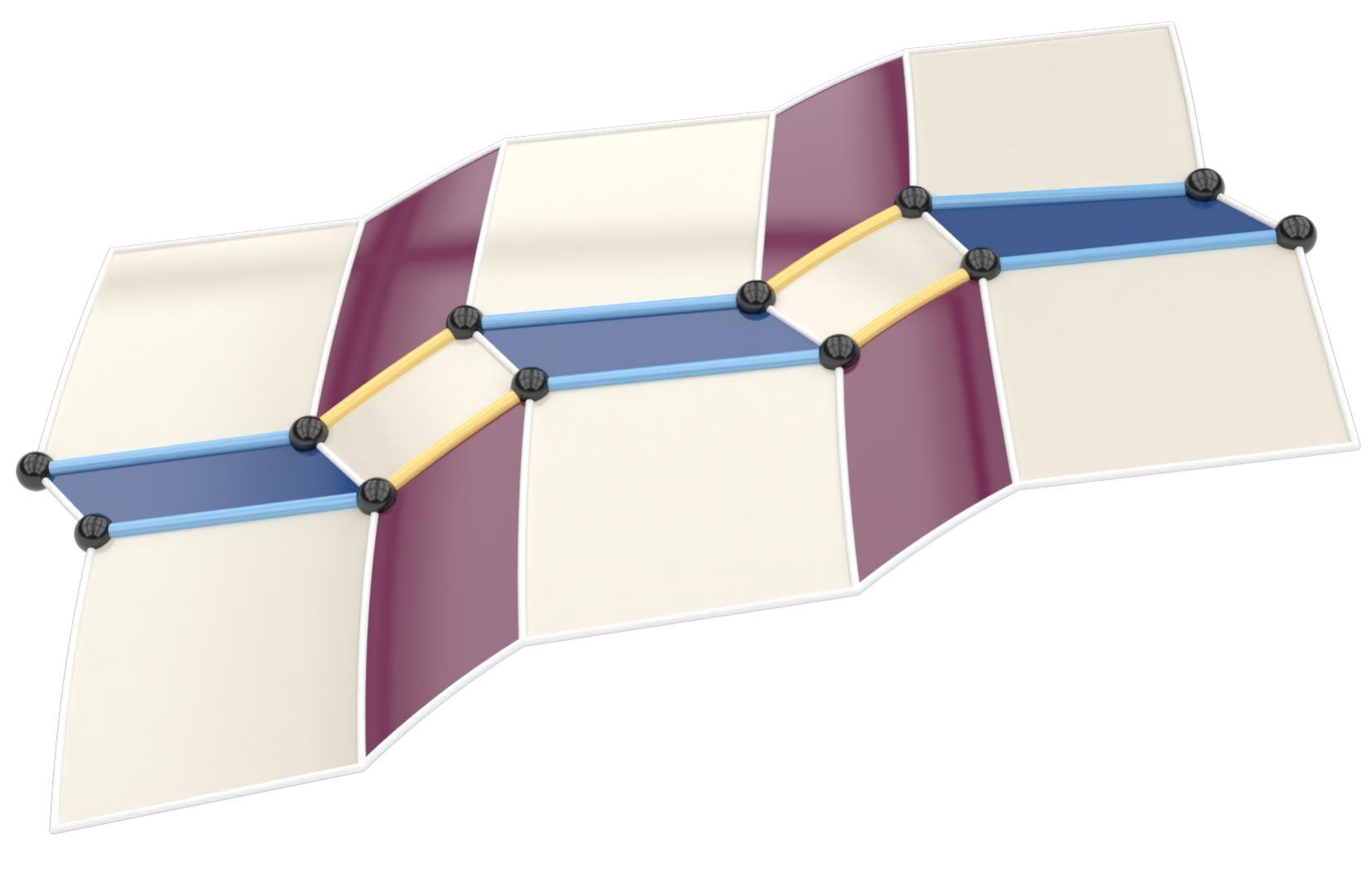}
    \put(17,22){\contour{white}{$\ell_1$}}
    \put(32,27){\contour{white}{$\gamma_1$}}
    \put(50,33){\contour{white}{$\ell_2$}}
    \put(67,38){\contour{white}{$\gamma_2$}}
    \put(80,42){\contour{white}{$\ell_3$}}
    \end{overpic}
    \caption{Piecewise-smooth curves alternating pattern between (light blue) straight segments $\ell_i$ and (yellow) circular arcs $\gamma_i$ along a strip.}
    \label{fig:curvesLnet}
\end{center}
\end{figure}
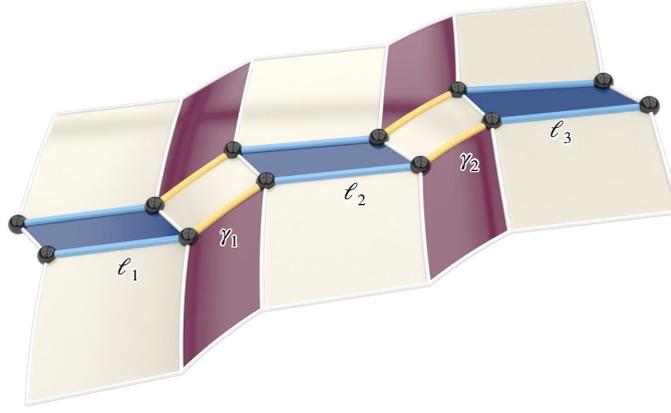

Aiming at an appropriate fairness term for the piecewise smooth curves,
we first need to express their vertices.
To constrain the straight segments, we consider three consecutive \edit{or-spheres} $s_i,s_j,s_k$ of signed radii $r_i,r_j,r_k$ and or-planes $p_0,p_1,p_2,p_3$ with normals $n_0,n_1,n_2,n_3$, where $p_0,p_3$ are in oriented contact 
with $s_i,s_j$ and $p_1,p_2$ are in oriented contact with %the spheres 
$s_j,s_k$ (see Fig.~\ref{fig:fairness}, left).  In this arrangement, we have two pairs of consecutive straight segments %$(\ell_i, \ell_{i+1}), \ (\bar \ell_i, \bar \ell_{i+1})$ that form part of the edges of the planes $p_0,p_1$ and $p_3,p_4$ respectively, where the endpoints of the straight segments are 
$\ell_i = (a_0,a_1)$, $\ell_{i+1} = (a_2,a_3)$, $\bar \ell_i = (b_0,b_1)$, and $\bar \ell_{i+1} = (b_2,b_3)$ with the endpoints %correspond to 
at the contact points between the or-spheres and or-planes: %given by 

\begin{align*}
    a_0 &= c_i - r_i  n_0, & b_0 &= c_i - r_i  n_3, \\
    a_1 &= c_j - r_j  n_0, & b_1 &= c_j - r_j  n_3, \\
    a_2 &= c_j - r_j  n_1, & b_2 &= c_j - r_j  n_2, \\
    a_3 &= c_k - r_k  n_1, & b_3 &= c_k - r_k  n_2.
\end{align*}
These segments are sides of the \edit{quadrilateral faces} contained in the planes $p_0$, $p_1$, $p_3$, and $p_2$, respectively.
The fairness term for straight segments is based on the idea that vectors 
$a_1-a_0$ and $a_3-a_2$ should be close to each other, and likewise for the $b$'s. This is expressed by a small value of the energy 
\begin{equation}
    E_{\ell\mathrm{fair}} = \sum_{s_i,s_j,s_k \text{  consecutive}}\left( \|a_1 - a_0 - a_3 + a_2\|^2 + \|b_1 - b_0 - b_3 + b_2\|^2\right).
\end{equation}

\begin{figure}
    \centering
    \begin{tabular}{l r}
              
\begin{overpic}[width=0.45\linewidth]
        {{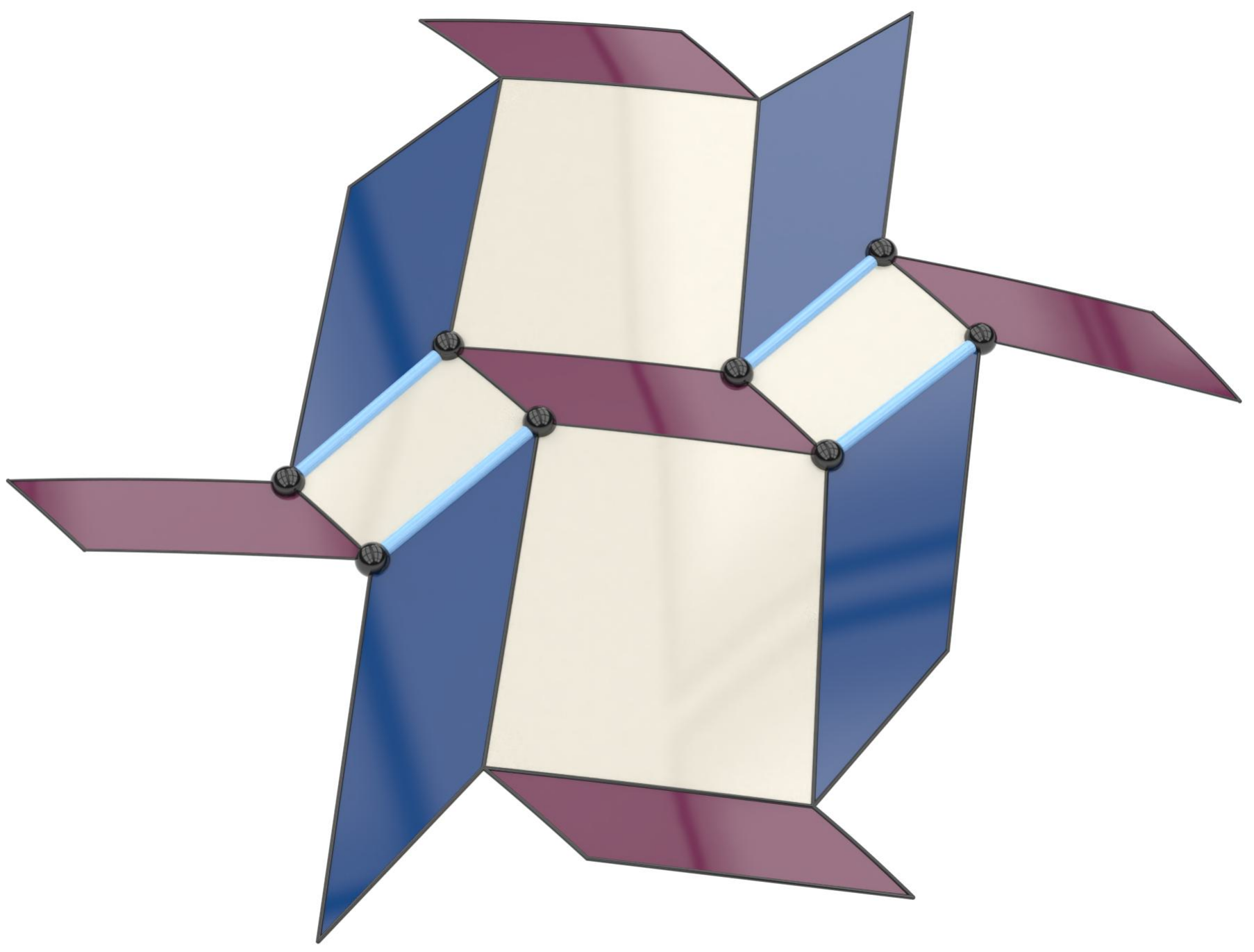}}
\put(24, 43){\scriptsize \contour{white}{$b_0$}}
\put(40, 50){\scriptsize \contour{white}{$b_1$}}
\put(57, 50){\scriptsize \contour{white}{$b_2$}}
\put(73, 55){\scriptsize \contour{white}{$b_3$}}
\put(30, 27.5){\scriptsize \contour{white}{$a_0$}}
\put(44, 37){\scriptsize \contour{white}{$a_1$}}
\put(65, 34){\scriptsize \contour{white}{$a_2$}}
\put(80, 43){\scriptsize \contour{white}{$a_3$}}
\put(12, 36){\contour{white}{$s_i$}}
\put(52, 42){\contour{white}{$s_j$}}
\put(86, 48){\contour{white}{$s_k$}}
\put(31, 21){\color{white}{$p_0$}}
\put(30, 54){\color{white}{$p_3$}}
\put(68, 27){\color{white}{$p_1$}}
\put(63, 58){\color{white}{$p_2$}}
    \end{overpic}  
    &
    \begin{overpic}[width=0.5\linewidth]
        {{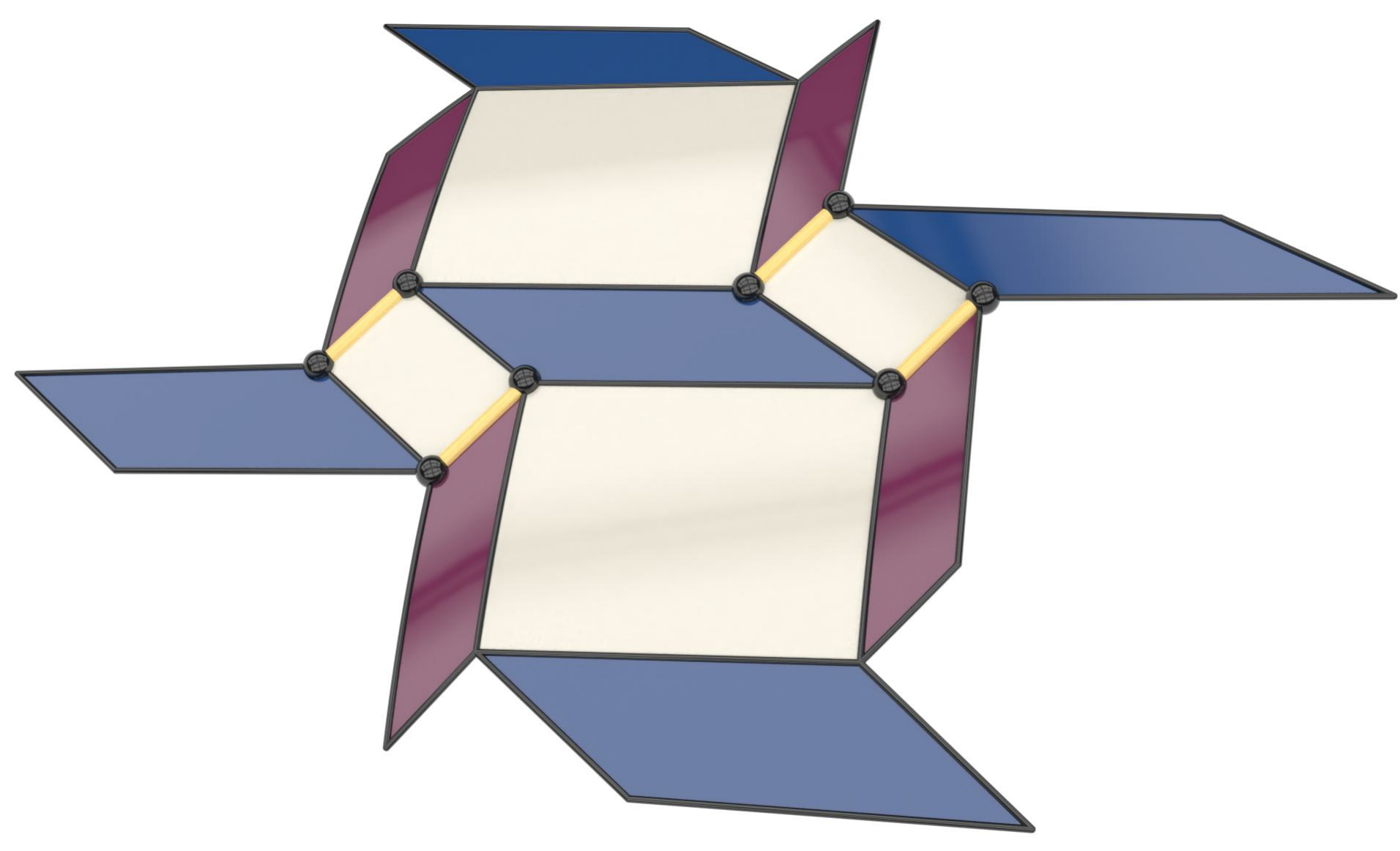}}
    \put(20, 36){\scriptsize \contour{white}{$\beta_0$}}
    \put(31, 42){\scriptsize \contour{white}{$\beta_1$}}
    \put(50, 43){\scriptsize \contour{white}{$\beta_2$}}
    \put(61, 47){\scriptsize \contour{white}{$\beta_3$}}
    \put(28, 24){\scriptsize \contour{white}{$\alpha_0$}}
    \put(36, 30){\scriptsize \contour{white}{$\alpha_1$}}
    \put(60, 30){\scriptsize \contour{white}{$\alpha_2$}}
    \put(70, 35){\scriptsize \contour{white}{$\alpha_3$}}
    \put(15, 30){\contour{white}{$p_i$}}
    \put(45, 36){\contour{white}{$p_j$}}
    \put(80, 42){\contour{white}{$p_k$}}
    \put(30, 17){\small \color{white}{$s_0$}}
    \put(64, 25){\small \color{white}{$s_1$}}
    \put(56.5, 51){\small \color{white}{$s_2$}}
    \put(27, 46){\small \color{white}{$s_3$}}
\end{overpic}  
\end{tabular}
    \caption{%Illustration of t
    The notation for the fairness energy terms $E_{\ell \mathrm{fair}}$ (left) and $E_{\gamma \mathrm{fair}}$ (right). }
    \label{fig:fairness}
\end{figure}
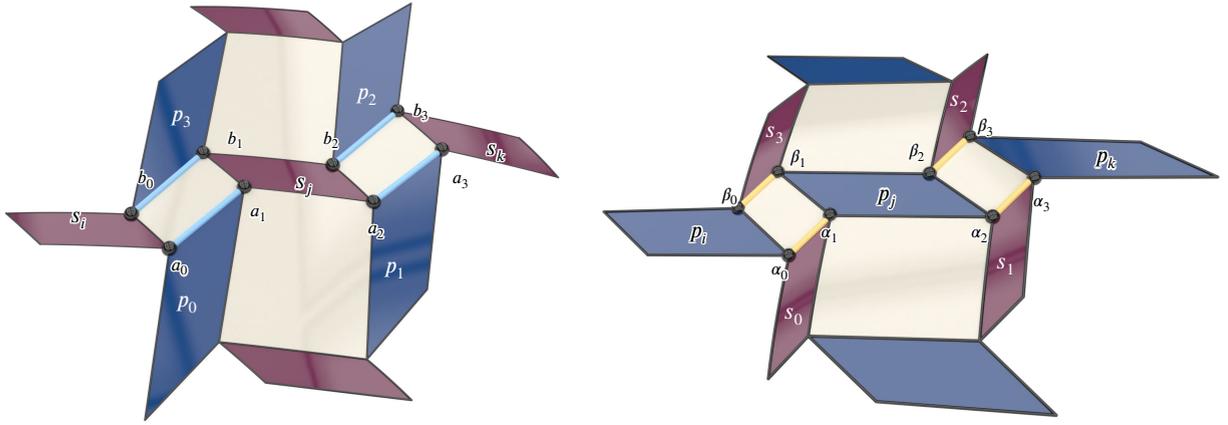

%We concentrate on the endpoints of the circular arcs  $\gamma_i = (\alpha_0,\alpha_1)$, $\gamma_{i+1} = (\alpha_2,\alpha_3)$, $\bar \gamma_i = (\beta_0,\beta_1)$ and $\bar \gamma_{i+1} = (\beta_2,\beta_3)$ given by 

\edit{This energy is well-defined even if the or-planes and or-spheres are not exactly in oriented contact, and its choice is motivated by the discrete approximation of gradually changing rulings.}

To constrain the circular arcs, we proceed in the same way. Consider 
three consecutive or-planes $p_i,p_j,p_k$ and or-spheres $s_0,s_1,s_2,s_3$,
where $s_0,s_3$ are in oriented contact with $p_i,p_j$ and $s_1,s_2$ are 
in oriented contact with $p_j,p_k$ (see Fig.~\ref{fig:fairness}, right). In this arrangement, we have two 
pairs of consecutive circular arcs %segments $(\gamma_i, \gamma_{i+1}), \ (\bar \gamma_i, \bar \gamma_{i+1})$ 
$\gamma_i = (\alpha_0,\alpha_1)$, 
$\gamma_{i+1} = (\alpha_2,\alpha_3)$, $\bar \gamma_i = (\beta_0,\beta_1)$, and 
$\bar \gamma_{i+1} = (\beta_2,\beta_3)$ with the endpoints
\begin{align*}
    \alpha_0 &= c_0 - r_0 n_i, & \beta_0 &= c_3 - r_3 n_i, \\
    \alpha_1 &= c_0 - r_0 n_j, & \beta_1 &= c_3 - r_3 n_j, \\
    \alpha_2 &= c_1 - r_1 n_j, & \beta_2 &= c_2 - r_2 n_j, \\
    \alpha_3 &= c_1 - r_1 n_k, & \beta_3 &= c_2 - r_2 n_k.
\end{align*}
These arcs are some of the boundary curves of the \edit{spherical panels} $s_0$, $s_1$, $s_3$, and $s_2$.
The fairness term for circular arcs is expressed by a small value of the energy
\begin{equation}
    E_{\gamma\mathrm{fair}} = \sum_{p_i,p_j,p_k \text{  consecutive}} \left(\|\alpha_1 - \alpha_0 - \alpha_3 + \alpha_2\|^2 + \|\beta_1 - \beta_0 - \beta_3 + \beta_2\|^2\right).
\end{equation}

\paragraph{Proximity Energies}
To ensure that the L-net approximates the reference surface $f$, we include terms that express the proximity of the \edit{plane-sphere} contact points in the L-net to the surface $f$. For each contact point $x$ in $L$, we denote by $x_f$ its closest-point projection onto $f$ and denote by $n_f$ the surface normal at $x_f$. We express the proximity in terms of the distance %$E_{prox}$ 
to the projected point $x_f$ and the distance %$E_{tan}$ 
to the tangent plane of the reference surface $f$ at $x_f$,
\begin{equation}
    E_{\mathrm{prox}} = \sum_{\edit{x} } \|x - x_f\|^2, \quad E_{\mathrm{tan}} = \sum_{\edit{x}} \langle x - x_f, n_f \rangle^2.
\end{equation}

\paragraph{Tangential distance} %[Optional]
\edit{For visual appearance and design purposes, it may be desirable}
for the spherical panels to be larger than
the conical and planar quads. This can be 
addressed by minimizing the width of the cone strips connecting adjacent 
spheres. The width equals the tangential distance between two or-spheres 
$s_i$ and $s_j$, and is computed as  $d_{tang}(s_i, s_j) = \sqrt{\|c_i - c_j\|^2 - (r_i - r_j)^2}$. 
We introduce the following energy term to minimize the tangential distance 
over all pairs of \edit{adjacent or-spheres} of the L-net:
\begin{equation}
    E_{\mathrm{td}} = \sum_{s_i, s_j \text{ adjacent}} \left(\|c_i - c_j\|^2 - (r_i - r_j)^2 \right)^2.
\end{equation}

\begin{figure}
    \centering
    \includegraphics[width=1\linewidth]
    {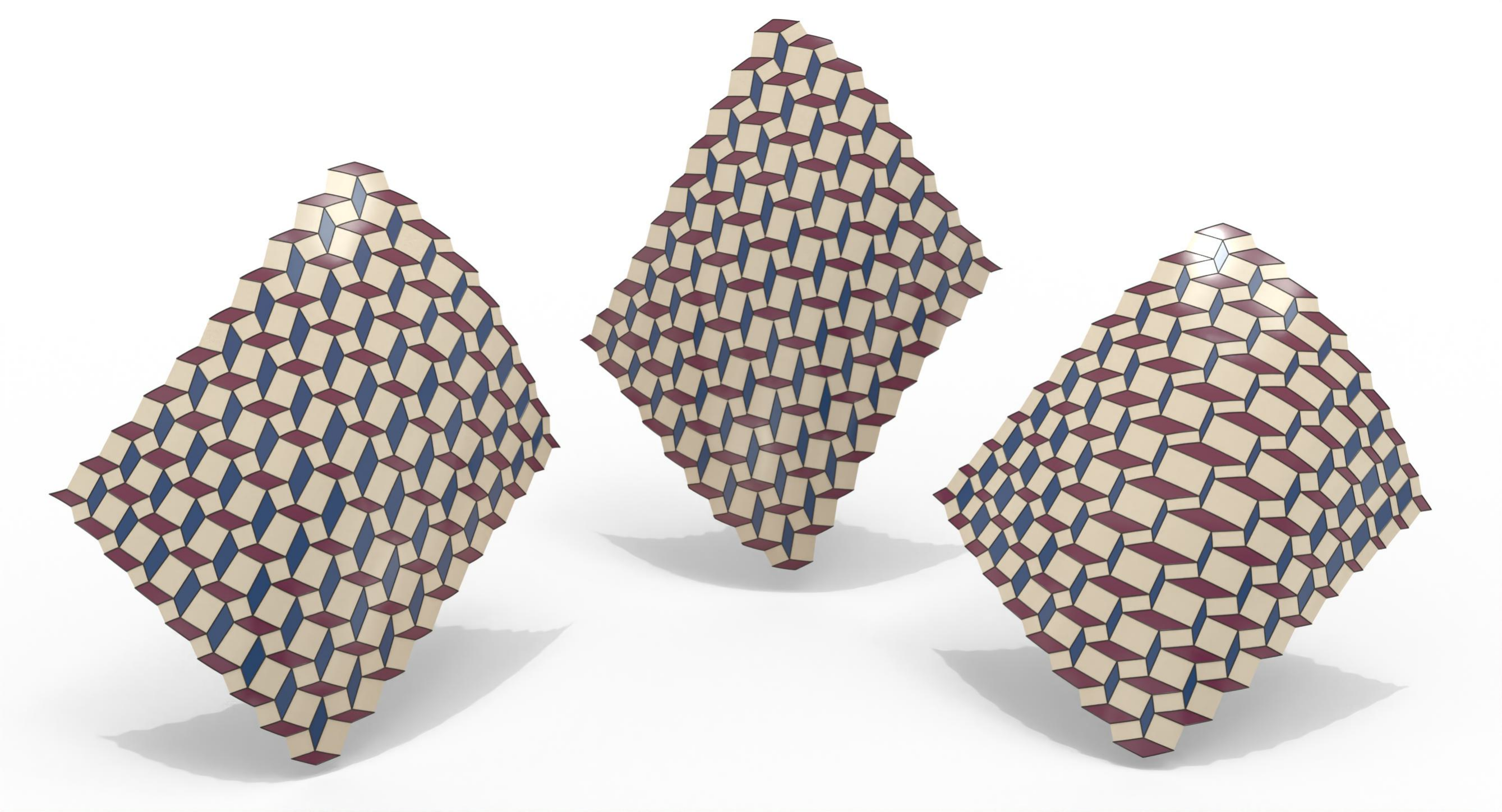}
    \caption{(Mid) L-mesh initialization using L-conjugate directions computed using $\theta(u,v) = \pi/4$ for the first direction and multiplier factor for the signed radius function $\tau = 0.75$.  (Left) Final optimization without tangential distance minimization. (Right) Final optimization with tangential distance minimization. }
    \label{fig:tdEffect}
\end{figure}

\paragraph{Optimization}
We also include the standard Levenberg-Marquardt regularization term $E_{\mathrm{reg}}=\|X- X^{\mathrm{prev}}\|^2$ to avoid 
abrupt changes during the optimization. Here $X$ is the vector of all
variables and $X^{\mathrm{prev}}$ contains their values from the
previous iteration. 
Combining all energy terms with their corresponding weights yields the 
total energy
\begin{equation}
\begin{split}
    E = \, &\omega_{\mathrm{oc}} E_{\mathrm{oc}} + \omega_{\gamma\mathrm{fair}} E_{\gamma\mathrm{fair}} + \omega_{\ell\mathrm{fair}} E_{\ell\mathrm{fair}} 
    %\\           &
    + \omega_{\mathrm{prox}} E_{\mathrm{prox}} + \omega_{\mathrm{tan}} E_{\mathrm{tan}} + \omega_{\mathrm{td}} E_{\mathrm{td}} + 
           \omega_{\mathrm{reg}} E_{\mathrm{reg}} + \omega_{\mathrm{unit}} E_{\mathrm{unit}}.
\end{split}
\end{equation}
All examples were computed with $\omega_{\mathrm{reg}} = 10^{-4}$, 
$\omega_{\ell\mathrm{fair}} = 10^{-3}$, $\omega_{\gamma\mathrm{fair}} = 10^{-3}$, $\omega_{\mathrm{unit}}=10$, 
and $\omega_{\mathrm{oc}} = 1$, where the only modifications appear on the proximity terms, and \edit{whether we include} the tangential distance minimization term or not (see Fig.~\ref{fig:tdEffect}).
Both fairness energies $\omega_{\ell\mathrm{fair}}$ and $\omega_{\gamma\mathrm{fair}}$ 
are multiplied by $0.1$ every ten steps to decrease the influence of these terms. After a certain number of iterations or convergence, we run 20 additional optimization 
steps with all weights set to zero except for $\omega_{oc}$ and $\omega_{unit}$. The resulting 
change in the L-net is minimal, but this final pass ensures that all 
elements satisfy the oriented contact constraints close to machine precision. It also addresses the fact that fairness energies,
which should just be small but in general not vanish, 
do not fight against the hard contact constraints.

%%%%%%%%%%%%%%%%%%%%%%%%%%%%%%%%%%%%%%%%%%%%%%%%%%%%%%%%
\section{Results, discussion, and future research}
%%%%%%%%%%%%%%%%%%%%%%%%%%%%%%%%%%%%%%%%%%%%%%%%%%%%%%%%%

\begin{figure}
    \centering
    \includegraphics[width=0.8\linewidth]{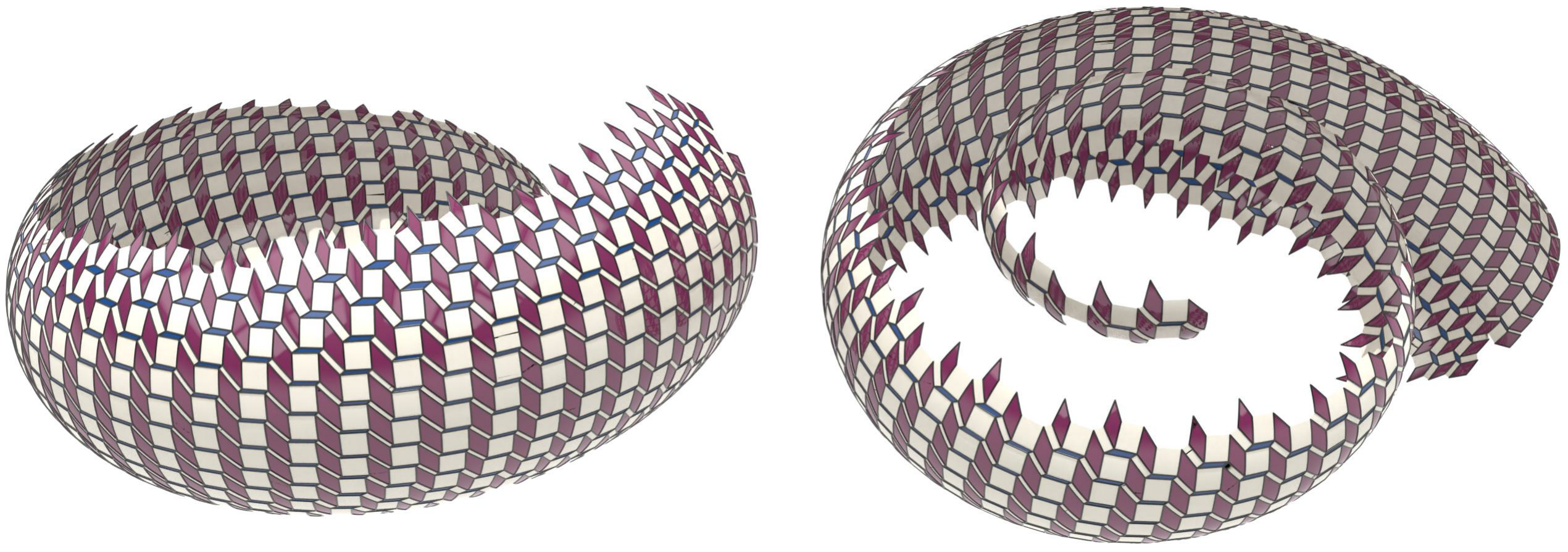}
    \includegraphics[width=0.8\linewidth]{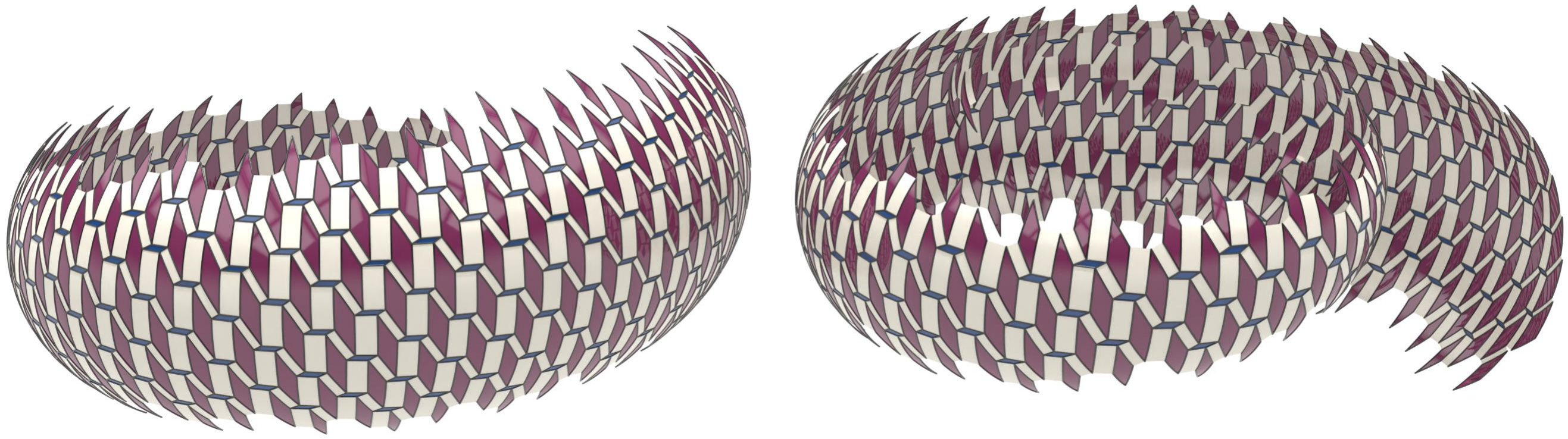}
    \caption{Spiral L-net designs obtained with different initializations. (Top) $\tau = 0.90$, $\theta_{\cos}^{(0,\pi/2)}(u)$; (bottom) $\tau = 0.95$, $\theta_C = \pi/4$. }
    \label{fig:nautiex1}
\end{figure}

For our results, we define a signed radius function 
$r(u,v) = \tau \cdot \min\{\rho_1, \rho_2\}$, 
with $\tau \in (0,1)$ ensuring that the sphere radius remains strictly 
below the smallest principal curvature radius, avoiding
swallow-tail singularities.

\begin{table}[H]
    \centering\resizebox{\columnwidth}{!}{%   % or \textwidth
    \begin{tabular}{ | l | l | l | c | c| c| c | c | c | c |}
    \hline 
   Fig	&	$\tau$	&	$\theta(u,v)$	&	$\omega_{\mathrm{prox}}$	&	$\omega_{\mathrm{tan}}$	&	$\omega_{\mathrm{td}}$	&	T(ms)/iter	&	iter	&	Residual	&	Residual $E_{\mathrm{oc}}$ 	\\	\hline \hline
\ref{fig:spiral_defM}	&	0.90	&	$\theta_C =  5 \pi/12 $	&	$10^{-4}$	&	$10^{-4}$	&	$10^{-3}$	&	$446.9$	&	$100$	&	$2.08 \cdot 10^{-7}$	&	$8.94 \cdot 10^{-18}$	\\	
\ref{fig:seashellex5}	&	0.75	&	$\theta_C = \pi/4 $	&	$10^{-4}$	&	$10^{-4}$	&	$10^{-5}$	&	$84.1$	&	$100$	&	$3.30 \cdot 10^{-5}$	&	$1.61 \cdot 10^{-26}$	\\
%\ref{fig:bigpanels}(top-left)	&	0.94	&	$\theta_C = \pi/3 $	&	$10^{-4}$	&	$10^{-4}$	&	$5 \cdot 10^{-3}$	&	$11.3$	&	$150$	&	$2.62 \cdot 10^{-6}$	&	$2.94 \cdot 10^{-14}$	\\	
\ref{fig:bigpanels} (left)	&	0.94	&	$\theta_C =  \pi/3 $	&	$10^{-4}$	&	$10^{-4}$	&	0	&	$10.5$	&	$150$	&	$4.75 \cdot 10^{-8}$	&	$8.41 \cdot 10^{-28}$	\\	
%\ref{fig:bigpanels}(top-right)	&	0.94	&	$\theta_C =  \pi/6$	&	$10^{-2}$	&	$10^{-2}$	&	$10^{-3}$	&	$13.2$	&	$150$	&	$2.54 \cdot 10^{-6}$	&	$8.54 \cdot 10^{-14}$	\\	
\ref{fig:bigpanels} (right)	&	0.94	&	$\theta_C = \pi/6$	&	$10^{-2}$	&	$10^{-2}$	&	0	&	$13.9$	&	$150$	&	$1.97 \cdot 10^{-7}$	&	$1.52 \cdot 10^{-27}$	\\	

\ref{fig:tdEffect} (left)	&	0.75	&	$\theta_C = \pi/6$	&	$10^{-4}$	&	$10^{-4}$	&	0	&	$7.4$	&	$200$	&	$5.30 \cdot 10^{-9}$	&	$1.93 \cdot 10^{-28}$	\\	
\ref{fig:tdEffect} (right)	&	0.75	&	$\theta_C = \pi/6 $	&	$10^{-4}$	&	$10^{-4}$	&	$10^{-3}$	&	$8.4$	&	$200$	&	$7.62 \cdot 10^{-7}$	&	$1.56 \cdot 10^{-26}$	\\	
\ref{fig:nautiex1} (top)	&	0.90	&	$\theta_{cos}^{(0, \pi/2)}(u) $	&	$10^{-4}$	&	$10^{-4}$	&	$10^{-3}$	&	$40.2$	&	$200$	&	$6.10 \cdot 10^{-6}$	&	$2.04 \cdot 10^{-26}$	\\	
\ref{fig:nautiex1} (bot)	&	0.95	&	$\theta_C = \pi/4 $	&	$10^{-4}$	&	$10^{-4}$	&	$5 \cdot 10^{-3}$	&	$40.4$	&	$200$	&	$9.70 \cdot 10^{-6}$	&	$8.70 \cdot 10^{-27}$	\\	
\ref{fig:paterns} (top-left)	&	0.80	&	$\theta_{0}^{5 \pi/12}(u) $	&	$10^{-4}$	&	$10^{-4}$	&	0	&	$19.2$	&	$150$	&	$2.26 \cdot 10^{-8}$	&	$1.51 \cdot 10^{-23}$	\\	
\ref{fig:paterns} (top-right)	&	0.80	&	$\theta_{0}^{\pi/3}(v) $	&	$10^{-4}$	&	$10^{-4}$	&	0	&	$22.0$	&	$150$	&	$2.63 \cdot 10^{-8}$	&	$2.37 \cdot 10^{-26}$	\\	
\ref{fig:paterns} (bot-left)	&	0.80	&	$\theta_{cos}^{(0, \pi/2)}(u) $	&	$10^{-4}$	&	$10^{-4}$	&	0	&	$13.6$	&	$150$	&	$1.98 \cdot 10^{-8}$	&	$8.23 \cdot 10^{-28}$	\\
\ref{fig:paterns} (bot-right)	&	0.80	&	$\theta_{cos}^{(0, \pi/2)}(v) $	&	$10^{-5}$	&	$10^{-5}$	&	0	&	$15.6$	&	$150$	&	$5.22 \cdot 10^{-8}$	&	$1.99 \cdot 10^{-24}$	\\

\hline
% \ref{fig:seashellex6}	&	0.85	&	$\theta_{cos}^{(0, \pi/2)}(v)$	&	$10^{-4}$	&	$10^{-4}$	&	0 &	$24.7$	&	$150$	&	$3.57 \cdot 10^{-8}$	&	$1.04 \cdot 10^{-11}$	
    \end{tabular}}
    \caption{Summary of parameter initializations, weights, residual, and times used for all our examples. The residual accounts only for the 
contributions of $E_{\mathrm{oc}}$, $E_{\mathrm{prox}}$, and $E_{\mathrm{tan}}$. As a quality measure of the L-net, we show the residual of $E_{\mathrm{oc}}$ measured after the extra optimization steps. All computations were performed using a Python implementation of our 
algorithm, running on a MacBook Pro with an Apple M1 chip (8-core CPU, 
8-core GPU, and 16-core Neural Engine). }
    \label{tab:results}
\end{table}

The first direction of the L-conjugate pair is 
initialized via a function $\theta(u,v) \in [0,\pi/2]$, representing 
the angle of the direction against the first 
principal direction. We consider three types of initialization:
\begin{itemize}
    \item[-] \emph{Constant:} $\theta_C(u,v) = c$, for some $c \in [0,\pi/2]$.
    \item[-] \emph{Linear:} along the $u$- and $v$-directions, respectively,
    \begin{align}
        \theta^{\theta_{max}}_{\theta_{min}}(u) &= (1-u)\,\theta_{min} + u\,\theta_{max}, \\
        \theta^{\theta_{max}}_{\theta_{min}}(v) &= (1-v)\,\theta_{min} + v\,\theta_{max}.
    \end{align}
    \item[-] \emph{Cosine:} along the $u$-direction and $v$-directions, respectively
    \begin{align}
        \theta_{\cos}^{(\theta_{min},\ \theta_{max})}(u) = \frac{\theta_{min} + \theta_{max}}{2} 
        + \frac{\theta_{max} - \theta_{min}}{2}\cos(2 \pi\,u), \\
        \theta_{\cos}^{(\theta_{min},\ \theta_{max})}(v) = \frac{\theta_{min} + \theta_{max}}{2} 
        + \frac{\theta_{max} - \theta_{min}}{2}\cos(2 \pi\,v).
    \end{align}
\end{itemize}

\begin{figure}
    \centering
    \begin{overpic}[width=1\linewidth]{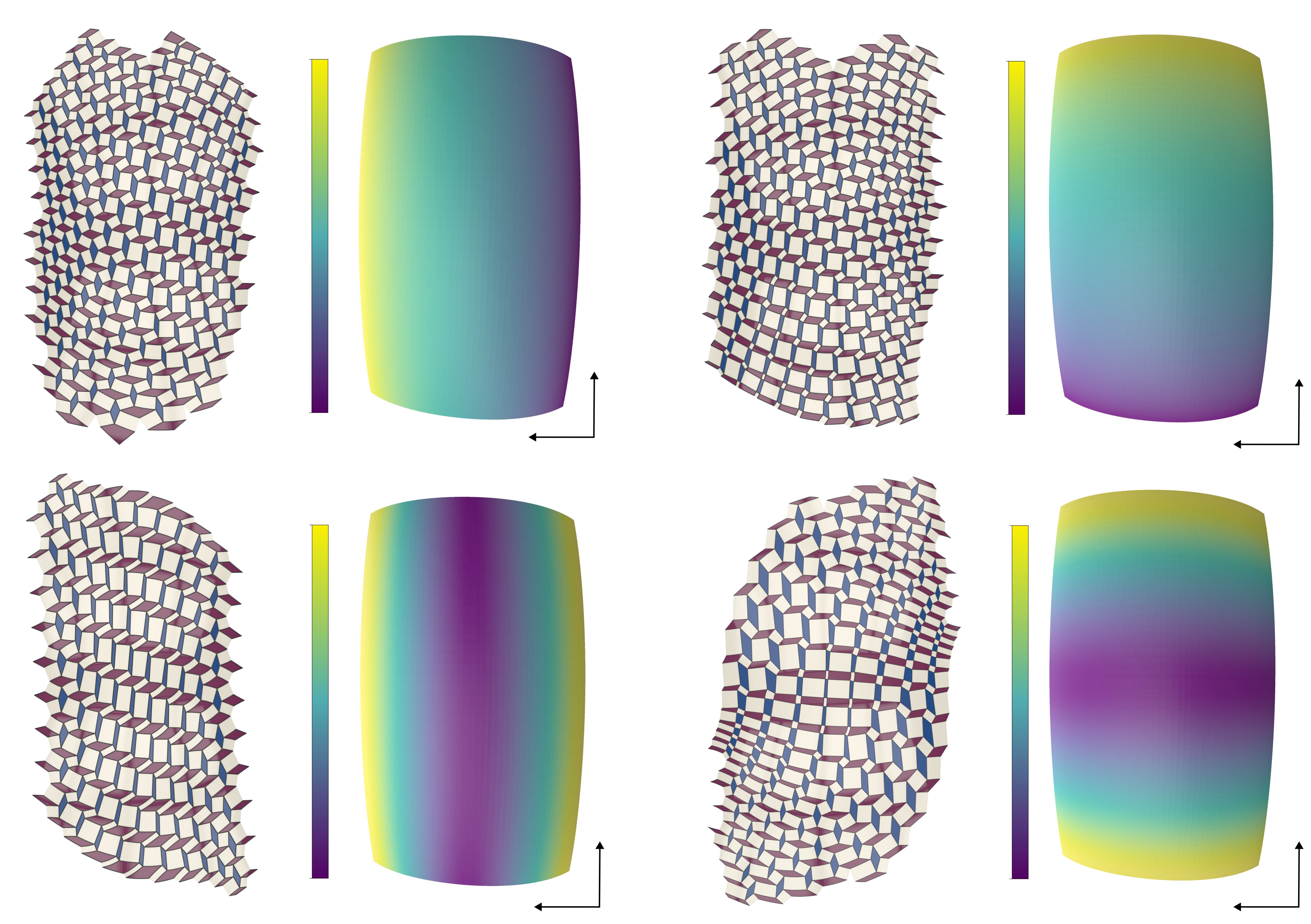}
    \put(38,36.6){$u$}
    \put(45,42.5){$v$}
    \put(92,36.6){$u$}
    \put(99,42.5){$v$}
    \put(38,1.1){$u$}
    \put(45,6.9){$v$}
    \put(92,1.1){$u$}
    \put(99,6.9){$v$}
    \put(21,37.9){$0$}
    \put(20,65){$\frac{5\pi}{12}$}
    \put(74.5,37.9){$0$}
    \put(74.5,65){$\frac{\pi}{3}$}
    \put(21,2){$0$}
    \put(21,30.1){$\frac{\pi}{2}$}
    \put(74.5,2){$0$}
    \put(74.5,30.1){$\frac{\pi}{2}$}
    \end{overpic}
    \caption{L-net patterns obtained from different initialization functions $\theta(u,v)$, 
    which prescribe the angle between the first direction of an L-conjugate pair 
    and the first principal direction at each point $(u,v)$ of the parameter domain. 
    For each pattern, the \emph{right} image shows the scalar field $\theta(u,v)$ 
    visualized on the surface (color scale shown in the color bar). (Top, left to right) $\theta_0^{5\pi/12}(u)$ and $\theta_0^{\pi/3}(v)$; (bottom, left to right) $\theta_{\cos}^{(0,\,\pi/2)}(u)$ and $\theta_{\cos}^{(0,\,\pi/2)}(v)$. }
    \label{fig:paterns}
\end{figure}
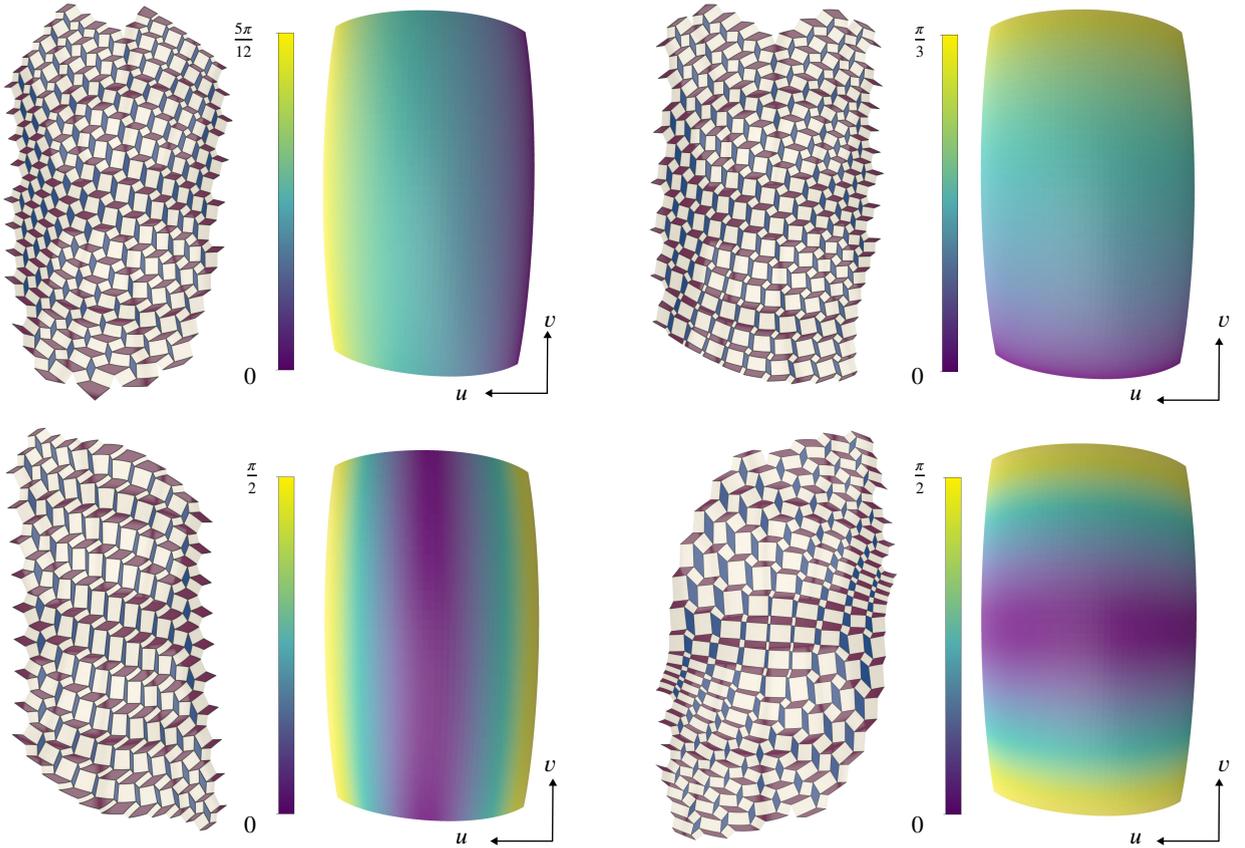

A summary of our results is provided in Table~\ref{tab:results}. Our method 
produces a variety of design patterns for a given input surface, where the user 
can explore different choices of radius and L-conjugate angles (see, for example, Figs.~\ref{fig:nautiex1} and~\ref{fig:paterns}). We observe that minimizing $E_{\mathrm{td}}$ has a fairness effect on the 
sphere center mesh; however, it also tends to shrink the L-net globally, 
as shown in Fig.~\ref{fig:tdEffect}, and the overall improvement in spherical 
panel size is minimal. We also observe that, in general, larger spherical panels are obtained 
when the radius $r$ is chosen close to the smallest principal curvature radius. 
However, the L-conjugate directions also influence the final panel shapes, as 
seen in Fig.~\ref{fig:paterns}, where despite using the same sphere congruence 
relatively close to the smallest principal curvature spheres, the resulting 
spherical panels exhibit considerable size variation. Similarly, Fig.~\ref{fig:spiral_defM} illustrates this effect, where smaller spherical panels appear in the narrow regions of the spiral design and larger ones in the wider regions. This indicates that obtaining larger panels 
likely has to be addressed prior to the choice of the conjugate frame
field and requires further exploration.

\label{sec:results}

\paragraph{Conclusion and future work}
We have introduced the new concept of L-conjugate parameterizations of a
surface with respect to an attached sphere congruence. Based on this theoretical
framework, we devised an algorithm for the approximation of positively curved
surfaces by L-nets. These are smooth surfaces formed by quadrilateral planar,
conical, and spherical patches. 

We are currently working on the M\"obius geometric counterpart, namely
the approximation of a surface $f$ by a quad mesh with spherical faces and
circular edges so that the spheres are close to those of an attached sphere
congruence. This is related to another concept of conjugacy,
which possesses similarly nice properties as L-conjugacy. 

The sphere geometries of  M\"obius and Laguerre are sub-geometries of Lie
sphere geometry, which would be another direction for future work. Special
cases of the corresponding
Lie geometric meshes can be obtained by applying  M\"obius transformations to
L-nets. This removes the restriction to positively curved surfaces. The
resulting surfaces are composed of spherical patches and patches of Dupin
cyclides, arranged in a checkerboard pattern. The cyclides are the images
of the cones under M\"obius transformations. The resulting surfaces have
more degrees of freedom than the cyclidic nets, which are principal curvature parameterizations that consist of smoothly joined %Dupin 
cyclide patches~\citep{bo+2011,bobenko+2012}.

Another interesting direction is the application of analogs (of Minkowski lifts) of our L-nets in 4-dimensional \emph{simply isotropic space} with the metric of signature $(+,+,+,0)$. 
%In the smooth case, 
These are lifts of ordinary conjugate nets in $\mathbb{R}^3$ to 4-dimensional space. 
We can then use them to initialize optimization towards L-nets in $\mathbb{R}^{3,1}$ by gradually changing the metric. Such a general approach has been proved successful for a class of \emph{Euclidean} problems, for instance, the construction of asymptotic nets with a constant node angle, % and flexible conjugate nets, 
as demonstrated in a recent series of papers \citep{%pirahmad2024area,pirahmad2025,
SkopenkovYorov2025CRPC,Yorov-Pottmann-Skopenkov-23,yorov2025}, where the last one serves as an introduction to this approach and summarizes the achievements. Having an application to \emph{Minkowskian} geometry is of additional interest.

\bibliographystyle{cas-model2-names}
%\bibliographystyle{cas-model1-num}
%\bibliography{ref,sphere-mesh}
\bibliography{referencesMerged}

\printcredits

\section*{Funding}
This research has been supported by King Abdullah University of Science and Technology
baseline funding.

\section*{Data availability}
Data sharing is not applicable to this article, as no datasets were generated or analyzed
during the current study.

\section*{Declaration of competing interest}
The authors declare that they have no known competing financial interests or personal relationships that could have appeared to influence the work reported in this paper

\end{document}